\documentclass[a4paper,11pt]{article}
\usepackage{bm,mathrsfs,amsmath,amsthm,amssymb,amscd,longtable,array,
graphicx,color}
\newcommand{\nc}{\newcommand}
\nc{\rnc}{\renewcommand}
\nc{\nn}{\nonumber}
\nc{\der}{{\partial}}
\rnc{\Im}{{\rm{Im}\,}}
\rnc{\Re}{{\rm{Re}\,}}
\nc{\db}{\displaybreak[0]\\}
\nc{\bra}{\langle}
\nc{\ket}{\rangle}
\nc{\bs}{\boldsymbol}

\newtheorem{theorem}{Theorem}[section]
\newtheorem{lemma}[theorem]{Lemma}
\newtheorem{conjecture}[theorem]{Conjecture}
\newtheorem{proposition}[theorem]{Proposition}
\newtheorem{corollary}[theorem]{Corollary}

\theoremstyle{definition}

\numberwithin{equation}{section}

\numberwithin{equation}{section}

\begin{document}%
%
\title{Tetrahedral $L$-operators, tensor Schur polynomials and $q$-deformed loop elementary symmetric functions}

\author{
Shinsuke Iwao, Kohei Motegi and Ryo Ohkawa
}

\date{\today}

\maketitle

\begin{abstract}
We study three-dimensional partition functions constructed from the tetrahedral
$L$-operator introduced and studied by Bazhanov-Sergeev and Kuniba-Maruyama-Okado.
First, we explore the $q=0$ case, extending {\color{black}the authors'} previous results and giving applications
by a further analysis on the Zamolodchikov-Faddeev algebra.
We introduce a class of partition functions which can be expressed as
the tensor Schur polynomials, a class of  products of Schur polynomials.
As an application, we derive the shuffle formula for the Schur polynomials which is geometrically
the pushforward formula by Jo\'zefiak-Pragacz-Lascoux.
We also give a derivation and a unification of the
Gustafson-Milne and Feh{\'e}r--N{\'e}methi--Rim{\'a}nyi identities,
and introduce
a family of Laurent polynomials using divided difference operators which 
imitates the Schubert polynomials from the perspective of our study.
We also present an application to the steady state of the multispecies totally asymmetric simple exclusion process.
Second, we investigate several classes
of partition functions for the generic $q$ case, and determine the explicit forms
as deformations of the elementary symmetric functions.
One of them can be regarded as (an extension of) a $q$-deformed loop elementary symmetric functions.
\end{abstract}

\section{Introduction}
The tetrahedron equation introduced by Zamolodchikov
\cite{Zamone,Zamtwo} is a three-dimensional generalization
of the Yang-Baxter equation
\cite{Yang,BaxterYBE}, which plays the key role of integrability
in three-dimensional statistical physics models.
The tetrahedron equation has been the subject of research on its solutions ever since its introduction
\cite{Baxterone,BaSt,BB,Kor,Kashaev,SMS}, and particularly in recent years.
The connections with quantum groups \cite{FST,FRT,Drinfeld,Jimbo}, as well as those between the $R$-matrices in two and three dimensions, have been gradually clarified over the years \cite{BaSe,BMS,MBS,KOS,KMY,Kuniba}.
There are also developments on constructing solutions \cite{GSZ,SY,IKT1,IKT2} using the cluster algbras \cite{FZ,FG},
and another approach using fermion algebras
\cite{PK,PKtwo}.  In addition, there are more recent developments \cite{BKMS,SSHY}.

Partition functions constructed from three-dimensional $L$-operators,
which are solutions to the tetrahedron equation, are much less studied than
their two-dimensional counterparts.
By investigating the $q=0$ version of a solution by
Bazhanov-Sergeev \cite{BaSe} and constructing a realization of the
Zamolodchikov-Faddeev algebra, Kuniba-Maruyama-Okado \cite{KMO1,KMO2}
identified a class of three-dimensional partition functions with
the (relative) steady state probabilities of a multispecies generalization of
a totally asymmetric symmetric exclusion process (TASEP) \cite{Spitzer} under the periodic boundary condition.
This was done by transforming a combinatorial construction of the
matrix product steady state of the multispecies TASEP by Ferrari-Martin \cite{FM}
in a nontrivial way using the crystal basis \cite{Kashiwara} description of the combinatorial $R$-matrix \cite{NY}.
There are studies on the steady states of (T)ASEP from various perspectives.
See foundational works \cite{DEHP,Sa} for open boundaries,
and \cite{AAMP,CDGW,Martin,CMW,KW,KOSc,ANP} for multispecies (T)ASEP,
including various types and settings, for example.

In a related but different setting from \cite{KMO2}, we introduced a class of partition functions
which can be expressed by the Schur polynomials.
This was done by taking the vacuum expectation values in the bosonic Fock spaces
instead of taking the traces. Another crucial point is that, for realizing Schur polynomials associated
with general Young diagrams, we do not need to use all types of the $X$-operators
introduced in \cite{KMO2}.

In this paper, we further explore the algebraic aspects of the partition functions and also give applications.
This paper largely consists of two parts.
The first part is a more detailed study for the $q=0$ case.
By a further analysis on the Zamolodchikov-Faddeev algebra, we  \\
\\
$\cdot$ introduce a family of partition functions whose
explicit forms are given by the tensor Schur polynomials,
generalizing the main result in \cite{IMO}; \\
$\cdot$ give a derivation of the
shuffle formula for the Schur polynomials, which corresponds to
the Jo\'zefiak-Pragacz-Lascoux \cite{JLP} formula in geometry; \\
$\cdot$ give a derivation and a unification of the 
Gustafson-Milne \cite{GM} and Feh{\'e}r--N{\'e}methi--Rim{\'a}nyi
 \cite{FNR} identities for 
the Schur polynomials, \\
$\cdot$ introduce an analog of Schubert polynomials
which represent a certain class of partition functions; \\
$\cdot$ give an application to the relative steady state
probabilities of the mutlispecies TASEP. \\

By a tensor Schur polynomial we simply mean a product of Schur polynomials.
The terminology ``tensor Schur functions (polynomials, basis)'' appears, for example,
in the study of wreath Macdonald polynomials in \cite{OSW}.

In the second part, we use the original tetrahedral $L$-operator by Bazhanov-Sergeev \cite{BaSe}
which includes the (quantum group) parameter $q$.
Since the parameter $q$ is generic,
the tetrahedral $L$-operator can be viewed as an operator-valued six-vertex model
which upgrades the operator-valued five-vertex model studied in the first part.
We introduce several classes of partition functions and
we determine their explicit forms as deformations of the elementary symmetric functions.
One class involves the $q$-exponential functions,
and the other class can be regarded as a $q$-deformation of the loop elementary symmetric functions
and its further extension.

This paper is organized as follows. 
In section 2, we introduce the $X$-operators and 
study partition functions constructed from the $X$-operators using the $q=0$ tetrahedral $L$-operator as local pieces.
We
relate with the tensor Schur polynomials, modified Schubert polynomials
and give applications to the Schur polynomial identities and steady state of multispecies TASEP.
In section 3, we introduce several classes of partition functions constructed from the $q$ generic tetrahedral $L$-operator,
and determine their explicit forms.

Here we collect the notations and conventions that will be used frequently in this paper.
We denote sets of variables using bold symbols, such as 
$\mathbf{w}_i$ and $\mathbf{z}_j$.
For integers $i$ and $j$ with $i \le j$, we write $[i,j]:=\{i,i+1,\dots,j\}$, and for any 
such family of sets $\mathbf{x}_i$ (e.g., $\mathbf{w}_i$ or $\mathbf{z}_i$), we define
\[
\mathbf{x}_{[i,j]}
  = \mathbf{x}_i \cup \mathbf{x}_{i+1} \cup \dots \cup \mathbf{x}_j .
\]
For integers $i$ and $j$ with $i>j$, we set $[i,j]:=\varnothing$.

Let us introduce several shorthand notations.
For a set of variables ${\bf z}=\{z_1,\dots,z_{|{\bf z}|}\}$, 
where $|{\bf z}|$ denotes its cardinality, we write
\[
X_i({\bf z}) := \prod_{j=1}^{|{\bf z}|} X_i(z_j),
\]
which is well defined for commuting operators $X_i(z)$.
For any integer $p$, we define
\[
{\bf z}^p := \prod_{j=1}^{|{\bf z}|} z_j^p.
\]
For two sets of variables ${\bf z}$ and ${\bf w}$, we further introduce the shorthand notations
\[
{\bf z}-{\bf w} := 
\prod_{j=1}^{|{\bf z}|}\prod_{k=1}^{|{\bf w}|}(z_j - w_k),
\qquad
1 - {\bf z}/{\bf w} :=
\prod_{j=1}^{|{\bf z}|}\prod_{k=1}^{|{\bf w}|}
\left(1 - \frac{z_j}{w_k}\right).
\]

We also use summation symbols of the form 
$\displaystyle  \sum_{({\bf w}_1,\dots,{\bf w}_m)}$, 
which denote sums over all $m$-tuples 
$({\bf w}_1,\dots,{\bf w}_m)$ 
such that each ${\bf w}_i$ $(i=1,\dots,m)$ is an unordered set of variables 
with $|{\bf w}_i| = |{\bf z}_i|$, and
\[
{\bf w}_1 \cup \cdots \cup {\bf w}_m 
   = {\bf z}_1 \cup \cdots \cup {\bf z}_m .
\]

A \emph{partition} is a weakly decreasing sequence of nonnegative integers
\(
  \lambda = (\lambda_1,\lambda_2,\dots,\lambda_n)
\)
with $\lambda_1 \ge \lambda_2 \ge \cdots \ge \lambda_n \ge 0$.
The Schur polynomial associated with a partition $\lambda$
and a set of symmetric variables \(
  \mathbf{z} = (z_1,z_2,\dots,z_n)
\)
 is defined by
\begin{align}
  s_\lambda(\mathbf{z})
  = \frac{\det ( z_i^{\,\lambda_j + n - j} )_{1\le i,j\le n}}
         {\det( z_i^{\,n-j} )_{1\le i,j\le n}}. \label{Schurdef}
\end{align}

For integers $i$ and $j\ge 0$, 
the notation $i^j$ denotes the string consisting of $j$ consecutive copies of $i$.
For example, $3^4$ represents $(3,3,3,3)$.

\section{$q=0$ three-dimensional $L$-operator}
In this section, we study three-dimensional partition functions based on 
the $q=0$ case of the tetrahedral $L$-operator introduced in \cite{BaSe} 
and further developed in \cite{KMO1,KMO2}. 
This section extends the results of \cite{IMO} and further investigates
algebraic identities and applications to the multi-species TASEP
through the multiple commutation relations of the
Zamolodchikov--Faddeev algebra.

\subsection{$q=0$ three-dimensional $L$-operator and $X$-operators}
We first introduce the $q=0$ three-dimensional $L$-operator.

Let $V=\mathbb{C}v_0 \oplus \mathbb{C}v_1$ be the two-dimensional complex vector space 
with standard basis $\{v_0,v_1\}$, and let 
$\mathcal{F}=\bigoplus_{m=0}^\infty \mathbb{C}|m\rangle$ be the bosonic Fock space 
with basis $\{|m\rangle\}_{m\ge0}$.
Its dual spaces are 
$V^\ast=\mathbb{C}v_0^\ast\oplus \mathbb{C}v_1^\ast$, where $v_j^\ast(v_i)=\delta_{ij}$,
and 
$\mathcal{F}^\ast=\bigoplus_{m=0}^\infty \mathbb{C}\langle m|$, 
where $\langle m'|m\rangle=\delta_{m',m}$ is the standard pairing.

We define linear operators $\mathbf{t},\mathbf{b}^+,\mathbf{b}^-$ on $\mathcal{F}$ by
\[
\mathbf{t}|m\rangle=\delta_{m,0}|m\rangle,\qquad
\mathbf{b}^+|m\rangle=|m+1\rangle,\qquad
\mathbf{b}^-|m\rangle=|m-1\rangle,
\]
where $|-1\rangle:=0$.
They satisfy
\[
\mathbf{t}\mathbf{b}^+=\mathbf{b}^- \mathbf{t}=0,\qquad
\mathbf{b}^+\mathbf{b}^-=1-\mathbf{t},\qquad
\mathbf{b}^-\mathbf{b}^+=1,
\]
and are referred to as the vacuum projector, creation operator, and annihilation operator, 
respectively.
This algebra is the $q\to0$ limit of the $q$-oscillator algebra.
See Figure \ref{bosonspaceactionqzero} for a graphical description of the bosonic Fock space.

\begin{figure}[htbp]
\centering
\includegraphics[width=8truecm]{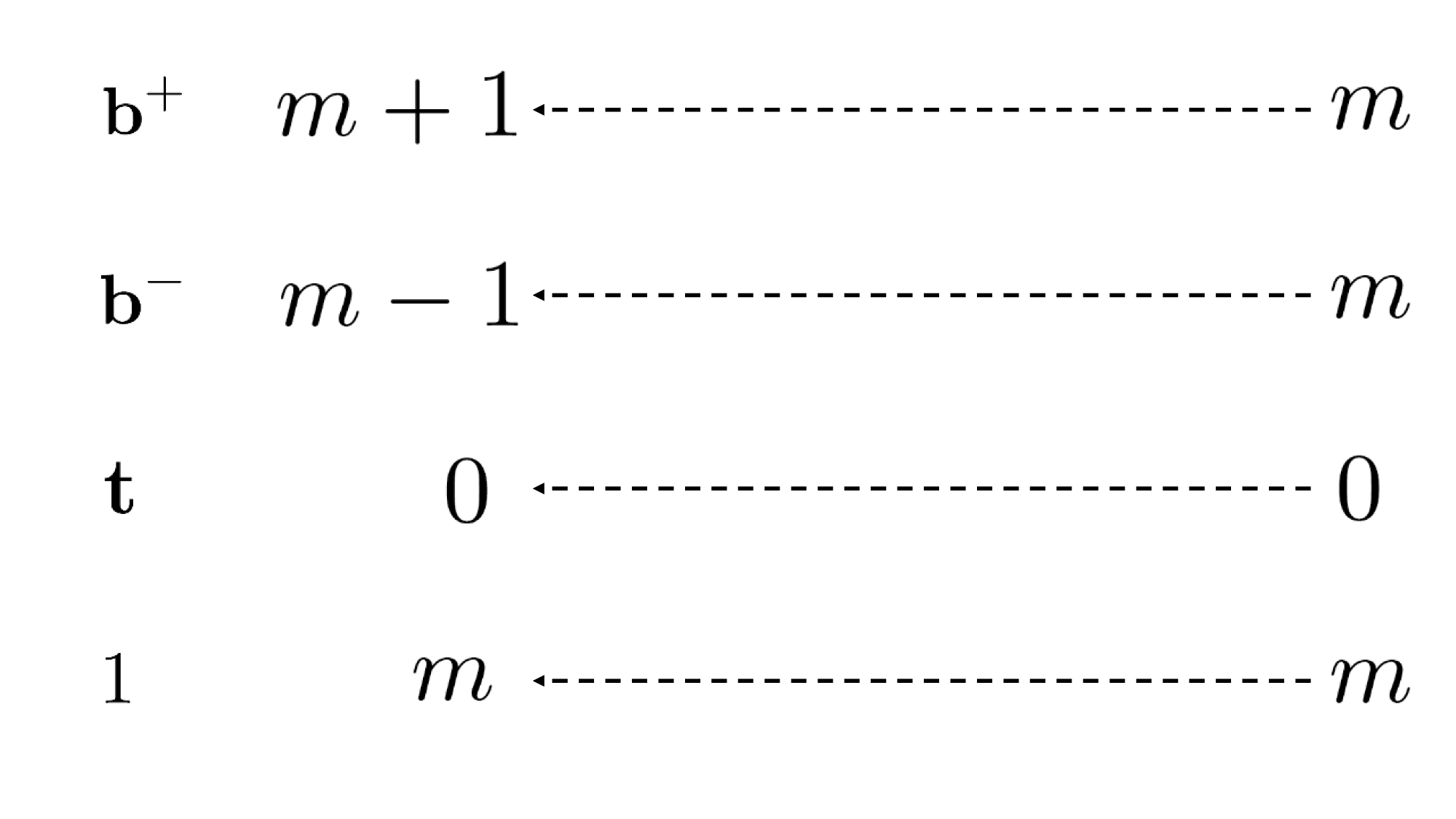}
\caption{A bosonic Fock space is depicted as a dashed line.  
The actions of the creation operator $\mathbf{b}^+$, the annihilation operator $\mathbf{b}^-$, 
the vacuum projection operator $\mathbf{t}$, and the identity $1$ on a state $|m\rangle$ 
are shown graphically.
}
\label{bosonspaceactionqzero}
\end{figure}

We next introduce the operator-valued two-dimensional $L$-operator 
$\mathcal{R}:V\otimes V\to V\otimes V$ defined by
\[
\mathcal{R}(v_i\otimes v_j)
=\sum_{a,b=0}^1 v_a\otimes v_b\, [\mathcal{R}]_{ij}^{ab},
\]
with nonzero matrix elements
\[
[\mathcal{R}]_{00}^{00}=1,\quad
[\mathcal{R}]_{11}^{11}=1,\quad
[\mathcal{R}]_{10}^{01}=\mathbf{b}^+,\quad
[\mathcal{R}]_{01}^{10}=\mathbf{b}^-,\quad
[\mathcal{R}]_{01}^{01}=\mathbf{t}.
\]
All other components vanish.
See Figure
\ref{qzerotetraloperator} for the graphical description of $\mathcal{R}$
which we use in this paper. 
In the figure, we color the half-edges which are connected to 1s in the two-dimensional vector spaces with
red, and color 0s with
blue.

\begin{figure}[htbp]
\centering
\includegraphics[width=12truecm]{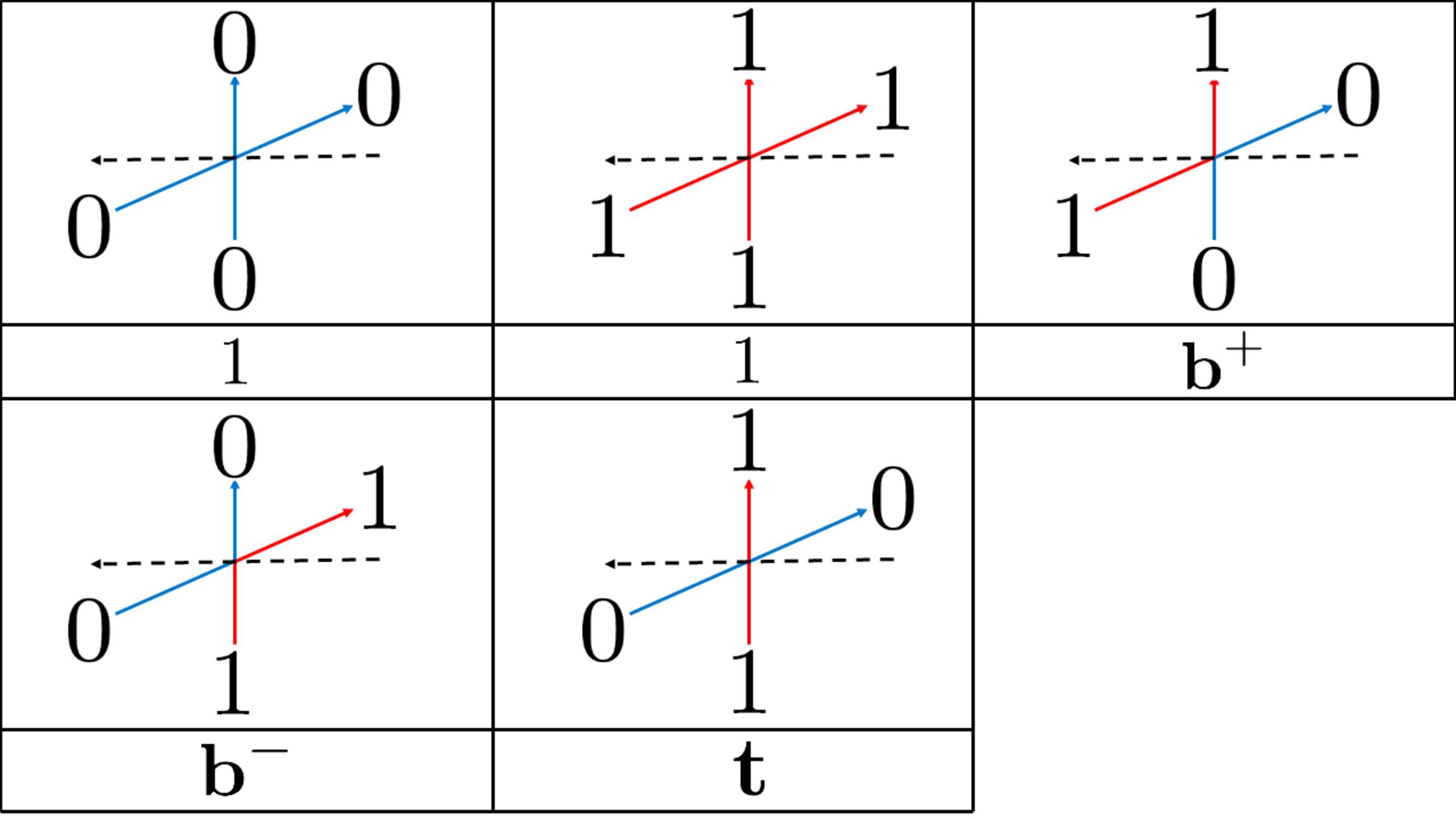}
\caption{
The non-zero operator-valued matrix elements of $\mathcal{R}$.
The two solid lines represent two two-dimensional spaces,
and the dashed line represents the bosonic Fock space.
}
\label{qzerotetraloperator}
\end{figure}

Let 
\[
\mathbb{Z}^3 = \mathbb{Z}\mathbf{e}_1 \oplus \mathbb{Z}\mathbf{e}_2 \oplus \mathbb{Z}\mathbf{e}_3
\] 
be the three-dimensional lattice with standard basis $\{\mathbf{e}_1,\mathbf{e}_2,\mathbf{e}_3\}$, 
and identify 
\[
\alpha\mathbf{e}_1+\beta\mathbf{e}_2+\gamma\mathbf{e}_3 \longleftrightarrow (\alpha,\beta,\gamma).
\]
We call the planes
\[
H^1_k = \{\alpha=k\},\qquad
H^2_\ell = \{\beta=\ell\},\qquad
H^3_m = \{\gamma=m\}
\]
the $k$-th row, the $\ell$-th column, and the $m$-th slice, respectively.

Let  
\[
D_n := \{(k,\ell) \in \mathbb{Z}^2 \mid k\ge 1,\ \ell\ge 1,\ k+\ell \le n\}
\]
be the triangular index set of cardinality $\frac{n(n-1)}{2}$.  
To each line 
\[
L_{k\ell} := H^1_k \cap H^2_\ell, \qquad (k,\ell)\in D_n,
\]
we assign a copy of the Fock space $\mathcal{F}$, denoted by $\mathcal{F}_{k\ell}$.

Each $\mathcal{F}_{k\ell}$ has basis 
\[
\{|m_{k\ell}\rangle_{k\ell} \;;\; m_{k\ell}=0,1,2,\dots\},
\] 
so the standard bases of  
\[
\mathcal{F}^{\otimes n(n-1)/2} \quad \text{and} \quad (\mathcal{F}^{\otimes n(n-1)/2})^*
\] 
are given by
\[
\Bigl\{ \bigotimes_{(k,\ell)\in D_n} |m_{k\ell}\rangle_{k\ell} \Bigr\},
\qquad
\Bigl\{ \bigotimes_{(k,\ell)\in D_n} {}_{k\ell}\langle m_{k\ell}| \Bigr\}.
\]
We write
\[
|\Omega\rangle := |0\rangle^{\otimes n(n-1)/2}, \qquad
\langle\Omega| := (|\Omega\rangle)^{*}
\]
for the vacuum and its dual.

We now introduce the linear operator $X_i^{(n)}(z)$ $(i=0,\dots,n)$ \cite{Kuniba,KMO1,KMO2} acting on 
\(
\mathcal{F}^{\otimes n(n-1)/2}\otimes \mathbb{C}[z].
\)
Its action is described graphically in Figure~\ref{XIoperatorfigureqzero}.  
Each of the $\frac{n(n-1)}{2}$ intersections carries a matrix element of the $L$-operator $\mathcal{R}$, chosen according to the local configuration of red and blue segments.

The operator $X_i^{(n)}(z)$ is defined as the weighted sum over all configurations of $\mathcal{R}^{\otimes n(n-1)/2}$ satisfying:  
(i) the bottom edges of columns $1,\dots,i$ are colored red;  
(ii) those of columns $i+1,\dots,n$ are colored blue.  
When no ambiguity arises, we write $X_i(z)$.

\begin{figure}[htbp]
\centering
\includegraphics[width=12truecm]{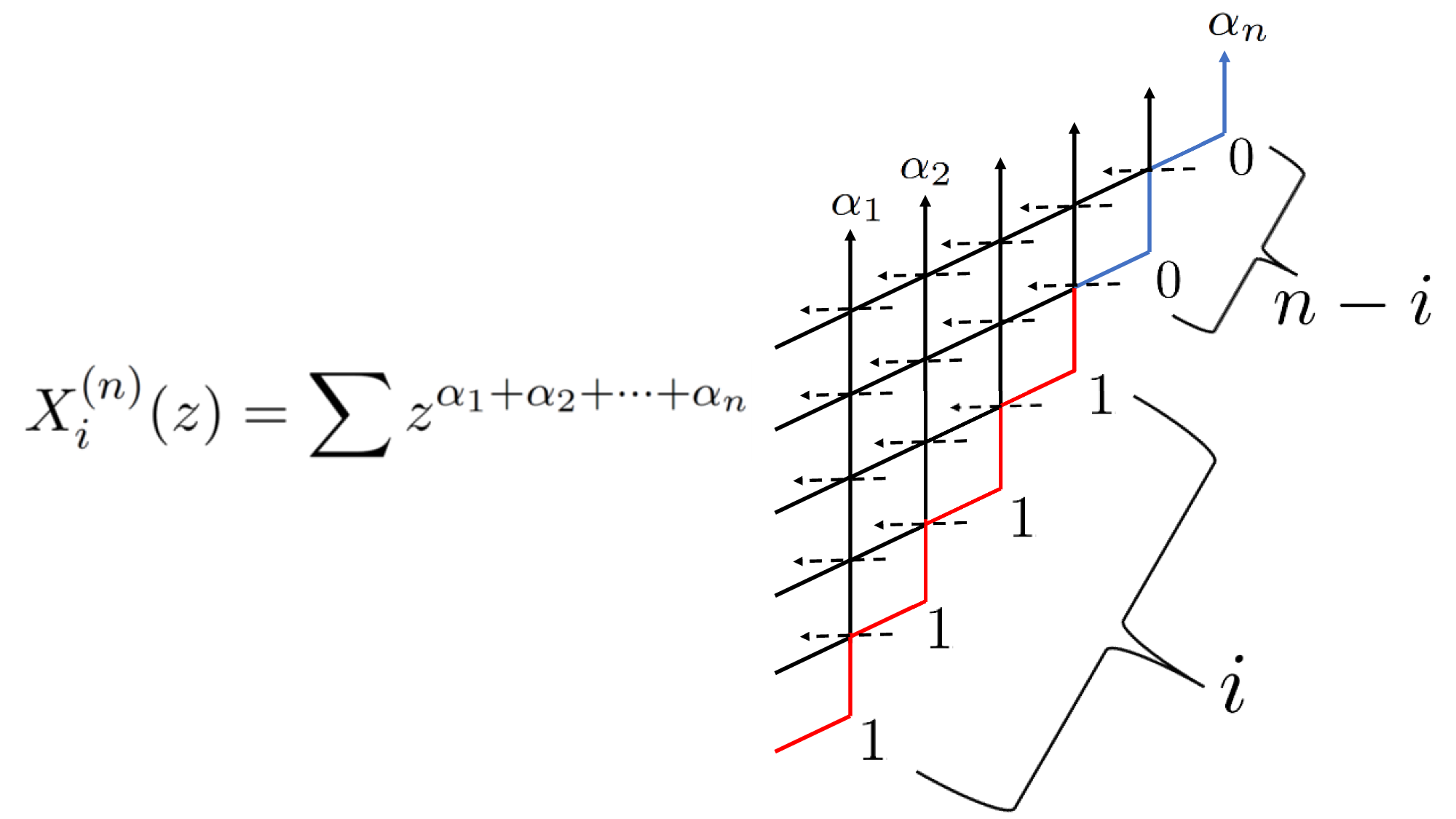}
\caption{The operator $X_i^{(n)}(z)$.  
We sum over all configurations except along one fixed boundary.  
The subscript $i$ indicates this fixed boundary condition, specifying that the first $i$ consecutive sites on that boundary are occupied by $1$s.  
The sum is weighted, and the power $\alpha_1 + \alpha_2 + \cdots + \alpha_n$ of $z$ counts the number of $1$s along the top boundary.  
Here, for $i=1,2,\dots,n$, 
each $\alpha_i \in \{0,1\}$ labels the state of the $i$-th site on the top boundary.
}
\label{XIoperatorfigureqzero}
\end{figure}

We also introduce $X_{i,j}^{(n)}$
as Figure~\ref{XIoperatorfigureqzerodecomp}, which we restrict the sum of $\alpha_1 + \alpha_2 + \cdots + \alpha_n$
to be $j$ and set $z=1$.
Note the relation $X_i^{(n)}=\sum_{j} z^j X_{i,j}^{(n)}$ which follows by definition.

\begin{figure}[htbp]
\centering
\includegraphics[width=8truecm]{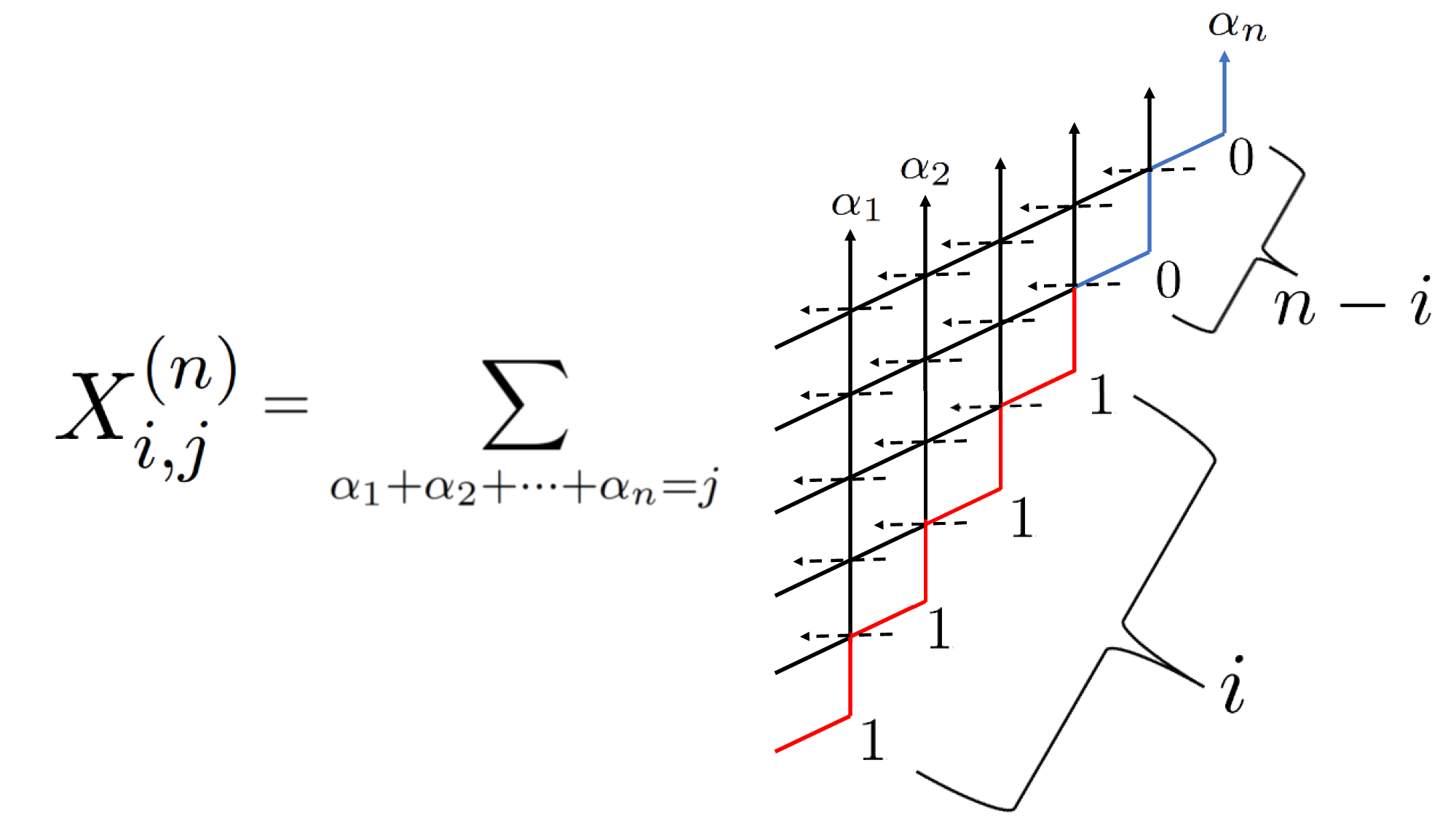}
\caption{The operator $X_{i,j}^{(n)}$. The sum of $\alpha_1 + \alpha_2 + \cdots + \alpha_n$
is restricted to $j$ and there is no $z$-dependence.
}
\label{XIoperatorfigureqzerodecomp}
\end{figure}

The operators $X_i(z)$
are shown to satisfy the following relations.
\begin{theorem} \cite{KMO2,Kuniba} \label{ThmFZ}
The operators $X_i(z)$ satisfy the  Zamolodchikov-Faddeev algebra relations
\begin{align}
X_i(x)X_j(y)=\left\{
\begin{array}{ll}
X_i(y)X_j(x)+(1-x/y)X_j(y)X_i(x) & i<j,  \\
X_i(y)X_i(x) & i=j,  \\
x/y X_i(y) X_j(x) & i>j, 
\end{array}
\right.
.
\end{align} \label{FZalg}
\end{theorem}
Theorem~\ref{ThmFZ} was derived in \cite{KMO2,Kuniba}
as a consequence of the bilinear identities for layer transfer matrices,
which in turn follow from the tetrahedron equation, a consistency condition
for the three-dimensional $L$-operator.
For further details, including an explanation of the tetrahedron equation,
see \cite{BaSe,KMO2,Kuniba}.

\subsection{Partition functions, tensor Schur polynomials and shuffle formula}

We introduce the following class of three-dimensional partition functions
\begin{align}\label{eq:partition_function}
\langle \Omega |X_{i_1}^{(n)}(z_{1}) X_{i_{2}}^{(n)} (z_{2}) \cdots
X_{i_{m-1}}^{(n)}(z_{m-1}) X_{i_m}^{(n)}(z_{m}) | \Omega \rangle,
\end{align}
which is graphically represented in Figure \ref{3Dpartitionfunctionfigure}.
\begin{figure}[htbp]
\centering
\includegraphics[width=12truecm]{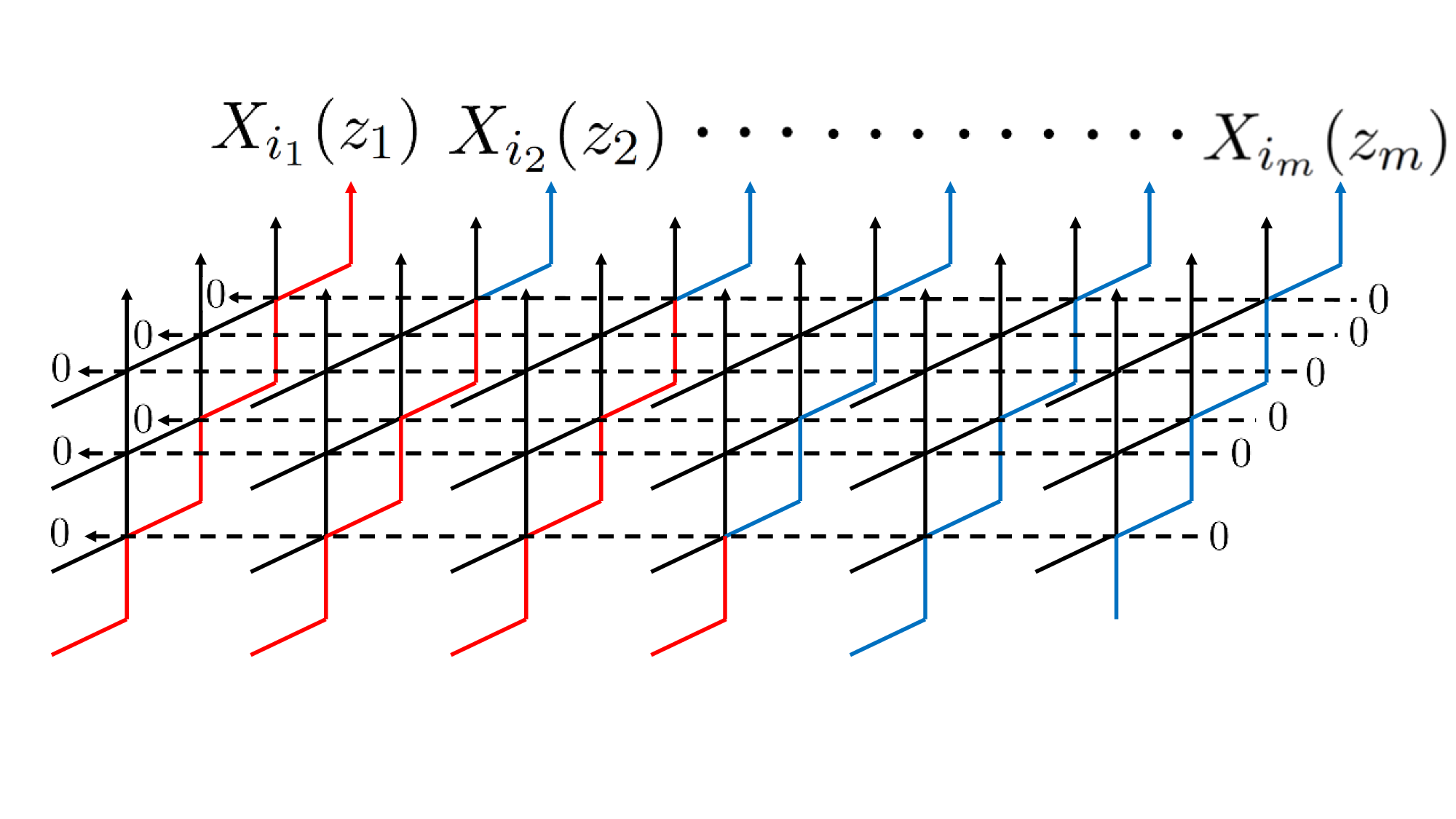}
\caption{The  partition functions
$\langle \Omega| X_{i_1}(z_{1}) X_{i_{2}} (z_{2}) \cdots
X_{i_{m-1}}(z_{m-1}) X_{i_m}(z_{m})|\Omega \rangle$.
}
\label{3Dpartitionfunctionfigure}
\end{figure}

In the previous paper \cite{IMO}, {\color{black}the authors} investigated the case
$n \geq i_1 \geq i_2 \geq \cdots \geq i_m \geq 0$ and established a
correspondence with the Schur polynomials.
We first study a more general class and investigate
its algebraic structures.

First, let us recall some previous results.
We also correct minor mistakes in the earlier work. 
Although it does not affect the main results there, 
this correction plays an important role in the present discussion.

\begin{proposition} \label{propmultipleformulaone}
For $n \geq i_1 > i_{2} > \cdots > i_m \geq 0$,
we have the following commutation relations:
\begin{align}
&\frac{X_{i_1}({\bf z}_1) X_{{i_{2}}}({\bf z}_{2}) \cdots X_{i_{m-1}}({\bf z}_{m-1}) X_{i_m}({\bf z}_m)}
{\prod_{k=1}^m {\bf z}_k^{m-k}}
\nonumber \\
=&\sum_{({\bf w}_1,{\bf w}_2,\dots,{\bf w}_m)}
\frac{1}{ \prod_{k=1}^m {\bf w}_k^{m-k}  \prod_{1 \le j<k \le m}(1-{\bf w}_k/{\bf w}_j )} 
 X_{i_m}({\bf w}_m) X_{i_{m-1}}({\bf w}_{m-1}) \cdots X_{i_2}({\bf w}_{2}) X_{i_1}({\bf w}_1).
\label{multiplecommrel}
\end{align}
Here, we take the sum over all 
$({\bf w}_1,{\bf w}_2,\dots,{\bf w}_m) $
such that
${\bf w}_i \ (i=1,\dots,m)$ are unordered sets of variables satisfying $|{\bf w}_i|=|{\bf z}_i|$
and 
${\bf w}_1  \cup {\bf w}_2  \cup \cdots \cup {\bf w}_m={\bf z}_1 \cup {\bf z}_2  \cup \cdots \cup {\bf z}_m$
.
\end{proposition}
$\mathbf{Note}$:
We use the form
$\frac{1}{ \prod_{k=1}^m {\bf w}_k^{m-k}  \prod_{1 \le j<k \le m}(1-{\bf w}_k/{\bf w}_j )} $,
 in the middle of
\cite[(3.12)]{IMO},
for the coefficients in the right-hand side of \eqref{multiplecommrel}.
Rewriting as $  \prod_{1 \le j<k \le m} \frac{1}{ {\bf w}_j -{\bf w}_k}$
is correct only when $|{\bf z}_j |=1$ for all $j$.

The case $t=0$ of \cite[Prop. 4]{Pragacz} (degeneration of the symmetrization
formula from the Hall-Littlewood polynomials to the Schur polynomials) is the following.
\begin{proposition} \cite{Pragacz} \label{Symmformula}
For ${\bf z}=({\bf z}_1,{\bf z}_2,\dots,{\bf z}_m)$, we have
\begin{align}
s_{(\lambda_1^{|{\bf z}_1|},\lambda_2^{|{\bf z}_2|},\dots,\lambda_m^{|{\bf z}_m|})}({\bf z})=\sum_{({\bf w}_1,{\bf w}_2,\dots,{\bf w}_m)}
\frac{1}{\prod_{1 \le j<k \le m}(1-{\bf w}_k/{\bf w}_j )}
 {\bf w}_1^{\lambda_1} {\bf w}_2^{\lambda_2}
\cdots {\bf w}_m^{\lambda_m}.
\label{symmformulaSchur}
\end{align}
Here, we take the sum over all 
$({\bf w}_1,{\bf w}_2,\dots,{\bf w}_m) $
such that
${\bf w}_i \ (i=1,\dots,m)$ are unordered sets of variables satisfying $|{\bf w}_i|=|{\bf z}_i|$
and 
${\bf w}_1  \cup {\bf w}_2  \cup \cdots \cup {\bf w}_m={\bf z}_1 \cup {\bf z}_2  \cup \cdots \cup {\bf z}_m$
.
\end{proposition}
$\mathbf{Note}$:
The correct form of the coefficients
in the right hand side of \eqref{symmformulaSchur} 
is $\frac{1}{\prod_{1 \le j<k \le m}(1-{\bf w}_k/{\bf w}_j )}$
as is written in \cite[Prop. 4]{Pragacz}, which we correct from \cite[(3.5)]{IMO}.
Rewriting
$\frac{1}{\prod_{1 \le j<k \le m}(1-{\bf w}_k/{\bf w}_j )}$
 as
$
\frac{1}{\prod_{1 \le j<k \le m}({\bf w}_j-{\bf w}_k) } {\bf w}_1^{m-1} {\bf w}_2^{m-2}
\cdots {\bf w}_m^{0}$
is correct only when $|{\bf z}_j |=1$ for all $j$.

These two corrections do not affect the main theorem from {\color{black}the previous} paper \cite[Thm. 3.6]{IMO}.

\begin{theorem} \label{mainthm} \cite{IMO}
For $n \geq i_1 > i_{2} > \cdots > i_m \geq 0$, we have
\begin{align}
&\langle \Omega |
X_{i_1}({\bf z}_1) X_{{i_{2}}}({\bf z}_{2}) \cdots X_{i_{m-1}}({\bf z}_{m-1}) X_{i_m}({\bf z}_m)
| \Omega \rangle
\nn \\
=&\prod_{k=1}^m {\bf z}_k^{m-k}
s_{((i_1-m+1)^{|{\bf z}_1|},(i_2-m+2)^{|{\bf z}_2|},\dots,i_m^{|{\bf z}_m|})}({\bf z}).
 \label{schurtheorem}
\end{align}
\end{theorem}
\eqref{schurtheorem} follows from \eqref{multiplecommrel},
\eqref{symmformulaSchur} and the following factorization property
\cite[Lem. 3.2]{IMO}.
\begin{lemma} \label{simplestpartition} \cite{IMO}
When $n \geq i_m \geq i_{m-1} \geq \dots \geq i_2 \geq i_1 \geq 0$,
we have
\begin{align}
\langle \Omega| X_{i_1}(z_{1}) X_{i_{2}} (z_{2}) \cdots
X_{i_{m-1}}(z_{m-1}) X_{i_m}(z_{m})|\Omega \rangle
=z_1^{i_1} \cdots z_{m}^{i_m}.
\label{top}
\end{align}
\end{lemma}

The following computation for $n \geq i_1 > i_{2} > \cdots > i_m \geq 0$
\begin{align}
&\langle \Omega| X_{i_1}({\bf z}_1) X_{{i_{2}}}({\bf z}_{2}) \cdots X_{i_{m-1}}({\bf z}_{m-1}) X_{i_m}({\bf z}_m) |\Omega \rangle
\nonumber \\
=&\prod_{k=1}^m {\bf z}_k^{m-k}
\sum_{({\bf w}_1,{\bf w}_2,\dots,{\bf w}_m)}
\frac{1}{ \prod_{k=1}^m {\bf w}_k^{m-k}  \prod_{1 \le j<k \le m}(1-{\bf w}_k/{\bf w}_j )} 
 \nn \\
&\times
\langle \Omega| X_{i_m}({\bf w}_m) X_{i_{m-1}}({\bf w}_{m-1}) \cdots X_{i_2}({\bf w}_{2}) X_{i_1}({\bf w}_1) | \Omega \rangle
\nonumber \\
=&\prod_{k=1}^m {\bf z}_k^{m-k}
\sum_{({\bf w}_1,{\bf w}_2,\dots,{\bf w}_m)}
\frac{1}{ \prod_{k=1}^m {\bf w}_k^{m-k}  \prod_{1 \le j<k \le m}(1-{\bf w}_k/{\bf w}_j )} 
{\bf w}_1^{i_1}
{\bf w}_2^{i_2}
\cdots {\bf w}_{m-1}^{i_{m-1}} {\bf w}_m^{i_m} \nn \\
=&
\prod_{k=1}^m {\bf z}_k^{m-k}
s_{((i_1-m+1)^{|{\bf z}_1|},(i_2-m+2)^{|{\bf z}_2|},\dots,i_m^{|{\bf z}_m|})}({\bf z}),
\end{align}
using
\eqref{multiplecommrel},
\eqref{symmformulaSchur} and \eqref{top} shows \eqref{schurtheorem}.

Our first new result in this paper is the extension of Theorem \ref{mainthm}.
We derive the correspondence between a larger class of partition functions
and the products of the Schur polynomials, which we call as the tensor Schur polynomials.

Introduce $2p$ ($p \in \mathbb{Z}_{ \geq 1}$) sequences of integers 
$\{ m_i \}:=\{ m_{i,1}>m_{i,2}>\cdots>m_{i, k_i} \}$,
$\{ n_i \}:=\{ n_{i,1} <n_{i,2}<\cdots<n_{i, \ell_i} \}$, $i=1,\dots,p$
such that $m_{i,1} \le n_{i,1}$ for $i=1,\dots,p$, $n_{i,\ell_i} \le m_{i+1,k_{i+1}}$ for $i=1,\dots,p-1$,
$0 \le m_{1, k_1}$ and $n_{p,\ell_p} \le n$.
If $\{ m_i \}=\varnothing$, we set $k_i=0$, and we set $\ell_i=0$ if $\{ n_i \}=\varnothing$.
We introduce the set of variables
$
\mathbf{z}_{i,j}=\{z_{i,j,k} \ | \ k=1,\dots, |\mathbf{z}_{i,j}| \}
$,
$
\mathbf{w}_{i,j}=\{w_{i,j,k} \ | \ k=1,\dots, |\mathbf{w}_{i,j}| \}
$
and define
\begin{align}
X_{m_{i,j}}(\mathbf{z}_{i,j}):=\prod_{k=1}^{|\mathbf{z}_{i,j}|}
X_{m_{i,j}}(z_{i,j,k}), \ \ \
X_{n_{i,j}}(\mathbf{w}_{i,j}):=\prod_{k=1}^{|\mathbf{w}_{i,j}|}
X_{n_{i,j}}(w_{i,j,k}).
\end{align}
Since $[X_i(x), X_i(y)]=0$, the ordering of the operators inside
$X_{m_{i,j}}(\mathbf{z}_{i,j})$ and $X_{n_{i,j}}(\mathbf{w}_{i,j})$ do not matter.
We further define $\mathbf{z}_{i}:=\mathbf{z}_{i,1} \cup  \mathbf{z}_{i,2}
\cup \cdots \cup \mathbf{z}_{i,k_i}$,
$\mathbf{w}_{i}:=\mathbf{w}_{i,1} \cup  \mathbf{w}_{i,2}
\cup \cdots \cup \mathbf{w}_{i,\ell_i}$,
and introduce
\begin{align}
\overrightarrow{X_{m_i}}(\mathbf{z}_{i})&:=
X_{m_{i,k_i}}(\mathbf{z}_{i,k_i}) \cdots
X_{m_{i,2}}(\mathbf{z}_{i,2})
X_{m_{i,1}}(\mathbf{z}_{i,1}), \\
\overleftarrow{X_{m_i}}(\mathbf{z}_{i})&:=
X_{m_{i,1}}(\mathbf{z}_{i,1})
X_{m_{i,2}}(\mathbf{z}_{i,2})
\cdots
X_{m_{i,k_i}}(\mathbf{z}_{i,k_i}),
 \\
X_{n_i}(\mathbf{w}_{i})&:=
X_{n_{i,1}}(\mathbf{w}_{i,1})
X_{n_{i,2}}(\mathbf{w}_{i,2})
\cdots
X_{n_{i,\ell_i}}(\mathbf{w}_{i,\ell_i}).
\end{align}

\begin{theorem}
We have
\begin{align}
&\langle \Omega|
\overleftarrow{X_{m_1}}(\mathbf{z}_{1})
X_{n_1}(\mathbf{w}_{1})
\overleftarrow{X_{m_2}}(\mathbf{z}_{2})
X_{n_2}(\mathbf{w}_{2})
\cdots
\overleftarrow{X_{m_p}}(\mathbf{z}_{p})
X_{n_p}(\mathbf{w}_{p})
|\Omega \rangle \nonumber \\
=&\prod_{i=1}^p \prod_{j=1}^{k_i} 
{\mathbf z_{i,j}}^{k_{i}-j}
\prod_{i=1}^p \prod_{j=1}^{\ell_i} 
{\mathbf{w} _{i,j} }^{n_{i,j}}
\prod_{i=1}^p
s_{(
(m_{i,1}-k_i+1)^{|\mathbf{z}_{i,1}|},
(m_{i,2}-k_i+2)^{|\mathbf{z}_{i,2}|},\dots,
m_{i,k_i}^{|\mathbf{z}_{i,k_i}|}
)
}({\mathbf{z}}_i).
\label{productSchur}
\end{align}
\end{theorem}

\begin{proof}
We can apply the same argument which is used to show \eqref{schurtheorem}.
First, from \eqref{top}, we have
\begin{align}
&\langle \Omega|
\overrightarrow{X_{m_1}}(\mathbf{z}_{1})
X_{n_1}(\mathbf{w}_{1})
\overrightarrow{X_{m_2}}(\mathbf{z}_{2})
X_{n_2}(\mathbf{w}_{2})
\cdots
\overrightarrow{X_{m_p}}(\mathbf{z}_{p})
X_{n_p}(\mathbf{w}_{p})
|\Omega \rangle \nonumber \\
=&\prod_{i=1}^p \prod_{j=1}^{k_i} 
{\mathbf z_{i,j}}^{m_{i,j}}
\prod_{i=1}^p \prod_{j=1}^{\ell_i} 
{\mathbf{w} _{i,j} }^{n_{i,j}}
=\prod_{i=1}^p \prod_{j=1}^{k_i} \prod_{k=1}^{|\mathbf{z}_{i,j}|}
z_{i,j,k}^{m_{i,j}}
\prod_{i=1}^p \prod_{j=1}^{\ell_i} \prod_{k=1}^{|\mathbf{w}_{i,j}|}
w_{i,j,k}^{n_{i,j}}. \label{freezecaserewrite}
\end{align}

Applying the commutation relations
\eqref{multiplecommrel} multiple times,
we get

\begin{align}
&\langle \Omega|
\overleftarrow{X_{m_1}}(\mathbf{z}_{1})
X_{n_1}(\mathbf{w}_{1})
\overleftarrow{X_{m_2}}(\mathbf{z}_{2})
X_{n_2}(\mathbf{w}_{2})
\cdots
\overleftarrow{X_{m_p}}(\mathbf{z}_{p})
X_{n_p}(\mathbf{w}_{p})
|\Omega \rangle \nonumber \\
=&\prod_{i=1}^p \prod_{j=1}^{k_i} 
{\mathbf z_{i,j}}^{k_{i}-j}
\sum_{\mathbf{u}_1=(\mathbf{u}_{1,1},
\mathbf{u}_{1,2},\dots,\mathbf{u}_{1,k_1}
)}
\frac{1}{
\prod_{j=1}^{k_1} {\bf u}_{1,j}^{k_1-j}
\prod_{1 \le j < k \le k_1}
(
1-\mathbf{u}_{1,k}/ \mathbf{u}_{1,j}
)
} \nonumber \\
&\times
\sum_{\mathbf{u}_2=(\mathbf{u}_{2,1},
\mathbf{u}_{2,2},\dots,\mathbf{u}_{2,k_2}
)}
\frac{1}{
\prod_{j=1}^{k_2} {\bf u}_{2,j}^{k_2-j}
\prod_{1 \le j < k \le k_2}
(
1-\mathbf{u}_{2,k}/ \mathbf{u}_{2,j}
)
} \cdots
\nonumber \\
&\times
\sum_{\mathbf{u}_p=(\mathbf{u}_{p,1},
\mathbf{u}_{p,2},\dots,\mathbf{u}_{p,k_p}
)}
\frac{1}{
\prod_{j=1}^{k_p} {\bf u}_{p,j}^{k_p-j}
\prod_{1 \le j < k \le k_p}
(
\mathbf{u}_{p,j}-\mathbf{u}_{p,k}
)
} \nonumber \\
&\times \langle \Omega|
\overrightarrow{X_{m_1}}(\mathbf{u}_{1})
X_{n_1}(\mathbf{w}_{1})
\overrightarrow{X_{m_2}}(\mathbf{u}_{2})
X_{n_2}(\mathbf{w}_{2})
\cdots
\overrightarrow{X_{m_p}}(\mathbf{u}_{p})
X_{n_p}(\mathbf{w}_{p})
|\Omega \rangle,
\label{multipleapplication}
\end{align}
where $\sum_{\mathbf{u}_i=(\mathbf{u}_{i,1},
\mathbf{u}_{i,2},\dots,\mathbf{u}_{i,k_i}
)}$ means we take sum over all
$\mathbf{u}_i=(\mathbf{u}_{i,1},
\mathbf{u}_{i,2},\dots,\mathbf{u}_{i,k_i}
)$ such that $\mathbf{u}_i=\mathbf{z}_i$ and $|\mathbf{u}_{i,j}|=|\mathbf{z}_{i,j}|$, $j=1,\dots,k_i$.
We insert
\eqref{freezecaserewrite}
into
\eqref{multipleapplication} to get
\begin{align}
&\langle \Omega|
\overleftarrow{X_{m_1}}(\mathbf{z}_{1})
X_{n_1}(\mathbf{w}_{1})
\overleftarrow{X_{m_2}}(\mathbf{z}_{2})
X_{n_2}(\mathbf{w}_{2})
\cdots
\overleftarrow{X_{m_p}}(\mathbf{z}_{p})
X_{n_p}(\mathbf{w}_{p})
|\Omega \rangle \nonumber \\
=&\prod_{i=1}^p \prod_{j=1}^{k_i} 
{\mathbf z_{i,j}}^{k_{i}-j}
\sum_{\mathbf{u}_1=(\mathbf{u}_{1,1},
\mathbf{u}_{1,2},\dots,\mathbf{u}_{1,k_1}
)}
\frac{1}{
\prod_{j=1}^{k_1} {\bf u}_{1,j}^{k_1-j}
\prod_{1 \le j < k \le k_1}
(
1-\mathbf{u}_{1,k}/ \mathbf{u}_{1,j}
)
} \nonumber \\
&\times
\sum_{\mathbf{u}_2=(\mathbf{u}_{2,1},
\mathbf{u}_{2,2},\dots,\mathbf{u}_{2,k_2}
)}
\frac{1}{
\prod_{j=1}^{k_2} {\bf u}_{2,j}^{k_2-j}
\prod_{1 \le j < k \le k_2}
(
1-\mathbf{u}_{2,k}/ \mathbf{u}_{2,j}
)
}
\cdots \nonumber \\
&\times
\sum_{\mathbf{u}_p=(\mathbf{u}_{p,1},
\mathbf{u}_{p,2},\dots,\mathbf{u}_{p,k_p}
)}
\frac{1}{
\prod_{j=1}^{k_p} {\bf u}_{p,j}^{k_p-j}
\prod_{1 \le j < k \le k_p}
(
1-\mathbf{u}_{p,k} / \mathbf{u}_{p,j}
)
} 
 \prod_{i=1}^p \prod_{j=1}^{k_i} 
{\mathbf u_{i,j}}^{m_{i,j}}
\prod_{i=1}^p \prod_{j=1}^{\ell_i} 
{\mathbf{w} _{i,j} }^{n_{i,j}}
\nonumber \\
=&\prod_{i=1}^p \prod_{j=1}^{k_i} 
{\mathbf z_{i,j}}^{k_{i}-j}
\prod_{i=1}^p \prod_{j=1}^{\ell_i} 
{\mathbf{w} _{i,j} }^{n_{i,j}}
\sum_{\mathbf{u}_1=(\mathbf{u}_{1,1},
\mathbf{u}_{1,2},\dots,\mathbf{u}_{1,k_1}
)}
\frac{ \prod_{j=1}^{k_1} 
{\mathbf u_{1,j}}^{m_{1,j}-k_1+j}
}{\prod_{1 \le j < k \le k_1}
(
1-\mathbf{u}_{1,k} / \mathbf{u}_{1,j}
)
} \nonumber \\
&\times
\sum_{\mathbf{u}_2=(\mathbf{u}_{2,1},
\mathbf{u}_{2,2},\dots,\mathbf{u}_{2,k_2}
)}
\frac{
\prod_{j=1}^{k_2} 
{\mathbf u_{2,j}}^{m_{2,j}-k_2+j}
}{\prod_{1 \le j < k \le k_2}
(
1-\mathbf{u}_{2,k} / \mathbf{u}_{2,j}
)
}
\cdots
\sum_{\mathbf{u}_p=(\mathbf{u}_{p,1},
\mathbf{u}_{p,2},\dots,\mathbf{u}_{p,k_p}
)}
\frac{
\prod_{j=1}^{k_p} 
{\mathbf u_{p,j}}^{m_{p,j}-k_p+j}
}{\prod_{1 \le j < k \le k_p}
(
1-\mathbf{u}_{p,k} / \mathbf{u}_{p,j}
)
}. \label{lastforsymm}
\end{align}
Applying the symmetrization formula for monomials
\eqref{symmformulaSchur} to the right-hand side of \eqref{lastforsymm},
we obtain \eqref{productSchur}.

\end{proof}

A special case $p=1$ of
\eqref{productSchur} is

\begin{align}
&\langle \Omega|
\overleftarrow{X_{m_1}}(\mathbf{z}_{1})
\overleftarrow{X_{m_2}}(\mathbf{z}_{2})
|\Omega \rangle 
=\prod_{i=1}^2 \prod_{j=1}^{k_i} 
{\mathbf z_{i,j}}^{k_{i}-j}
\prod_{i=1}^2
s_{(
(m_{i,1}-k_i+1)^{|\mathbf{z}_{i,1}|},
(m_{i,2}-k_i+2)^{|\mathbf{z}_{i,2}|},\dots,
m_{i,k_i}^{|\mathbf{z}_{i,k_i}|}
)
}({\mathbf{z}}_i).
\label{specialcaseproductofschur}
\end{align}

Using
\eqref{specialcaseproductofschur},
we can derive the shuffle formula for the Schur polynomials,
which corresponds to the pushforward formula
by Jo\'zefiak-Lascoux-Pragacz
\cite[Prop. 1]{JLP} in geometry.
First, we prepare the following multiple commutation relations.

\begin{proposition}
We have
\begin{align}
\overleftarrow{X_{m_2}}(\mathbf{z}_{2})
\overleftarrow{X_{m_1}}(\mathbf{z}_{1})
=& \prod_{j=1}^{k_2} 
{\mathbf z_{2,j}}^{k_1+k_{2}-j}
\prod_{j=1}^{k_i} 
{\mathbf z_{1,j}}^{k_{1}-j} \nonumber \\
&\times
\sum_{(
\mathbf{w}_{1}, \mathbf{w}_{2}
)}
 \prod_{j=1}^{k_2} 
{\mathbf w_{2,j}}^{-(k_1+k_{2}-j)}
\prod_{j=1}^{k_i} 
{\mathbf w_{1,j}}^{-(k_{1}-j)} 
\frac{1}{1-
\mathbf{w}_{1}/\mathbf{w}_{2}
}
\overleftarrow{X_{m_1}}(\mathbf{w}_{1})
\overleftarrow{X_{m_2}}(\mathbf{w}_{2}), \label{multipleforshuffle}
\end{align}
where $\sum_{({\bf w}_1, {\bf w}_2)}$
denotes the sum over all 
$({\bf w}_1, {\bf w}_2)$ such that ${\bf w}_1,{\bf w}_2$
are unordered sets of variables satisfying $|{\bf w}_1|=|{\bf z}_1|$,
$|{\bf w}_2|=|{\bf z}_2|$
and ${\bf w}_1  \cup {\bf w}_2 ={\bf z}_1 \cup {\bf z}_2$.
\end{proposition}
\begin{proof}
We apply the argument given in \cite{SU,IMO} for example.
Using the Zamolodchikov-Faddeev algebra relations \eqref{FZalg}
rewritten
in the following forms
\begin{align}
X_i(x)X_j(y)&=(1-y/x)^{-1} X_j(y) X_i(x)-(1-y/x)^{-1} X_j(x) X_i(y), \ i>j, \label{commrelone} \\
X_i(x)X_i(y)&=X_i(y)X_i(x), \label{commreltwo} \\
X_i(x)X_j(y)&=x/y X_i(y) X_j(x), \  i>j, \label{commrelthree}
\end{align}
the commutation relations can be written
in the following form
\begin{align}
\overleftarrow{X_{m_2}}(\mathbf{z}_{2})
\overleftarrow{X_{m_1}}(\mathbf{z}_{1})
=\sum_{
(
{\bf w}_1,{\bf w}_2
)
}
A(
{\bf w}_1,{\bf w}_2
)
\overleftarrow{X_{m_1}}(\mathbf{w}_{1})
\overleftarrow{X_{m_2}}(\mathbf{w}_{2}). \label{generalformforshuffle}
\end{align}
We also note that $
\prod_{j=1}^{k_2} 
{\mathbf z_{2,j}}^{j-k_1-k_{2}}
\prod_{j=1}^{k_i} 
{\mathbf z_{1,j}}^{j-k_{1}}
\overleftarrow{X_{m_2}}(\mathbf{z}_{2})
\overleftarrow{X_{m_1}}(\mathbf{z}_{1})
$ is invariant under any permutation of $({\bf z}_1,{\bf z}_2)$ which follows from \eqref{commreltwo} and \eqref{commrelthree}.

To get $A(
{\bf w}_1,{\bf w}_2
)$ for a fixed $(
{\bf w}_1,{\bf w}_2
)$, we rewrite the left-hand side of \eqref{generalformforshuffle} using the invariance stated above as
\begin{align}
\overleftarrow{X_{m_2}}(\mathbf{z}_{2})
\overleftarrow{X_{m_1}}(\mathbf{z}_{1})
=
& \prod_{j=1}^{k_2} 
{\mathbf z_{2,j}}^{k_1+k_{2}-j}
\prod_{j=1}^{k_i} 
{\mathbf z_{1,j}}^{k_{1}-j} \nonumber \\
&\times
 \prod_{j=1}^{k_2} 
{\mathbf w_{2,j}}^{-(k_1+k_{2}-j)}
\prod_{j=1}^{k_i} 
{\mathbf w_{1,j}}^{-(k_{1}-j)} 
\overleftarrow{X_{m_2}}(\mathbf{w}_{2})
\overleftarrow{X_{m_1}}(\mathbf{w}_{1})
,
\end{align}
and reverse the order of 
$\overleftarrow{X_{m_2}}(\mathbf{w}_{2})$ and $\overleftarrow{X_{m_1}}(\mathbf{w}_{1})$ by
just
using the first term of the right-hand side of \eqref{commrelone},
from which we get the factor $\frac{1}{
1-{\bf w}_1 / {\bf w}_2
}$.
Hence we have
\begin{align}
A(
{\bf w}_1,{\bf w}_2
)=& 
\prod_{j=1}^{k_2} 
{\mathbf z_{2,j}}^{k_1+k_{2}-j}
\prod_{j=1}^{k_i} 
{\mathbf z_{1,j}}^{k_{1}-j} \nonumber \\
&\times
 \prod_{j=1}^{k_2} 
{\mathbf w_{2,j}}^{-(k_1+k_{2}-j)}
\prod_{j=1}^{k_i} 
{\mathbf w_{1,j}}^{-(k_{1}-j)} 
\frac{1}{
1-{\bf w}_1 / {\bf w}_2
}.
\end{align}
\end{proof}

With the use of
\eqref{multipleforshuffle},
we give a derivation of the shuffle formula for the Schur polynomials.

\begin{theorem}
We have
\begin{align}
&
s_{(
(m_{2,1}-k_1-k_2+1)^{|\mathbf{z}_{2,1}|},
\dots,
(m_{2,k_2}-k_1)^{|\mathbf{z}_{2,k_2}|},
(m_{1,1}-k_1+1)^{|\mathbf{z}_{1,1}|},
\dots,
m_{1,k_1}^{|\mathbf{z}_{1,k_1}|}
)
}({\mathbf{z}}_1,{\mathbf{z}}_2)
\nonumber \\
=&
\sum_{(
\mathbf{w}_{1}, \mathbf{w}_{2}
)}
\frac{1}{\mathbf{w}_{2}-
\mathbf{w}_{1}
}
s_{(
(m_{1,1}-k_1+1)^{|\mathbf{z}_{1,1}|},
(m_{1,2}-k_1+2)^{|\mathbf{z}_{1,2}|},\dots,
m_{1,k_1}^{|\mathbf{z}_{1,k_1}|}
)
}({\mathbf{w}}_1) \nonumber \\
&\times
s_{(
(m_{2,1}+|{\mathbf{z}}_1|-k_1-k_2+1)^{|\mathbf{z}_{2,1}|},
(m_{2,2}+|{\mathbf{z}}_1|-k_1-k_2+2)^{|\mathbf{z}_{2,2}|},\dots,
(m_{2,k_2}+|{\mathbf{z}}_1|-k_1)^{|\mathbf{z}_{2,k_2}|}
)
}({\mathbf{w}}_2). \label{JLP}
\end{align}
Here, $\sum_{({\bf w}_1, {\bf w}_2)}$
denotes the sum over all 
$({\bf w}_1, {\bf w}_2)$ such that ${\bf w}_1,{\bf w}_2$
are unordered sets of variables satisfying $|{\bf w}_1|=|{\bf z}_1|$,
$|{\bf w}_2|=|{\bf z}_2|$
and ${\bf w}_1  \cup {\bf w}_2 ={\bf z}_1 \cup {\bf z}_2$.

\end{theorem}

\begin{proof}

First, from \eqref{schurtheorem},
we have the following expression for $\langle \Omega|
\overleftarrow{X_{m_2}}(\mathbf{z}_{2})
\overleftarrow{X_{m_1}}(\mathbf{z}_{1})
|\Omega \rangle$
\begin{align}
&\langle \Omega|
\overleftarrow{X_{m_2}}(\mathbf{z}_{2})
\overleftarrow{X_{m_1}}(\mathbf{z}_{1})
|\Omega \rangle 
= \prod_{j=1}^{k_2} 
{\mathbf z_{2,j}}^{k_1+k_{2}-j}
\prod_{j=1}^{k_i} 
{\mathbf z_{1,j}}^{k_{1}-j} \nonumber \\
&
\times s_{(
(m_{2,1}-k_1-k_2+1)^{|\mathbf{z}_{2,1}|},
\dots,
(m_{2,k_2}-k_1)^{|\mathbf{z}_{2,k_2}|},
(m_{1,1}-k_1+1)^{|\mathbf{z}_{1,1}|},
\dots,
m_{1,k_1}^{|\mathbf{z}_{1,k_1}|}
)
}({\mathbf{z}}_1,{\mathbf{z}}_2). \label{oneforJLP}
\end{align}

On the other hand, combining
\eqref{specialcaseproductofschur}
and \eqref{multipleforshuffle}, we have
\begin{align}
&\langle \Omega|
\overleftarrow{X_{m_2}}(\mathbf{z}_{2})
\overleftarrow{X_{m_1}}(\mathbf{z}_{1})
|\Omega \rangle 
= \prod_{j=1}^{k_2} 
{\mathbf z_{2,j}}^{k_1+k_{2}-j}
\prod_{j=1}^{k_i} 
{\mathbf z_{1,j}}^{k_{1}-j} \nonumber \\
&\times
\sum_{(
\mathbf{w}_{1}, \mathbf{w}_{2}
)}
 \prod_{j=1}^{k_2} 
{\mathbf w_{2,j}}^{-k_1}
\frac{1}{1-
\mathbf{w}_{1}/\mathbf{w}_{2}
}
\prod_{i=1}^2
s_{(
(m_{i,1}-k_i+1)^{|\mathbf{z}_{i,1}|},
(m_{i,2}-k_i+2)^{|\mathbf{z}_{i,2}|},\dots,
m_{i,k_i}^{|\mathbf{z}_{i,k_i}|}
)
}({\mathbf{w}}_i). \label{preanotherevaljlp}
\end{align}

Using
\begin{align}
\prod_{j=1}^{k_2} 
{\mathbf w_{2,j}}^{-k_1}
\frac{1}{1-
\mathbf{w}_{1}/\mathbf{w}_{2}
}
={\mathbf{w}}_2^{|{\mathbf{z}}_1|-k_1}
\frac{1}{\mathbf{w}_{2}-
\mathbf{w}_{1}
}, \nonumber
\end{align} 
and
\begin{align}
&{\mathbf{w}}_2^{|{\mathbf{z}}_1|-k_1}
s_{(
(m_{2,1}-k_2+1)^{|\mathbf{z}_{2,1}|},
(m_{2,2}-k_2+2)^{|\mathbf{z}_{2,2}|},\dots,
m_{2,k_2}^{|\mathbf{z}_{2,k_2}|}
)
}({\mathbf{w}}_2) \nonumber \\
=&
s_{(
(m_{2,1}+|{\mathbf{z}}_1|-k_1-k_2+1)^{|\mathbf{z}_{2,1}|},
(m_{2,2}+|{\mathbf{z}}_1|-k_1-k_2+2)^{|\mathbf{z}_{2,2}|},\dots,
(m_{2,k_2}+|{\mathbf{z}}_1|-k_1)^{|\mathbf{z}_{2,k_2}|}
)
}({\mathbf{w}}_2), \nonumber
\end{align}
\eqref{preanotherevaljlp} can be rewritten as
\begin{align}
&\langle \Omega|
\overleftarrow{X_{m_2}}(\mathbf{z}_{2})
\overleftarrow{X_{m_1}}(\mathbf{z}_{1})
|\Omega \rangle 
= \prod_{j=1}^{k_2} 
{\mathbf z_{2,j}}^{k_1+k_{2}-j}
\prod_{j=1}^{k_i} 
{\mathbf z_{1,j}}^{k_{1}-j} \nonumber \\
&\times
\sum_{(
\mathbf{w}_{1}, \mathbf{w}_{2}
)}
\frac{1}{\mathbf{w}_{2}-
\mathbf{w}_{1}
}
s_{(
(m_{1,1}-k_1+1)^{|\mathbf{z}_{1,1}|},
(m_{1,2}-k_1+2)^{|\mathbf{z}_{1,2}|},\dots,
m_{1,k_1}^{|\mathbf{z}_{1,k_1}|}
)
}({\mathbf{w}}_1) \nonumber \\
&\times
s_{(
(m_{2,1}+|{\mathbf{z}}_1|-k_1-k_2+1)^{|\mathbf{z}_{2,1}|},
(m_{2,2}+|{\mathbf{z}}_1|-k_1-k_2+2)^{|\mathbf{z}_{2,2}|},\dots,
(m_{2,k_2}+|{\mathbf{z}}_1|-k_1)^{|\mathbf{z}_{2,k_2}|}
)
}({\mathbf{w}}_2), \label{anotherforJLP}
\end{align}
which is another expression for $\langle \Omega|
\overleftarrow{X_{m_2}}(\mathbf{z}_{2})
\overleftarrow{X_{m_1}}(\mathbf{z}_{1})
|\Omega \rangle$.
Comparing 
\eqref{oneforJLP} and \eqref{anotherforJLP},
we get \eqref{JLP}.
\end{proof}

Setting
$k_1=|\mathbf{z}_1|=:r-d, k_2=|\mathbf{z}_2|=:d$
and $\mu_1:=m_{2,1}-d+1,\dots,\mu_d:=m_{2,d}$,
$\nu_{d+1}=m_{1,1}-(r-d)+1,\dotsm\nu_r=m_{1,r-d}$,
\eqref{JLP}
becomes
\begin{align}
s_{(\mu_1+d-r,\dots,\mu_{d}+d-r,\nu_{d+1},\dots,\nu_r)}
(\mathbf{z}_1,\mathbf{z}_2)
=\sum_{(\mathbf{w}_1,\mathbf{w}_2)}
\frac{1}{\mathbf{w}_1-\mathbf{w}_2}
s_{(\nu_{d+1},\dots,\nu_r)}(\mathbf{w}_1) s_{(\mu_1,\dotsm\mu_d)}(\mathbf{w}_2). \label{JLPtoseecorrespondence}
\end{align}

This identity is the symmetric function counterpart of the J\'ozefiak--Lascoux--Pragacz pushforward formula for Grassmann bundles.
Let us recall the geometric setup underlying this formula.

Let $X$ be a nonsingular variety, and let $E$ be a vector bundle of rank $r$ over $X$.
Let
\[
\pi : G_X(d, E) \to X
\]
be the Grassmann bundle parametrizing rank $d$ quotient bundles of $E$, and let
\[
0 \to S \to \pi^*E \to Q \to 0
\]
be the universal exact sequence, where $S$ and $Q$ are the universal subbundle and quotient bundle of ranks $r-d$ and $d$, respectively.

For a partition $\lambda=(\lambda_1,\dots,\lambda_d)$, the class
\[
\Delta_\lambda(s(Q))
:=
\det\bigl[s_{\lambda_i+j-i}(Q)\bigr]_{1\le i,j\le d}
\]
is defined in terms of the Segre classes $s_k(Q)$ of $Q$.
If $\mathbf{z}_2=(\alpha_1,\dots,\alpha_d)$ are the Chern roots of $Q$, then the Segre classes satisfy
\[
s_k(Q)=h_k(\mathbf{z}_2),
\]
where $h_k$ is the complete symmetric function, and therefore by the Jacobi--Trudi identity,
\[
\Delta_\lambda(s(Q)) = s_\lambda(\mathbf{z}_2),
\]
the Schur polynomial in the variables $\mathbf{z}_2$.
Similarly,
\[
\Delta_\nu(s(S)) = s_\nu(\mathbf{z}_1),
\]
where $\mathbf{z}_1=(\beta_1,\dots,\beta_{r-d})$ are the Chern roots of $S$.

\begin{theorem} \cite[Proposition 4.1]{JLP}
Let
\[
\pi : G_X(d,E) \to X
\]
be the Grassmann bundle parametrizing rank $d$ quotient bundles of $E$, and let
\[
0 \to S \to \pi^*E \to Q \to 0
\]
be the universal exact sequence, where $S$ and $Q$ are the universal subbundle and quotient bundle of ranks $r-d$ and $d$, respectively.

Let $\mu=(\mu_1,\dots,\mu_d)$ and $\nu=(\nu_{d+1},\dots,\nu_r)$ be partitions.
Then the pushforward along $\pi$ satisfies
\begin{align}
\pi_*\Bigl(
\Delta_\mu(s(Q))\,\Delta_\nu(s(S))
\Bigr)
=
\Delta_{(\mu_1+d-r,\dots,\mu_d+d-r,\nu_{d+1},\dots,\nu_r)}(s(E)),
\label{JLPgeometry}
\end{align}
where for any sequence $\lambda$,
\[
\Delta_\lambda(s(E))
:=
\det\bigl[s_{\lambda_i+j-i}(E)\bigr],
\]
and $s_k(E)$ denotes the $k$th Segre class of $E$.
If the sequence $(\mu_1+d-r,\dots,\mu_d+d-r,\nu_{d+1},\dots,\nu_r)$ is not a partition, the right-hand side is understood to be zero.
\end{theorem}

\eqref{JLPtoseecorrespondence}
is the symmetric function counterpart of the pushforward formula \eqref{JLPgeometry}.
Note pushforward corresponds to symmetrization in terms of functions.
More concretely, for a function $P(\mathbf{z}_1;\mathbf{z}_2) \in \mathbb{Z}[\mathbf{z}_1;\mathbf{z}_2]$
invariant under permutations of the Chern roots $\mathbf{z}_2=(\alpha_1,\dots,\alpha_d)$ of $Q$
and also invariant under permutations of the Chern roots $\mathbf{z}_1=(\beta_1,\dots,\beta_{r-d})$ 
of $S$,
then we have
\begin{align}
\pi_*(P(\mathbf{z}_1;\mathbf{z}_2))=\sum_{(\mathbf{w}_1,\mathbf{w}_2)}
\frac{1}{\mathbf{w}_1-\mathbf{w}_2}
P(\mathbf{w}_1;\mathbf{w}_2). \label{pushforwardsymmetrization}
\end{align}
Taking $P(\mathbf{z}_1;\mathbf{z}_2)$ to be
$P(\mathbf{z}_1;\mathbf{z}_2)=\Delta_\mu(s(Q))\,\Delta_\nu(s(S))=s_{(\nu_{d+1},\dots,\nu_r)}(\mathbf{z}_1) s_{(\mu_1,\dotsm\mu_d)}(\mathbf{z}_2)$, \\
$\pi_*(P(\mathbf{z}_1;\mathbf{z}_2))$ is nothing but the right hand side of \eqref{JLPtoseecorrespondence}.
\eqref{JLPgeometry} states that this becomes \\
 $\Delta_{(\mu_1+d-r,\dots,\mu_d+d-r,\nu_{d+1},\dots,\nu_r)}(s(E))
=s_{(\mu_1+d-r,\dots,\mu_{d}+d-r,\nu_{d+1},\dots,\nu_r)}
(\mathbf{z}_1,\mathbf{z}_2)$, the left hand side of \eqref{JLPtoseecorrespondence}.
\eqref{pushforwardsymmetrization} is a special case
of the formula for the partial flag bundles and complex cobordism.
See \cite[Theorem 1.8]{BresslerEvens}, \cite[Theorem 2.5, Corollary 2.6]{NakagawaNaruse}
for example. \eqref{JLPgeometry} has applications to geometry,
for example the degree
formula of the Grassmann bundles in \cite{KT}.


%

\subsection{Identities}

We derive other algebraic formulas in this subsection.
First we prepare the following lemma,
\begin{lemma}
We have
\begin{align}
X_n^{(n)}(z)|\Omega \rangle=z^n |\Omega \rangle, \label{actiononvacuum}
\\
\langle \Omega|X_0^{(n)}(z)=\langle \Omega|. \label{actionondualvacuum}
\end{align}
\end{lemma}
\begin{proof}
We show \eqref{actiononvacuum} as
\eqref{actionondualvacuum} can be proved in a similar way.
Figure \ref{XIoperatorfigureactionvac}, left panel is the graphical description of
$X_n^{(n)}(z)|\Omega \rangle$. We make the following observations
on the $L$-operators of $X_n^{(n)}(z)$ acting on the Fock spaces $\mathcal{F}_{k \ell}$, $k+\ell=n$.
For each $L$-operator, since the two edges are already colored by red,
we note there are two choices: color the other two edges by red or blue, which produces
the identity operator and the annihilation operator. Since we act on the vacuum state $|\Omega \rangle$,
annihilation operators are not allowed to use, hence we find the remaining two edges must be colored by
red for every $L$-operator, which is graphically Figure \ref{XIoperatorfigureactionvac}, middle panel.
Repeating this observation, we find that all edges must be colored by red (Figure \ref{XIoperatorfigureactionvac}, right panel).
Every $L$-operator produces an identity operator, and since $\alpha_1=\cdots=\alpha_n=1$,
we get \eqref{actiononvacuum}.
\end{proof}

\begin{figure}[htbp]
\centering
\includegraphics[width=12truecm]{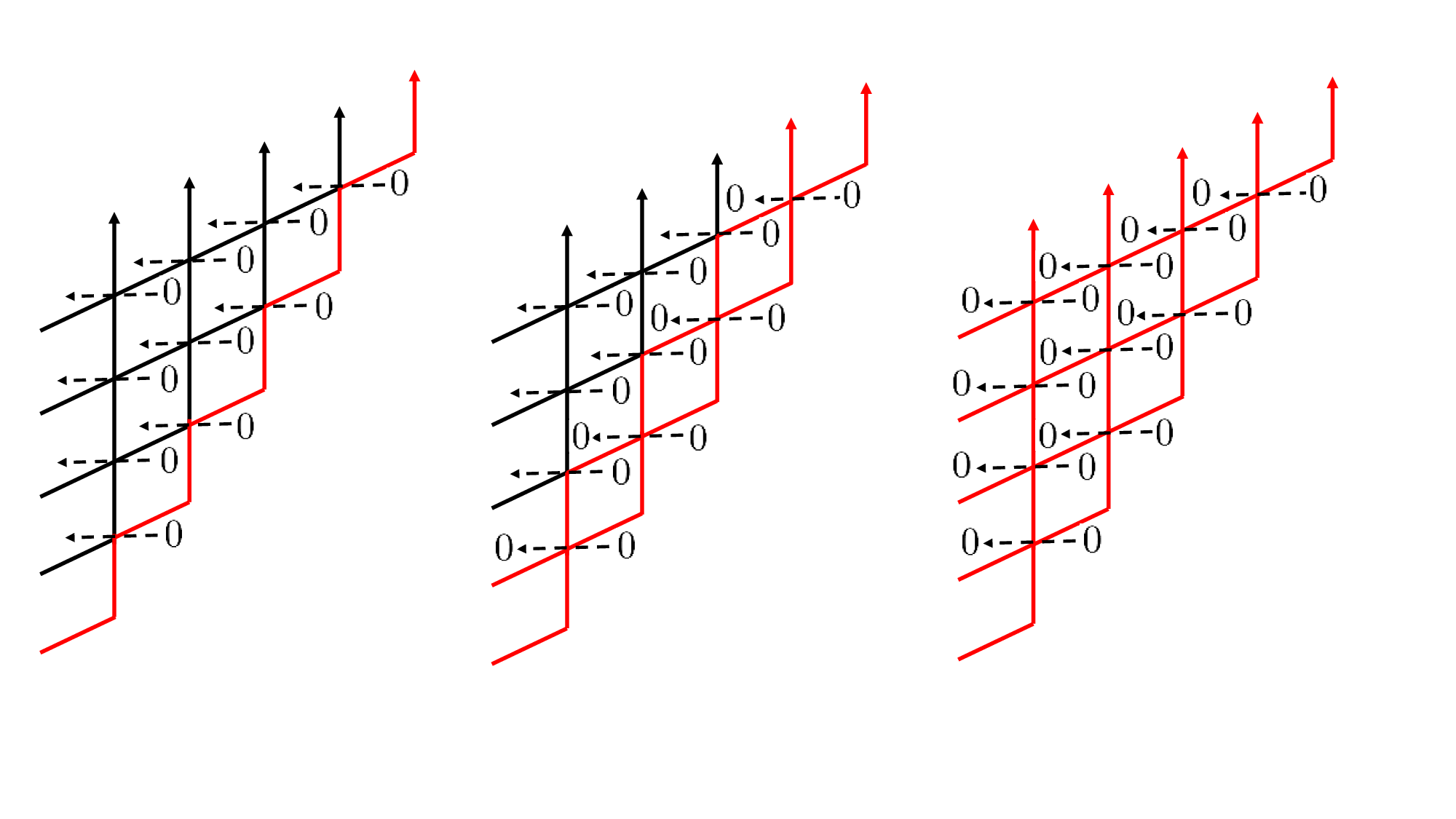}
\caption{The action of $X_n^{(n)}(z)$ on  $|\Omega \rangle$.
}
\label{XIoperatorfigureactionvac}
\end{figure}

We also prepare multiple commutation relations.

For $i_{2} >i_3> \cdots > i_{m}$ and ${\bf z}_{[2,m]}=
(
{\bf z}_2,{\bf z}_3,\dots,{\bf z}_m
)$, introduce
\begin{align}
X_{i_2,\dots,i_m}({\bf z}_{[2,m]})
:=
\frac{1}{\prod_{k=2}^m {\bf z}_k^{m-k}}
X_{i_2}({\bf z}_2) X_{i_3}({\bf z}_3)
\cdots X_{i_m}({\bf z}_m).
\end{align}
Note
$X_{i_2,\dots,i_m}({\bf z}_{[2,m]})$ is invariant under the permutation 
of ${\bf z}_{[2,m]}$.

\begin{proposition}
For $n \geq  i_{1} >i_2> \cdots > i_{m} \geq 0$ and ${\bf z}_{[2,m]}=
(
{\bf z}_2,{\bf z}_3,\dots,{\bf z}_m
)$, we have
\begin{align}
\frac{X_{i_1}({\bf z}_1) X_{i_2,\dots,i_m}({\bf z}_{[2,m]})}
{
{\bf z}_1^{m-1}
}
=\sum_{
(
{\bf w}_1,{\bf w}_{[2,m]}
)
}
\frac{1}{
1-{\bf w}_{[2,m]} / {\bf w}_1
}
\frac{ X_{i_2,\dots,i_m}({\bf w}_{[2,m]})
X_{i_1}({\bf w}_1)
}
{
{\bf w}_1^{m-1}
}, \label{multipleformulaone}
\end{align}
where 
$\displaystyle \sum_{
(
{\bf w}_1,{\bf w}_{[2,m]}
)
}
$
denotes the sum over all 
$({\bf w}_1, {\bf w}_{[2,m]})$ such that ${\bf w}_1,{\bf w}_{[2,m]}$
are unordered sets of variables satisfying $|{\bf w}_1|=|{\bf z}_1|$,
$|{\bf w}_{[2,m]}|=|{\bf z}_{[2,m]}|$
and ${\bf w}_1  \cup {\bf w}_{[2,m]} ={\bf z}_1 \cup {\bf z}_{[2,m]}$.

\end{proposition}
\begin{proof}
The argument to derive \eqref{multipleforshuffle} can be applied to derive \eqref{multipleformulaone}
as well.
The Zamolodchikov-Faddeev algebra relations in the form
\eqref{commrelone}, \eqref{commreltwo}, \eqref{commrelthree}
implies that
we can write the commutation relation in the following form
\begin{align}
X_{i_1}({\bf z}_1) X_{i_2,\dots,i_m}({\bf z}_{[2,m]})=\sum_{
(
{\bf w}_1,{\bf w}_{[2,m]}
)
}
A(
{\bf w}_1,{\bf w}_{[2,m]}
)
X_{i_2,\dots,i_m}({\bf w}_{[2,m]})
X_{i_1}({\bf w}_1). \label{generalform}
\end{align}
To get $A(
{\bf w}_1,{\bf w}_{[2,m]}
)$ for a fixed $(
{\bf w}_1,{\bf w}_{[2,m]}
)$, we rewrite the left-hand side of \eqref{generalform} as
\begin{align}
X_{i_1}({\bf z}_1) X_{i_2,\dots,i_m}({\bf z}_{[2,m]})
=
\frac{{\bf z}_1^{m-1}}{{\bf w}_1^{m-1}}
X_{i_1}({\bf w}_1) X_{i_2,\dots,i_m}({\bf w}_{[2,m]}),
\end{align}
and reverse the order of 
$X_{i_1}({\bf w}_1)$ and $X_{i_2,\dots,i_m}({\bf w}_{[2,m]})$ by
just
using the first term of the right hand side of \eqref{commrelone},
from which we get the factor $\frac{1}{
1-{\bf w}_{[2,m]} / {\bf w}_1
}$.
Hence we have
$A(
{\bf w}_1,{\bf w}_{[2,m]}
)= 
\frac{{\bf z}_1^{m-1}}{{\bf w}_1^{m-1}}
\frac{1}{
1-{\bf w}_{[2,m]} / {\bf w}_1
}$
\end{proof}

\begin{proposition}
For $n  > i_{2} > \cdots > i_{m} \geq 0$ and
${\bf z}_{[2,m]} =
(
{\bf z}_2,{\bf z}_3,\dots,{\bf z}_{m}
)$,
we have
\begin{align}
&s_{((n-m+1)^{|{\bf z}_1|},(i_2-m+2)^{|{\bf z}_2|},\dots,
i_{m}^{| {\bf z}_{m}       |})}({\bf z}) \nonumber \\
=
&\sum_{{\bf w}=({\bf w}_1, {\bf w}_{[2,m]})}
s_{((i_2-m+2)^{|{\bf z}_2|},
(i_3-m+3)^{|{\bf z}_3|},
\dots,
i_{m}^{| {\bf z}_{m}       |} )}({\bf w}_{[2,m]} )
\frac{{\bf w}_1^{n+1-m+|{\bf z}_{[2,m]}|}}{ {\bf w}_1-{\bf w}_{[2,m]} },
 \label{FNR}
\end{align}
\end{proposition}
where $\sum_{({\bf w}_1, {\bf w}_{[2,m]})}$
denotes the sum over all 
$({\bf w}_1, {\bf w}_{[2,m]})$ such that ${\bf w}_1,{\bf w}_{[2,m]}$
are unordered sets of variables satisfying $|{\bf w}_1|=|{\bf z}_1|$,
$|{\bf w}_{[2,m]}|=|{\bf z}_{[2,m]}|$
and ${\bf w}_1  \cup {\bf w}_{[2,m]} ={\bf z}_1 \cup {\bf z}_{[2,m]}$.

\begin{proof}

We set $i_1=n$ and take
the vacuum expectation values $\langle \Omega | \cdot |\Omega \rangle$ of both-hand sides of
\eqref{multipleformulaone}. Applying \eqref{schurtheorem},
the left-hand side becomes $s_{((n-m+1)^{|{\bf z}_1|},(i_2-m+2)^{|{\bf z}_2|},\dots,
i_{m}^{| {\bf z}_{m}       |})}({\bf z})$.
The right-hand side becomes
\begin{align}
\sum_{
(
{\bf w}_1,{\bf w}_{[2,m]}
)
}
\frac{1}{
1-{\bf w}_{[2,m]} / {\bf w}_1
}
\frac{ 
\langle \Omega|
X_{i_2,\dots,i_m}({\bf w}_{[2,m]})
X_{n}({\bf w}_1) | \Omega \rangle
}
{
{\bf w}_1^{m-1}
}. \label{multiplecommrelforfnrvac}
\end{align}
From
\eqref{actiononvacuum}, we have
\begin{align}
X_{n}({\bf w}_1) |\Omega \rangle
 =&{\bf w}_1^n |\Omega \rangle,
 \end{align}
 and
\eqref{multiplecommrelforfnrvac} becomes
\begin{align}
\sum_{
(
{\bf w}_1,{\bf w}_{[2,m]}
)
}
\frac{
{\bf w}_1^{n+1-m}
}{
1-{\bf w}_{[2,m]} / {\bf w}_1
}
\langle \Omega|
X_{i_2,\dots,i_m}({\bf w}_{[2,m]})
 | \Omega \rangle
. \label{multiplecommrelforfnrvactwo}
\end{align}
Applying \eqref{schurtheorem}, we have
\begin{align}
\langle \Omega|
X_{i_2,\dots,i_m}({\bf w}_{[2,m]})
 | \Omega \rangle
  =s_{((i_2-m+2)^{|{\bf z}_2|},
(i_3-m+3)^{|{\bf z}_3|},
\dots,
i_{m}^{| {\bf z}_{m}       |} )}({\bf w}_{[2,m]} ),
\end{align}
and  \eqref{multiplecommrelforfnrvactwo} is rewritten as
\begin{align}
\sum_{
(
{\bf w}_1,{\bf w}_{[2,m]}
)
}
\frac{
{\bf w}_1^{n+1-m}
}{
1-{\bf w}_{[2,m]} / {\bf w}_1
}
s_{((i_2-m+2)^{|{\bf z}_2|},
(i_3-m+3)^{|{\bf z}_3|},
\dots,
i_{m}^{| {\bf z}_{m}       |} )}({\bf w}_{[2,m]} ).
\end{align}
Using $1-{\bf w}_{[2,m]} / {\bf w}_1=({\bf w}_1-{\bf w}_{[2,m]}) {\bf w}_1^{-|{\bf w}_{[2,m]}|}
=({\bf w}_1-{\bf w}_{[2,m]}) {\bf w}_1^{-|{\bf z}_{[2,m]}|}
$, we get the right-hand side of \eqref{FNR}.
\end{proof}

\eqref{FNR} is the identity by Feh{\'e}r--N{\'e}methi--Rim{\'a}nyi
 \cite{FNR}.
See also Guo-Sun \cite{GS} where a generalization is given.
The identity written in the form 
\begin{align}
s_{((m-k)^{n-k},\lambda_1,\lambda_2,\dots,\lambda_k)}
(z_1,z_2,\dots,z_n)
=\sum_{S \in \binom{[n]}{k}} s_\lambda (z_S)
\frac{\displaystyle
\prod_{j \in \overline{S}} z_j^m
}{\displaystyle
\prod_{i \in S} \prod_{j \in \overline{S}}(z_j-z_i)
}.
\end{align}
in \cite[(1.4)]{GS} is the same identity, written in the traditional form.
Here $\binom{[n]}{k}$ denotes the family of all $k$-subsets of $[n]=\{1,2,\dots,n\}$. 
For $S\in\binom{[n]}{k}$ we write $\overline{S}=[n]\setminus S$ for its complement. 
The notation $s_\lambda(z_S)$ means the Schur polynomial in the variables 
$\{z_i : i\in S\}$, viewed as a $k$-variable symmetric polynomial. 
The products over $i\in S$ and $j\in\overline{S}$ represent all cross pairs between 
the chosen variables $z_i$ and the remaining variables $z_j$.

Another similar but distinct identity due to Gustafson and Milne \cite{GM}
is written in the form
\begin{align}
s_{(\lambda_1 - n + k, \lambda_2 - n + k, \dots, \lambda_k - n + k)}
(z_1,z_2, \dots, z_n)
= \sum_{S \in \binom{[n]}{k}}
\frac{s_\lambda(z_S)}
{ \prod_{i \in S} \prod_{j \in \overline{S}} (z_i - z_j)},
\end{align}
which appears as equation~\cite[(1.1)]{GS}
and can be derived in a similar way, as we explain below.

For $i_{1} >i_2> \cdots > i_{m-1}$ and ${\bf z}_{[1,m-1]}=
(
{\bf z}_1,{\bf z}_2,\dots,{\bf z}_{m-1}
)$, introduce
\begin{align}
X_{i_1,\dots,i_{m-1}}({\bf z}_{[1,m-1]})
&:=
\frac{1}{\prod_{k=1}^{m-1} {\bf z}_k^{m-k}}
X_{i_1}({\bf z}_1) X_{i_2}({\bf z}_2)
\cdots X_{i_{m-1}}({\bf z}_{m-1}) \nonumber \\
&=\frac{1}{\prod_{k=1}^{m-1} {\bf z}_k}
\frac{1}{\prod_{k=1}^{m-1} {\bf z}_k^{m-k-1}}
X_{i_1}({\bf z}_1) X_{i_2}({\bf z}_2)
\cdots X_{i_{m-1}}({\bf z}_{m-1})
.
\end{align}
Note
$X_{i_1,\dots,i_{m-1}}({\bf z}_{[1,m-1]} )$ is invariant under the permutation 
of ${\bf z}_{[1,m-1]}$.

We can show the following commutation relations,
which are analogous to \eqref{multipleformulaone} and can be proved by
a similar argument.

\begin{proposition}
For $n \geq  i_{1} >i_2> \cdots > i_{m} \geq 0$ and ${\bf z}_{[1,m-1]}=
(
{\bf z}_1,{\bf z}_2,\dots,{\bf z}_{m-1}
)$, we have
\begin{align}
X_{i_1,\dots,i_{m-1}}({\bf z}_{[1,m-1]})
X_{i_m}({\bf z}_m) =\sum_{
(
{\bf w}_{[1,m-1]},{\bf w}_m 
)
}
\frac{1}{
1-{\bf w}_m /  {\bf w}_{[1,m-1]}
}
X_{i_m}({\bf w}_m)
X_{i_1,\dots,i_{m-1}}({\bf w}_{[1,m-1]})
, \label{multipleformulatwo}
\end{align}
where 
$\displaystyle \sum_{
(
{\bf w}_{[1,m-1]}, {\bf w}_m
)
}
$
denotes the sum over all 
$(
{\bf w}_{[1,m-1]}, {\bf w}_m
)$ such that $
{\bf w}_{[1,m-1]}, {\bf w}_m
$
are unordered sets of variables satisfying 
$|{\bf w}_{[1,m-1]}|=|{\bf z}_{[1,m-1]}|$,
$|{\bf w}_m|=|{\bf z}_m|$,
and $ {\bf w}_{[1,m-1]} \cup {\bf w}_m = {\bf z}_{[1,m-1]} \cup {\bf z}_m  $.

\end{proposition}

\begin{proposition}
For $n \geq  i_{1} >i_2> \cdots > i_{m-1} > 0$ and ${\bf z}_{[1,m-1]}=
(
{\bf z}_1,{\bf z}_2,\dots,{\bf z}_{m-1}
)$, we have
\begin{align}
&s_{((i_1-m+1)^{|{\bf z}_1|},(i_2-m+2)^{|{\bf z}_2|},\dots,(i_{m-1}+1)^{|{\bf z}_{m-1}|})}({\bf z})
\nonumber \\
=&\sum_{
(
{\bf w}_{[1,m-1]}, {\bf w}_m
)
}
\frac{1}{
{\bf w}_{[1,m-1]}- {\bf w}_m
}
s_{((i_1-m+1+|{\bf z}_m|)^{|{\bf z}_1|},(i_2-m+2+|{\bf z}_m|)^{|{\bf z}_2|},\dots,(i_{m-1}+1+|{\bf z}_m|)^{|{\bf z}_{m-1}|})}({\bf w}_{[1,m-1]}), \label{GM}
\end{align}
where 
$\displaystyle \sum_{
(
{\bf w}_{[1,m-1]}, {\bf w}_m
)
}
$
denotes the sum over all 
$(
{\bf w}_{[1,m-1]}, {\bf w}_m
)$ such that $
{\bf w}_{[1,m-1]}, {\bf w}_m
$
are unordered sets of variables satisfying 
$|{\bf w}_{[1,m-1]}|=|{\bf z}_{[1,m-1]}|$,
$|{\bf w}_m|=|{\bf z}_m|$,
and $ {\bf w}_{[1,m-1]} \cup {\bf w}_m = {\bf z}_{[1,m-1]} \cup {\bf z}_m  $.
\end{proposition}

\begin{proof}
We set $i_m=0$ and take the vacuum expectation values $\langle \Omega | \cdot |\Omega \rangle$ of both-hand sides of
\eqref{multipleformulatwo}.
Applying \eqref{schurtheorem},
the left-hand side becomes
$s_{((i_1-m+1)^{|{\bf z}_1|},(i_2-m+2)^{|{\bf z}_2|},\dots,(i_{m-1}+1)^{|{\bf z}_{m-1}|})}({\bf z})$,
and the right-hand side becomes
\begin{align}
\sum_{
(
{\bf w}_{[1,m-1]},{\bf w}_m 
)
}
\frac{1}{
1-{\bf w}_m /  {\bf w}_{[1,m-1]}
}
\langle \Omega|
X_{i_m}({\bf w}_m)
X_{i_1,\dots,i_{m-1}}({\bf w}_{[1,m-1]}) |\Omega \rangle
,  \label{gmrhsone}
\end{align}
Using \eqref{actionondualvacuum} and \eqref{schurtheorem},
we get
\begin{align}
\langle \Omega|
X_{0}({\bf w}_m)
&=\langle \Omega|,
\\
\langle \Omega|
X_{i_1,\dots,i_{m-1}}({\bf w}_{[1,m-1]}) |\Omega \rangle
&=\frac{1}{{\bf w}_{[1,m-1]}}
s_{((i_1-m+1)^{|{\bf z}_1|},(i_2-m+2)^{|{\bf z}_2|},\dots,(i_{m-1}+1)^{|{\bf z}_{m-1}|})}({\bf w}_{[1,m-1]}),
\end{align}
and
\eqref{gmrhsone}
becomes
\begin{align}
\sum_{
(
{\bf w}_{[1,m-1]},{\bf w}_m 
)
}
\frac{1}{
1-{\bf w}_m /  {\bf w}_{[1,m-1]}
}
\frac{1}{{\bf w}_{[1,m-1]}}
s_{((i_1-m+1)^{|{\bf z}_1|},(i_2-m+2)^{|{\bf z}_2|},\dots,(i_{m-1}+1)^{|{\bf z}_{m-1}|})}({\bf w}_{[1,m-1]}).
\label{gmrhstwo}
\end{align}
Using
\begin{align}
\frac{1}{
1-{\bf w}_m /  {\bf w}_{[1,m-1]}
}
\frac{1}{{\bf w}_{[1,m-1]}}
=\frac{1}{
{\bf w}^{\prime \prime}
-{\bf w}_m 
}
({\bf w}_{[1,m-1]})^{|{\bf w}_m |-1}
=\frac{1}{
{\bf w}_{[1,m-1]}
-{\bf w}_m 
}
({\bf w}_{[1,m-1]})^{|{\bf z}_m |-1},
\end{align}
and the factorization of Schur polynomials
\begin{align}
&s_{((i_1-m+1+|{\bf z}_m|)^{|{\bf z}_1|},(i_2-m+2+|{\bf z}_m|)^{|{\bf z}_2|},\dots,(i_{m-1}+1+|{\bf z}_m|)^{|{\bf z}_{m-1}|})}({\bf w}_{[1,m-1]}) \nonumber \\
&=
({\bf w}_{[1,m-1]})^{|{\bf z}_m |-1}
s_{((i_1-m+1)^{|{\bf z}_1|},(i_2-m+2)^{|{\bf z}_2|},\dots,(i_{m-1}+1)^{|{\bf z}_{m-1}|})}({\bf w}_{[1,m-1]}),
\end{align}
which can be derived from \eqref{Schurdef}, we find
\eqref{gmrhstwo}
can be further rewritten as the right-hand side of
\eqref{GM}.
\end{proof}

We can also unify the two types of identities.

For $i_{2} >i_3> \cdots > i_{m-1}$ $(m \geq 3)$ and ${\bf z}_{[2,m-1]}=
(
{\bf z}_2,{\bf z}_3,\dots,{\bf z}_{m-1}
)$, introduce
\begin{align}
X_{i_2,\dots,i_{m-1}}({\bf z}_{[2,m-1]})
&:=
\frac{1}{\prod_{k=2}^{m-1} {\bf z}_k^{m-k}}
X_{i_2}({\bf z}_2) X_{i_3}({\bf z}_3)
\cdots X_{i_{m-1}}({\bf z}_{m-1}),
\end{align}
which is invariant under the permutation 
of ${\bf z}_{[2,m-1]}$.

By applying a similar argument to
\eqref{multipleforshuffle},
we obtain the following multiple commutation relations.

\begin{proposition}
For $n \geq  i_{1} >i_2> \cdots > i_{m-1} > i_m \geq 0$ and ${\bf z}_{[2,m-1]}=
(
{\bf z}_2,{\bf z}_3,\dots,{\bf z}_{m-1}
)$, we have
\begin{align}
&
\frac{
X_{i_1}({\bf z}_1)
X_{i_2,\dots,i_{m-1}}({\bf z}_{[2,m-1]})
X_{i_m}({\bf z}_m)
}{
{\bf z}_1^{m-1}
}
\nonumber \\
 =&\sum_{
({\bf w}_1,
{\bf w}_{[2,m-1]},{\bf w}_m 
)
}
\frac{1}{(
1-{\bf w}_{[2,m-1]} /  {\bf w}_1)
(
1-{\bf w}_m /  {\bf w}_1)
(
1-{\bf w}_m /  {\bf w}_{[2,m-1]})
}
\nonumber \\
&\times
\frac{
X_{i_m}({\bf w}_m)
X_{i_2,\dots,i_{m-1}}({\bf w}_{[2,m-1]})
X_{i_1}({\bf w}_1)}
{ {\bf w}_1^{m-1}
}
, \label{multipleformulathree}
\end{align}
where 
$\displaystyle \sum_{
({\bf w}_1,
{\bf w}_{[2,m-1]}, {\bf w}_m
)
}
$
denotes the sum over all 
$({\bf w}_1,
{\bf w}_{[2,m-1]}, {\bf w}_m
)$ such that ${\bf w}_1,
{\bf w}_{[2,m-1]}, {\bf w}_m
$
are unordered sets of variables satisfying $|{\bf w}_1|=|{\bf z}_1|$,
$|{\bf w}_{[2,m-1]}|=|{\bf z}_{[2,m-1]}|$,
$|{\bf w}_m|=|{\bf z}_m|$,
and ${\bf w}_1 \cup {\bf w}_{[2,m-1]} \cup {\bf w}_m ={\bf z}_1 \cup {\bf z}_{[2,m-1]} \cup {\bf z}_m  $.

\end{proposition}
By an argument similar to that used to derive
\eqref{FNR} and \eqref{GM} from
\eqref{schurtheorem} and
\eqref{multipleformulathree},
we obtain the following identities.

\begin{theorem}
For $n >i_2> \cdots > i_{m-1} > 0$ and ${\bf z}_{[2,m-1]}=
(
{\bf z}_2,{\bf z}_3,\dots,{\bf z}_{m-1}
)$, we have
\begin{align}
&s_{((n-m+1)^{|{\bf z}_1|},(i_2-m+2)^{|{\bf z}_2|},\dots,(i_{m-1}+1)^{|{\bf z}_{m-1}|})}({\bf z})
\nonumber \\
 =&\sum_{
({\bf w}_1,
{\bf w}_{[2,m-1]},{\bf w}_m 
)
}
\frac{1}{(
1-{\bf w}_{[2,m-1]}/  {\bf w}_1)
(
1-{\bf w}_m /  {\bf w}_1)
(
1-{\bf w}_m /  {\bf w}_{[2,m-1]})
}
\frac{
 {\bf w}_1^{n+1-m}
 }{
 {\bf w}_{[2,m-1]}
 } \nonumber \\
&\times s_{((i_2-m+3)^{|{\bf z}_2|},(i_3-m+4)^{|{\bf z}_3|},\dots,i_{m-1}^{|{\bf z}_{m-1}|})}({\bf w}_{[2,m-1]}),
\end{align}
where 
$\displaystyle \sum_{
({\bf w}_1,
{\bf w}_{[2,m-1]}, {\bf w}_m
)
}
$
denotes the sum over all 
$({\bf w}_1,
{\bf w}_{[2,m-1]}, {\bf w}_m
)$ such that ${\bf w}_1,
{\bf w}_{[2,m-1]}, {\bf w}_m
$
are unordered sets of variables satisfying $|{\bf w}_1|=|{\bf z}_1|$,
$|{\bf w}_{[2,m-1]}|=|{\bf z}_{[2,m-1]}|$,
$|{\bf w}_m|=|{\bf z}_m|$,
and ${\bf w}_1 \cup {\bf w}_{[2,m-1]} \cup {\bf w}_m ={\bf z}_1 \cup {\bf z}_{[2,m-1]} \cup {\bf z}_m  $.

\end{theorem}

Finally, in this subsection, we present a 
corollary of \eqref{schurtheorem}.

For a Young diagram $\lambda=(\lambda_1,\dots,\lambda_m)$,
the Kostka numbers $K_{\lambda,\alpha}$ are defined as the expansion coefficients of the
Schur polynomial $s_\lambda$ in terms of monomial symmetric functions
\begin{align}
s_\lambda(z_1,\dots,z_m)=\sum_{\alpha=(\alpha_1,\dots,\alpha_m) }
K_{\lambda,\alpha}
\prod_{k=1}^m z_k^{\alpha_k}.
\label{kostkadef}
\end{align}

The following is a three-dimensional realization of the
Kostka numbers, which is a simple corollary of \eqref{schurtheorem}.
\begin{theorem} We have
\begin{align}
\langle \Omega| 
X_{\lambda_1+m-1,\alpha_1+m-1} X_{\lambda_2+m-2,\alpha_2+m-2}
\cdots X_{\lambda_m,\alpha_m}
| \Omega \rangle
=K_{(\lambda_1,\dots,\lambda_m),(\alpha_1,\dots,\alpha_m)}.
\end{align}
\end{theorem}
\begin{proof}
From $X_i(z)=\sum_{j=0}^n z^j X_{i,j}$, 
$|{\bf z}_1|=\cdots=|{\bf z}_m|=1$ case of
\eqref{schurtheorem}
and \eqref{kostkadef}, we have
\begin{align}
&\sum_{j_1,\dots,j_m=0}^n \prod_{k=1}^m z_k^{j_k}
\langle \Omega|
X_{i_1,j_1}X_{i_2,j_2} \cdots X_{i_m,j_m}
|\Omega \rangle \nonumber \\
=&\prod_{k=1}^{m} z_k^{m-k} 
\sum_{\alpha=(\alpha_1,\dots,\alpha_m) }
K_{(i_1-m+1,i_2-m+2,\dots,i_m),\alpha}
\prod_{k=1}^m z_k^{\alpha_k}.
\end{align}
Comparing coefficients of both-hand sides of the above equation, we get
\begin{align}
\langle \Omega|
X_{i_1,\alpha_1+m-1}X_{i_2,\alpha_2+m-2} \cdots X_{i_m,\alpha_m}
|\Omega \rangle =
K_{(i_1-m+1,i_2-m+2,\dots,i_m),(\alpha_1,\dots,\alpha_m)}
.
\end{align}
\end{proof}

\subsection{Application to multispecies TASEP}
In this subsection, we present an application to the multispecies TASEP
by considering partition functions defined through traces,
rather than vacuum expectation values.

We define the trace of $X \in \mathrm{End}(\mathcal{F}^{\otimes n(n-1)/2})$ to be
\begin{align}
\mathrm{Tr}(X):=\sum_{(j,k)} \sum_{m_{j,k} \geq 0}
\Bigg( \bigotimes_{(j,k) \in D_n} \langle m_{j,k}| \Bigg)
X
\Bigg( \bigotimes_{(j,k) \in D_n}
| m_{j,k} \rangle 
\Bigg).
\end{align}
Here we investigate some class of partition functions
in the following form
\begin{align}
\mathrm{Tr}(  X_{i_1}^{(n)}(z_1) X_{i_2}^{(n)}(z_2) \cdots X_{i_m}^{(n)}(z_m)  ),
\end{align}
which is graphically represented as Figure \ref{figuretracepartitionfunction}.

\begin{figure}[htbp]
\centering
\includegraphics[width=12truecm]{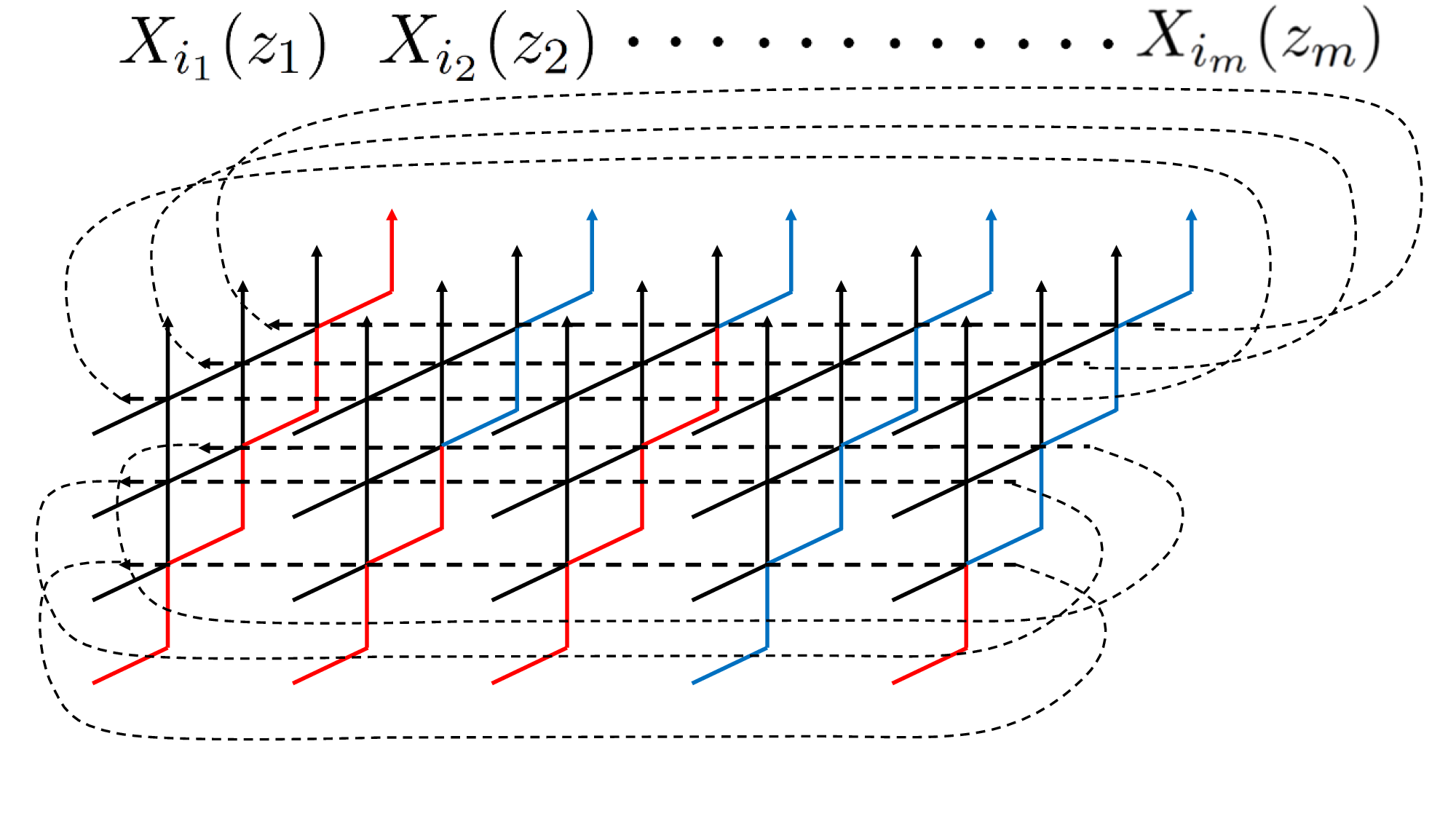}
\caption{The partition functions $\mathrm{Tr}(  X_{i_1}^{(n)}(z_1) X_{i_2}^{(n)}(z_2) \cdots X_{i_m}^{(n)}(z_m)  )$.
}
\label{figuretracepartitionfunction}
\end{figure}


We consider partition functions which at least one $X_j^{(n)}$ is used as components
for all $j=0,\dots,n$.
The following is the simplest type.
\begin{lemma}
We have
\begin{align}
\mathrm{Tr}
(
X_n({\bf z}_n)
X_{n-1}({\bf z}_{n-1} )
\cdots X_0({\bf z}_0)
)
=\prod_{j=0}^n {\bf z}_j^j. \label{tracefactorization}
\end{align}
\end{lemma}
\begin{proof}
We investigate the coloring of the edges of the $L$-operators along the sequence of
the bosonic Fock spaces
\[
\mathcal{F}_{1,n-1} \to \mathcal{F}_{2,n-2} \to \mathcal{F}_{1,n-2}
    \to \mathcal{F}_{3,n-3} \to \mathcal{F}_{2,n-3} \to \cdots
    \to \mathcal{F}_{1,1}.
\]

We start from Figure \ref{figure:pf1},
and first examine the $L$-operators acting on $\mathcal{F}_{1,n-1}$.
We note from the ice-rule
of the $L$-operator $[\mathcal{R}]_{ij}^{ab}=0$ unless $i+j=a+b$ 
that the $L$-operator in the $X_{n-1}$-operator produces the zero-number projection operator,
and one edge is colored by red and the other one is colored by blue
(Figure \ref{figure:pf2}).
The $L$-operator in the $X_{n-1}$-operator can either produce an annihilation operator or the identity operator
in principle, but the existence of the zero-number projection operator in the $X_{n-1}$-operator forces
the  $L$-operator in the $X_{n-1}$-operator to be the identity operator, and the edges are colored by red
(Figure \ref{figure:pf3}).
By a similar consideration, we note each of the $L$-operators in the $X_0,\dots,X_{n-2}$-operators
is forced to be the identity operator, and the edges are colored by blue
(Figure \ref{figure:pf4}).

We then carry out the same observations for the $L$-operators acting on $\mathcal{F}_{2,n-2}$
(Figure \ref{figure:pf5}).
Repeating this argument along the sequence mentioned above, we finally get Figure \ref{figure:pf6}
which is the unique configuration for $\mathrm{Tr}
(
X_n({\bf z}_n)
X_{n-1}({\bf z}_{n-1} )
\cdots X_0({\bf z}_0)
)$.
One notes that, for the $X_j$-operator, there are $j$ red edges on the top boundary,
which gives a power of $j$ for the variable of the $X_j$-operator,
and we conclude
\eqref{tracefactorization}.

%
\end{proof}

\begin{figure}[htbp]
\centering

\begin{minipage}{0.45\textwidth}
\centering
\includegraphics[width=\textwidth]{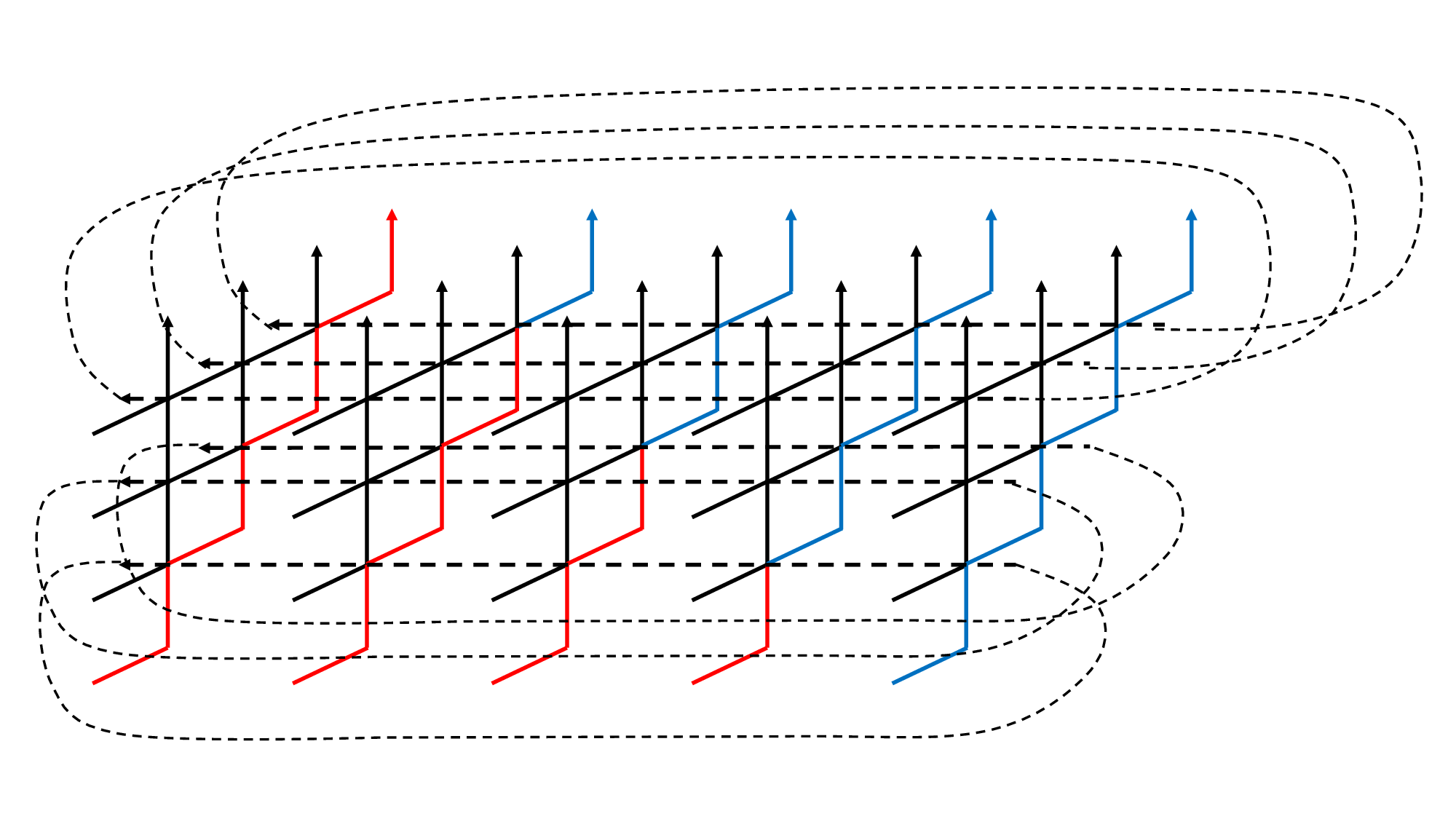}
\caption{The partition function $\mathrm{Tr}
(
X_n({\bf z}_n)
X_{n-1}({\bf z}_{n-1} )
\cdots X_0({\bf z}_0)
)$.}
\label{figure:pf1}
\end{minipage}
\hfill
\begin{minipage}{0.45\textwidth}
\centering
\includegraphics[width=\textwidth]{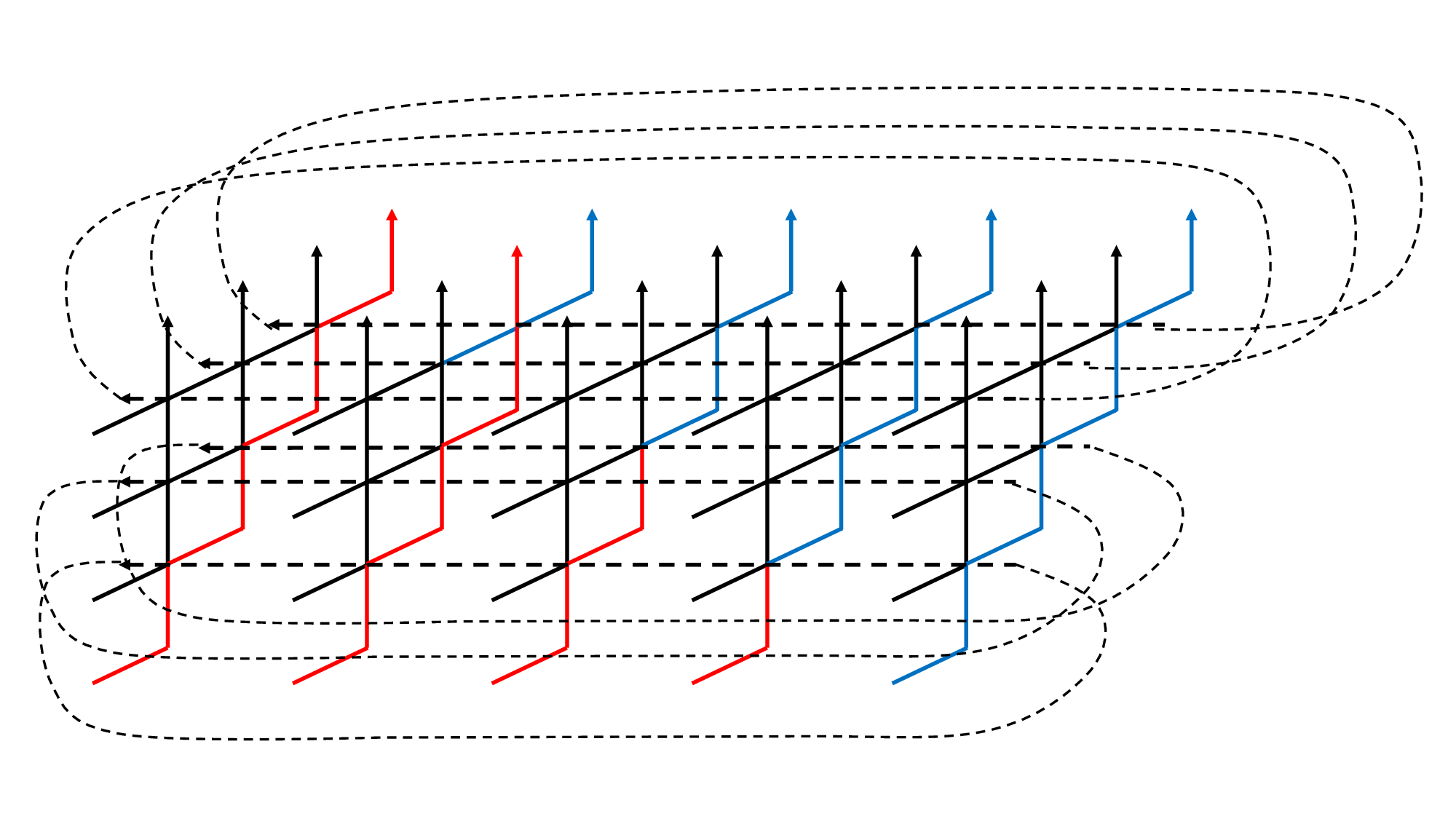}
\caption{The $L$-operator on $\mathcal{F}_{1,n-1}$ in the $X_{n-1}$-operator produces the zero-number
projection operator.}
\label{figure:pf2}
\end{minipage}

\end{figure}

\begin{figure}[htbp]
\centering

\begin{minipage}{0.45\textwidth}
\centering
\includegraphics[width=\textwidth]{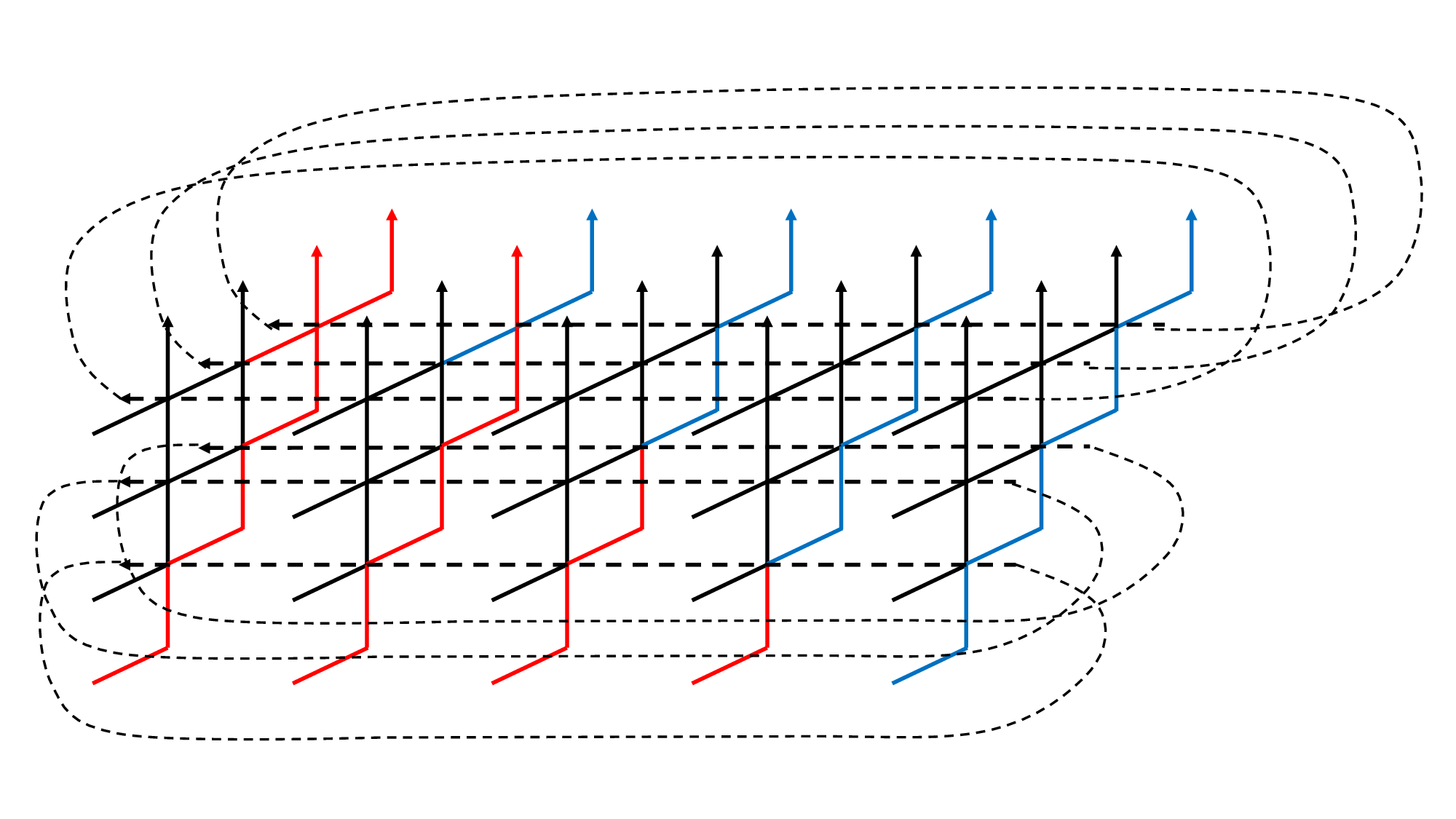}
\caption{
The $L$-operator acting on $\mathcal{F}_{1,n-1}$ in the $X_{n}$-operator cannot produce a annihilation operator
and is determined to be the identity operator, and the edges are colored red.
}
\label{figure:pf3}
\end{minipage}
\hfill
\begin{minipage}{0.45\textwidth}
\centering
\includegraphics[width=\textwidth]{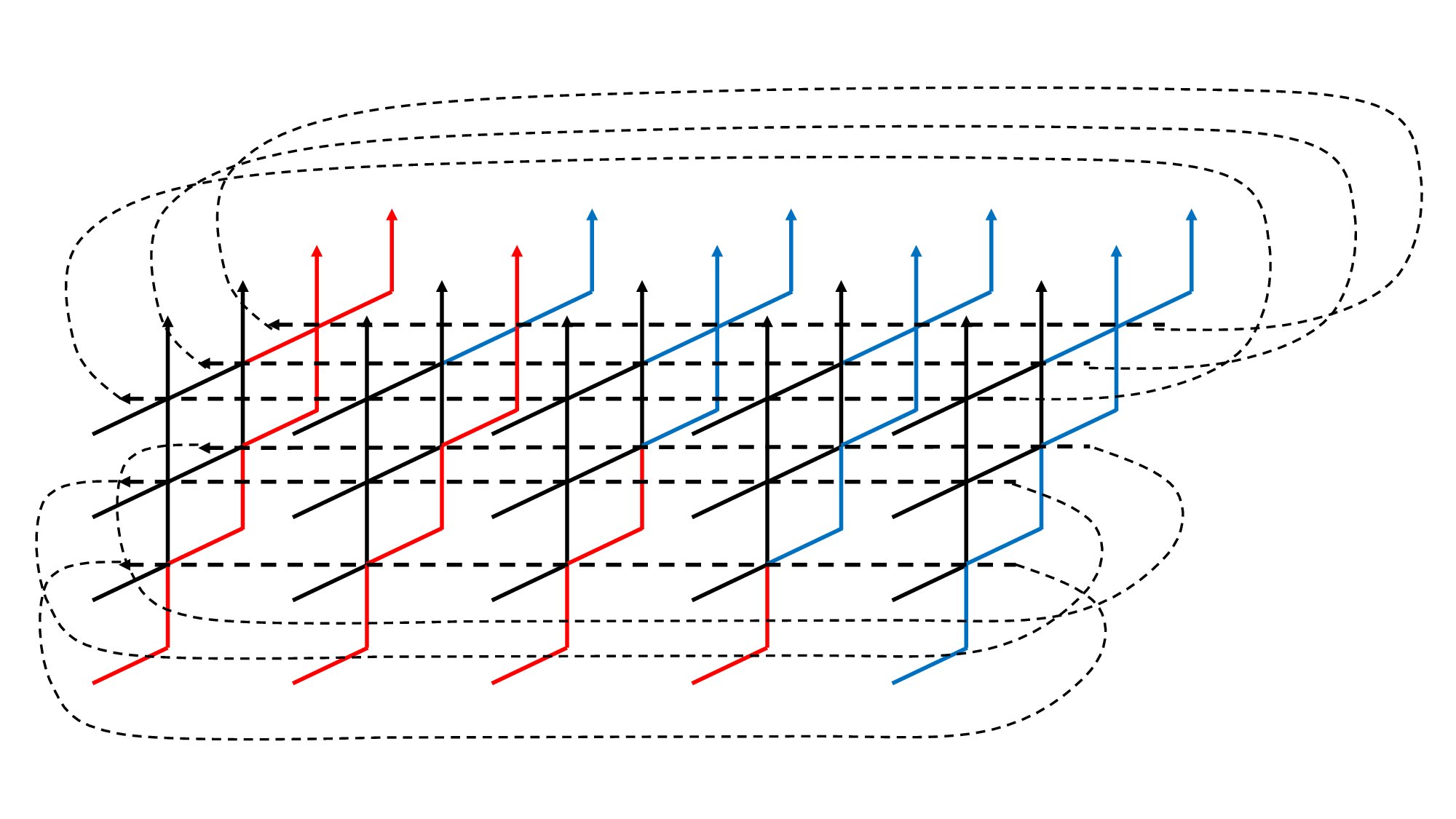}
\caption{
The $L$-operators acting on $\mathcal{F}_{1,n-1}$ in the $X_0,\dots,X_{n-2}$-operators cannot produce creation operators
and each of them is determined to be the identity operator, and the edges are colored blue.
}
\label{figure:pf4}
\end{minipage}

\end{figure}

\begin{figure}[htbp]
\centering

\begin{minipage}{0.45\textwidth}
\centering
\includegraphics[width=\textwidth]{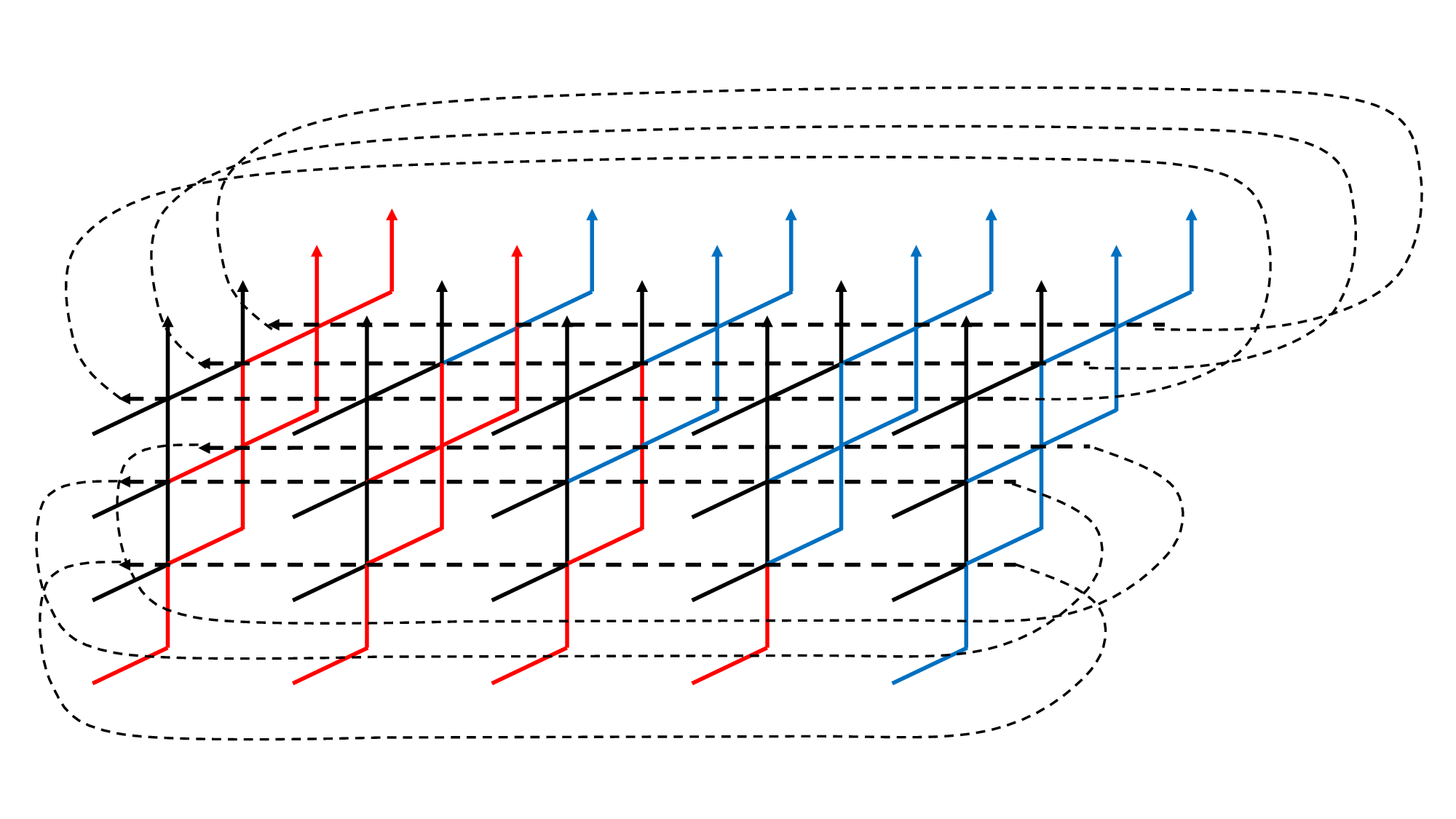}
\caption{The figure after
the edges of the $L$-operators acting on $\mathcal{F}_{2,n-2}$ are colored.
}
\label{figure:pf5}
\end{minipage}
\hfill
\begin{minipage}{0.45\textwidth}
\centering
\includegraphics[width=\textwidth]{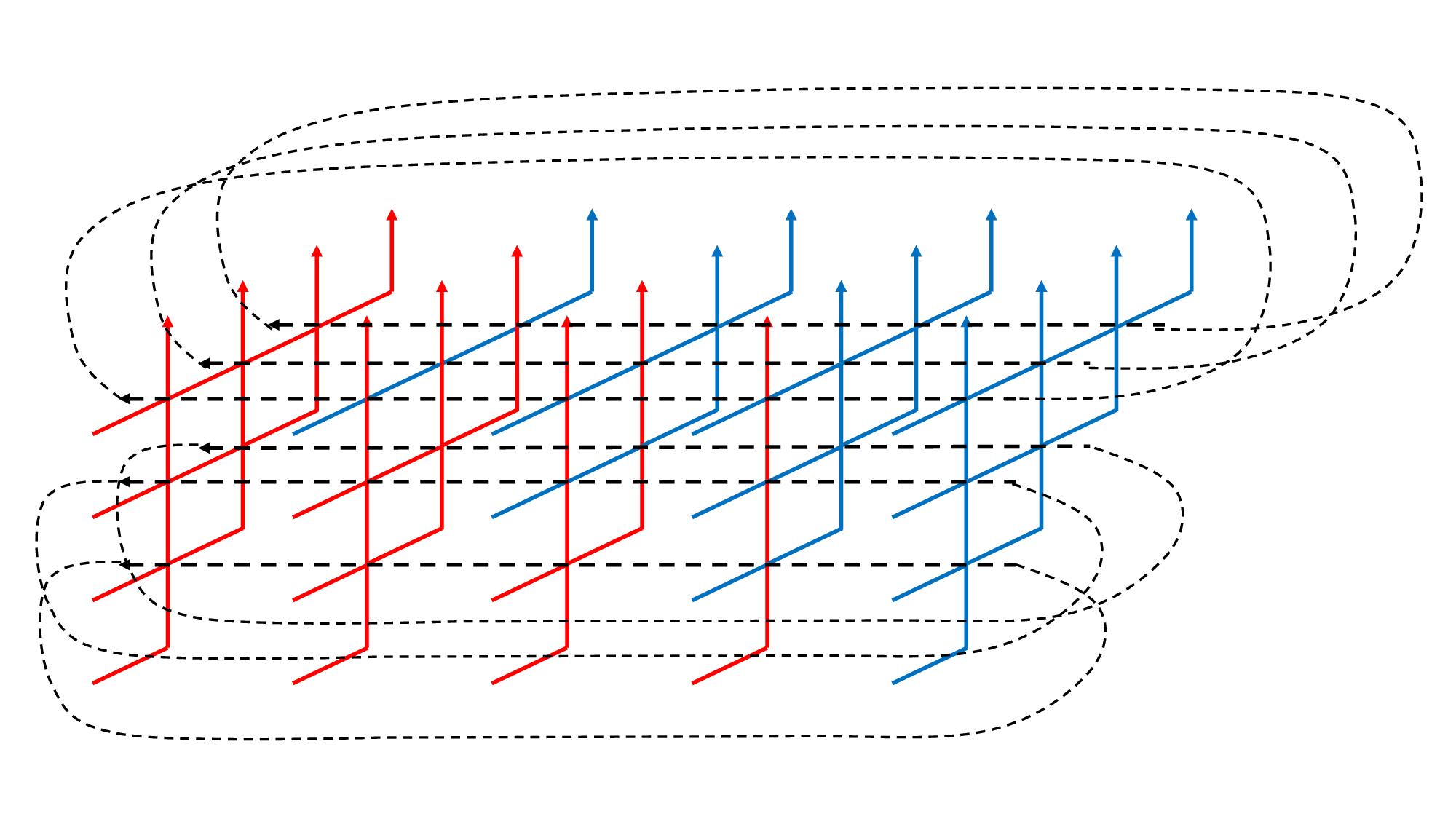}
\caption{The figure after 
the edges of the $L$-operators acting on all bosonic Fock spaces are colored.
This corresponds to the unique configuration.
}
\label{figure:pf6}
\end{minipage}

\end{figure}

Now we recall a theorem of Kuniba-Maruyama-Okado \cite{KMO2}
which gives a description of the steady state of the multispecies TASEP.
We introduce the $n$-multispecies TASEP on a one-dimensional periodic chain
with $L$ sites, denoted by $\mathbb{Z}_L$.
Let $\sigma_i \in \{0,1,\dots,n\}$ $(i=1,\dots,L)$ denote the species (color)
of the particle at site $i$, where color $0$ represents an empty site.
A configuration is written as $\boldsymbol{\sigma}=(\sigma_1,\sigma_2,\dots,\sigma_L)$.
We also introduce the associated vectors
\[
|\boldsymbol{\sigma} \rangle=
|\sigma_1,\sigma_2,\dots,\sigma_L\rangle
:= |\sigma_1\rangle_1 \otimes |\sigma_2\rangle_2 \otimes \cdots \otimes |\sigma_L\rangle_L
\in (\mathbb{C}^{\,n+1})^{\otimes L},
\]
where $|\sigma\rangle$ is the standard basis vector of the $(n+1)$-dimensional
complex vector space.

The dynamics is defined by the following local transition rule.
For each neighbouring pair of local states, if
\[
(\sigma_i,\sigma_{i+1}) = (\alpha,\beta)
\quad \text{with }\alpha>\beta,
\]
then the exchange
\[
\alpha\beta \;\longrightarrow\; \beta\alpha
\]
occurs with the uniform transition rate.

The time evolution of the probability vector
\[
|\mathbb{P}(t)\rangle = \sum_{\boldsymbol{\sigma}}
\mathbb{P}(\boldsymbol{\sigma},t)\,
|\boldsymbol{\sigma}\rangle,
\qquad
\boldsymbol{\sigma}=(\sigma_1,\dots,\sigma_L),
\]
is governed by the master equation
\[
\frac{d}{dt}|\mathbb{P}(t)\rangle = H\,|\mathbb{P}(t)\rangle,
\]
where $H$ is the Markov generator (transition matrix) defined by
\[
H=\sum_{i \in \mathbb{Z}_L} h_{i,i+1}.
\]
The local operator $h_{i,i+1}$ acts nontrivially only on sites $(i,i+1)$ and is
given by
\[
h_{i,i+1}\,|\alpha,\beta\rangle
=
\begin{cases}
|\beta,\alpha\rangle - |\alpha,\beta\rangle, & \alpha>\beta,\\[4pt]
0, & \alpha\le\beta,
\end{cases}
\]
where $|\alpha,\beta\rangle$ denotes the tensor product basis at sites $(i,i+1)$.

Sectors are invariant subspaces of the Markov matrix $H$ which are spanned by
basis vectors $|\boldsymbol{\sigma}\rangle$ with fixed particle
multiplicities.  Namely, for
\[
\boldsymbol{m}=(m_0,m_1,\dots,m_n),
\]
where $m_k$ denotes the number of particles of species~$k$, we define
\[
S(\boldsymbol{m})
=
\Bigl\{
\boldsymbol{\sigma}=(\sigma_1,\sigma_2,\dots,\sigma_L)
\,\Big|\,
\sum_{i=1}^L \delta_{k,\sigma_i}=m_k,\quad k=0,\dots,n
\Bigr\}.
\]
The Markov generator $H$ preserves these sectors, i.e.
\[
H\,|\boldsymbol{\sigma}\rangle \in 
\mathrm{span}\bigl\{\,|\boldsymbol{\tau}\rangle \;\big|\;\boldsymbol{\tau}\in S(\boldsymbol{m})\,\bigr\}
\qquad
\text{for all }\boldsymbol{\sigma}\in S(\boldsymbol{m}).
\]

For a fixed sector $S(\boldsymbol{m})$, the steady state is the vector
\[
|\mathbb{P}\rangle=\sum_{  \boldsymbol{\sigma}  \in S(\boldsymbol{m})}\mathbb{P}(  \boldsymbol{\sigma}  )\,|  \boldsymbol{\sigma}   \rangle
\]
satisfying
\[
H\,|\mathbb{P}\rangle=0.
\]
Here the components $\mathbb{P}( \boldsymbol{\sigma}   )$ are the stationary probabilities of the configuration $| \boldsymbol{\sigma}  \rangle$.

%

\begin{theorem} \cite{KMO2}
The components of the steady state vector $\mathbb{P}(\sigma_1,\sigma_2,\dots,\sigma_m)$,
\\
$0 \leq \sigma_1,\sigma_2,\dots,\sigma_m \leq n$
of $n$-species TASEP with  $m$-sites
are given by
\begin{align}
\mathbb{P}(\sigma_1,\sigma_2,\dots,\sigma_m)=
\mathrm{Tr}
(
X_{\sigma_1}(1)
X_{\sigma_2}(1)
\cdots X_{\sigma_m}(1)
), \label{normzaliationtasep}
\end{align}
under the normalization
\begin{align}
\mathbb{P}(\sigma_1,\sigma_2,\dots,\sigma_m)=1, \label{coeffone}
\end{align}
for $n=\sigma_1 \geq \sigma_2 \geq \cdots \geq \sigma_m=0$.
\end{theorem}
Specializing all the spectral variables in
\eqref{tracefactorization} to one
leads to
\begin{align}
\mathbb{P}(n^{|{\bf z}_n|},(n-1)^{|{\bf z}_{n-1}|},\dots,0^{|{\bf z}_0|})=1,
\end{align}
which corresponds to the normalization \eqref{coeffone}.

We determine several other components using the multiple commutation relations.
We first derive explicit formulas for partition functions involving spectral variables.
\begin{theorem}
For $k \geq j$, we have
\begin{align}
&\mathrm{Tr}
(
X_{n}({\bf z}_n) \cdots X_{k+1}({\bf z}_{k+1})
X_{j-1}({\bf z}_{j-1}) \cdots X_{0}({\bf z}_{0})
X_{k}({\bf z}_k) \cdots X_{j}({\bf z}_{j})
) \nonumber \\
=&\prod_{\ell=k+1}^n {\bf z}_\ell^{\ell+j-k-1} \prod_{\ell=0}^{k} {\bf z}_\ell^\ell
\
s_{(k+1-j)^{|{\bf z}_{[k+1,n]}|}}({\bf z}_{[0,j-1]}, {\bf z}_{[k+1,n]}).
\label{tasepprobgen}
\end{align}
\end{theorem}
\begin{proof}
We can derive the following commutation relations
\begin{align}
&X_{n}({\bf z}_n) \cdots X_{k+1}({\bf z}_{k+1})
X_{j-1}({\bf z}_{j-1}) \cdots X_{0}({\bf z}_{0})
\nonumber \\
=&\prod_{\ell=k+1}^n {\bf z}_\ell^{\ell+j-k-1} \prod_{\ell=0}^{j-1} {\bf z}_\ell^\ell
\sum_{({\bf w}_{[0,j-1]}, {\bf w}_{[k+1,n]})}
\frac{1}{\prod_{\ell=k+1}^n {\bf w}_\ell^{\ell+j-k-1} \prod_{\ell=0}^{j-1} {\bf w}_\ell^\ell}
\frac{1}{1-{\bf w}_{[0,j-1]}/ {\bf w}_{[k+1,n]}}
\nonumber \\
&\times
X_{j-1}({\bf w}_{j-1}) \cdots X_{0}({\bf w}_{0})
X_{n}({\bf w}_n) \cdots X_{k+1}({\bf w}_{k+1}),
\label{multcommfortrace}
\end{align}
in a similar way to derive
\eqref{multiplecommrel}.
Multiplying both-hand sides of \eqref{multcommfortrace} by
$X_{k}({\bf z}_k) \cdots X_{j}({\bf z}_{j})$
and taking trace, we get
\begin{align}
&\mathrm{Tr}( X_{n}({\bf z}_n) \cdots X_{k+1}({\bf z}_{k+1})
X_{j-1}({\bf z}_{j-1}) \cdots X_{0}({\bf z}_{0})
X_{k}({\bf z}_k) \cdots X_{j}({\bf z}_{j}) )
\nonumber \\
=&\prod_{\ell=k+1}^n {\bf z}_\ell^{\ell+j-k-1} \prod_{\ell=0}^{j-1} {\bf z}_\ell^\ell
\sum_{({\bf w}_{[0,j-1]}, {\bf w}_{[k+1,n]})}
\frac{1}{\prod_{\ell=k+1}^n {\bf w}_\ell^{\ell+j-k-1} \prod_{\ell=0}^{j-1} {\bf w}_\ell^\ell}
\frac{1}{1-{\bf w}_{[0,j-1]}/ {\bf w}_{[k+1,n]}}
\nonumber \\
&\times
\mathrm{Tr}(
X_{j-1}({\bf w}_{j-1}) \cdots X_{0}({\bf w}_{0})
X_{n}({\bf w}_n) \cdots X_{k+1}({\bf w}_{k+1})
X_{k}({\bf z}_k) \cdots X_{j}({\bf z}_{j})
).
\label{fordergenprob}
\end{align}
Using the trace property and
\eqref{tracefactorization},
we have
\begin{align}
&
\mathrm{Tr}(
X_{j-1}({\bf w}_{j-1}) \cdots X_{0}({\bf w}_{0})
X_{n}({\bf w}_n) \cdots X_{k+1}({\bf w}_{k+1})
X_{k}({\bf z}_k) \cdots X_{j}({\bf z}_{j})
) \nonumber \\
=&
\mathrm{Tr}(
X_{n}({\bf w}_n) \cdots X_{k+1}({\bf w}_{k+1})
X_{k}({\bf z}_k) \cdots X_{j}({\bf z}_{j})
X_{j-1}({\bf w}_{j-1}) \cdots X_{0}({\bf w}_{0})
) \nonumber \\
=&\prod_{\ell=0}^{j-1} {\bf w}_\ell^\ell
\prod_{\ell=k+1}^{n} {\bf w}_\ell^\ell
\prod_{\ell=j}^{k} {\bf z}_\ell^\ell
.  \label{cyclicfactor}
\end{align}
Using \eqref{cyclicfactor},  \eqref{fordergenprob} can be rewritten as
\begin{align}
&\mathrm{Tr}( X_{n}({\bf z}_n) \cdots X_{k+1}({\bf z}_{k+1})
X_{j-1}({\bf z}_{j-1}) \cdots X_{0}({\bf z}_{0})
X_{k}({\bf z}_k) \cdots X_{j}({\bf z}_{j}) )
\nonumber \\
=&\prod_{\ell=k+1}^n {\bf z}_\ell^{\ell+j-k-1} \prod_{\ell=0}^{k} {\bf z}_\ell^\ell
\sum_{({\bf w}_{[0,j-1]}, {\bf w}_{[k+1,n]})}
\frac{\prod_{\ell=k+1}^n {\bf w}_\ell^{k+1-j}}{1-{\bf w}_{[0,j-1]}/{\bf w}_{[k+1,n]}}
.
\label{beforefinalexpressiontasep}
\end{align}
Finally, applying
\eqref{symmformulaSchur},
we get
\eqref{tasepprobgen}.
\end{proof}
As a corollary of \eqref{tasepprobgen}, by specializing all the spectral variables to one,
we get the following.
\begin{theorem}
For $k \geq j$, we have
\begin{align}
&\mathbb{P}
(
n^{|{\bf z}_n|}, \dots, (k+1)^{|{\bf z}_{k+1}|},
(j-1)^{|{\bf z}_{j-1}|}, \dots, 0^{|{\bf z}_{0}|},
k^{|{\bf z}_k|}, \dots,j^{|{\bf z}_{j}|}
) \nonumber \\
=&
s_{(k+1-j)^{|{\bf z}_{[k+1,n]}|}}(1^{ |{\bf z}_{[0,j-1]}|+ |{\bf z}_{[k+1,n]}|}).
\label{tasepprob}
\end{align}
\end{theorem}

{\bf Example}: One of the equations given in
\cite[Example 2.1.]{KMO2} is
\begin{align}
|\xi(2,1,1,1) \rangle
=&24|00123\rangle+6|00213 \rangle+12|01023 \rangle
+17|01203 \rangle 
+8|02013\rangle+3|02103 \rangle
\nonumber \\
&+4|10023 \rangle
+7|10203 \rangle 
+9|12003\rangle+6|20013 \rangle+3|20103 \rangle
+|21003 \rangle.
\end{align}
This correponds to the steady state vector of three-species TASEP $(n=3)$ in the sector $S(\boldsymbol{m})$,
$\boldsymbol{m}=(2,1,1,1)$.
The coefficient 1 of $|21003 \rangle$ and the cyclic property implies
$\mathbb{P}(21003)=\mathbb{P}(32100)=1$,
which matches with the normalization \eqref{coeffone}.
The coefficient 6 of $|00213 \rangle$ and the cyclic property implies
$\mathbb{P}(00213)=\mathbb{P}(30021)=6$.
The case
$n=3$, $j=1$, $k=2$, $|{\bf z}_0|=2$, $|{\bf z}_1|=|{\bf z}_2|=|{\bf z}_3|=1$ of
\eqref{tasepprob} gives $\mathbb{P}(30021)= s_{2}(1^3)=6$.
As another example,
the coefficient 4 of $|10023 \rangle$ and the cyclic property implies
$\mathbb{P}(10023)=\mathbb{P}(31002)=4$.
The case
$n=3$, $j=2$, $k=2$, $|{\bf z}_0|=2$, $|{\bf z}_1|=|{\bf z}_2|=|{\bf z}_3|=1$ of
\eqref{tasepprob} gives $\mathbb{P}(31002)= s_{1}(1^4)=4$.

\subsection{Modified Schubert polynomials}

In this subsection, we introduce a subclass of partition functions,
identify them as an analogue of Schubert polynomials defined via
divided difference operators, and discuss several of their properties.

We introduce the following version of the divided difference operator
for a function of variables $\mathbf{z}=(z_1,\dots,z_m)$
\begin{align}
\mathcal{D}_i
f
:=\frac{z_{i+1}}{z_i}
\frac{z_i f-z_{i+1}
s_i f
 }{z_i-z_{i+1}}, \ \ \ i=1,\dots,m-1, \label{twdivdiff}
\end{align}
where $s_i$ swaps $z_i$ and $z_{i+1}$.

For $w \in S_m$ where $S_m$ is the symmetric group of order $m$,
we introduce the {\it modified Schubert polynomials}
$\tilde{\mathfrak{S}}_w$
recursively by
\[
\begin{aligned}
\tilde{\mathfrak{S}}_{w_0}
&= z_1^{m-1} z_2^{m-2} \cdots z_{m-1}, \\
\tilde{\mathfrak{S}}_{w}
&= \mathcal{D}_i \, \tilde{\mathfrak{S}}_{w s_i},
\qquad \text{whenever } \ell(ws_i) = \ell(w) + 1,
\end{aligned}
\]
where $w_0$ is the longest permutation in $S_m$.
Denote $w \in S_m$ using one-line notation as $w=[w(1),w(2),\dots,w(m)]$.

\begin{align}
&\widetilde{\mathfrak S}_{[1,2,3]}
= \frac{
z_3^{2}\,(z_2+z_3)\,(z_1+z_2)\,(z_1+z_3)
}{z_1^{2}}, \nonumber \\[6pt]
&\widetilde{\mathfrak S}_{[1,3,2]}
= \frac{z_3\bigl(z_1^3 z_2 + z_1^3 z_3 + z_1^2 z_2^2 + z_1^2 z_2 z_3 + z_1 z_2^3 + z_1 z_2^2 z_3 + z_2^3 z_3\bigr)}{z_1^2}, \nonumber \\[6pt]
&\widetilde{\mathfrak S}_{[2,1,3]}
= \frac{z_3\bigl(z_1 z_2^2 + z_1 z_2 z_3 + z_1 z_3^2 + z_2^3 + z_2^2 z_3 + z_2 z_3^2 + z_3^3\bigr)}{z_2}, \nonumber \\[6pt]
&\widetilde{\mathfrak S}_{[2,3,1]}
= z_2^2\bigl(z_1 + z_2\bigr), \nonumber \\[6pt]
&\widetilde{\mathfrak S}_{[3,1,2]}
= \frac{z_1^2 z_3\bigl(z_2 + z_3\bigr)}{z_2}, \nonumber \\[6pt]
&\widetilde{\mathfrak S}_{[3,2,1]}
= z_1^2 z_2.
\end{align}

Note that the modified Schubert polynomials resembles but differs from
the Schubert polynomials $\mathfrak{S}_w$ 
\cite{LaSc,MacSch}
defined recursively by
\[
\begin{aligned}
\mathfrak{S}_{w_0}
&= z_1^{m-1} z_2^{m-2} \cdots z_{m-1}, \\[2mm]
\mathfrak{S}_{w}
&= \partial_i\, \mathfrak{S}_{w s_i},
\qquad \text{whenever } \ell(ws_i)=\ell(w)+1,
\end{aligned}
\]
which uses the same monomial $z_1^{m-1} z_2^{m-2} \cdots z_{m-1}$ for the longest element $w_0$, but
the type of divided difference operator used is the standard divided difference operator
\begin{align}
\partial_i f
= \frac{f - s_i f}{z_i - z_{i+1}},  \ \ \ i=1,\dots,m-1. \label{orddivdiff}
\end{align}
Note $\tilde{\mathfrak{S}}_w$ is in general a Laurent polynomial
in constrast to $\mathfrak{S}_w$ which is a polynomial for all $w \in S_n$. 
Note also the following relation
\begin{align}
\mathcal{D}_i =\frac{z_{i+1}}{z_i}  \partial_i z_i =\frac{z_{i+1}}{z_i} (1+z_{i+1} \partial_i). \label{relationdivdiff}
\end{align}

\begin{proposition}
For $m \le \frac{n+2}{2}$, we have
\begin{align}
\langle \Omega| X_{2(w(m)-1)}(z_m)
X_{2(w(m-1)-1)}(z_{m-1}) \cdots X_{2(w(1)-1)}(z_1)|\Omega \rangle
=\prod_{k=1}^m z_k^{m-k} \times \tilde{\mathfrak{S}}_w.
\label{relationmockSchubert}
\end{align}
\end{proposition}
\begin{proof}
For $w_0=[m,m-1,\dots,1]$, $w(i)=m+1-i$, $i=1,\dots,m$ and
since
\begin{align}
\langle \Omega| X_{0}(z_m)
X_{2}(z_{m-1}) \cdots X_{2(m-1)}(z_1)|\Omega \rangle
=\prod_{k=1}^m z_k^{2(m-k)},
\end{align}
is a special case of Lemma 3.2, we conclude
\eqref{relationmockSchubert} for $w_0$ is satisfied.
Suppose \eqref{relationmockSchubert} for $w$ replaced by $w s_i$ such that
$\ell(ws_i) = \ell(w) + 1$ is satisfied,
i.e. we assume
\begin{align}
&\langle \Omega| X_{2(w(m)-1)}(z_m)
 \cdots
 X_{2(w(i)-1)}(z_{i+1})X_{2(w(i+1)-1)}(z_i)
 \cdots
  X_{2(w(1)-1)}(z_1)|\Omega \rangle \nonumber \\
=&\prod_{k=1}^m z_k^{m-k} \times \tilde{\mathfrak{S}}_{ws_i}.
\label{relationmockSchubertforinduction}
\end{align}
We rewrite \eqref{relationmockSchubertforinduction} as
\begin{align}
\tilde{\mathfrak{S}}_{ws_i}
=\prod_{k=1}^m z_k^{k-m}
\langle \Omega| X_{2(w(m)-1)}(z_m)
 \cdots
 X_{2(w(i)-1)}(z_{i+1})X_{2(w(i+1)-1)}(z_i)
 \cdots
  X_{2(w(1)-1)}(z_1)|\Omega \rangle,
\end{align}
from which we also get
\begin{align}
s_i \tilde{\mathfrak{S}}_{ws_i}
=&z_i^{i+1-m} z_{i+1}^{i-m}
\prod_{\substack{ k=1 \\ k \neq i,i+1}}^m z_k^{k-m} \nonumber \\
&\times
\langle \Omega| X_{2(w(m)-1)}(z_m) 
 \cdots
 X_{2(w(i)-1)}(z_{i})X_{2(w(i+1)-1)}(z_{i+1})
 \cdots
  X_{2(w(1)-1)}(z_1)|\Omega \rangle.
\end{align}
We use these two equalities to compute $\mathcal{D}_i \, \tilde{\mathfrak{S}}_{w s_i}$
\begin{align}
\mathcal{D}_i \, \tilde{\mathfrak{S}}_{w s_i}=&
\frac{z_{i+1}}{z_i} \frac{
z_i \tilde{\mathfrak{S}}_{w s_i}-z_{i+1} s_i \tilde{\mathfrak{S}}_{w s_i}
}{z_i-z_{i+1}} \nonumber \\
=&
z_i^{i-1-m} z_{i+1}^{i+2-m}
\prod_{\substack{ k=1 \\ k \neq i,i+1}}^m z_k^{k-m}
\langle \Omega| X_{2(w(m)-1)}(z_m)
\cdots \nonumber \\
&\times
\frac{z_i}{z_i-z_{i+1}}
(
X_{2(w(i)-1)}(z_{i+1})X_{2(w(i+1)-1)}(z_i)-X_{2(w(i)-1)}(z_{i})X_{2(w(i+1)-1)}(z_{i+1})
) \nonumber \\
&\times \cdots
  X_{2(w(1)-1)}(z_1)|\Omega \rangle.
  \label{tochurecurrence}
\end{align}
Recall the Zamolodchikov-Faddeev relation \eqref{FZalg} rewritten in the form
\begin{align}
X_i(x)X_j(y)=(1-y/x)^{-1}X_j(y)X_i(x)-(1-y/x)^{-1}X_j(x)X_i(y), \ \ \ i>j. \label{FZalgrewrite}
\end{align}
Since $\ell(ws_i) = \ell(w) + 1$
is equivalent to $w(i)<w(i+1)$,
one notes that we can apply \eqref{FZalgrewrite} to
$
\frac{z_i}{z_i-z_{i+1}}
(
X_{2(w(i)-1)}(z_{i+1})X_{2(w(i+1)-1)}(z_i)-X_{2(w(i)-1)}(z_{i})X_{2(w(i+1)-1)}(z_{i+1})
)
$ in
 \eqref{tochurecurrence} to get the following simplification
 \begin{align}
\mathcal{D}_i \, \tilde{\mathfrak{S}}_{w s_i}=&
z_i^{i-1-m} z_{i+1}^{i+2-m}
\prod_{\substack{ k=1 \\ k \neq i,i+1}}^m z_k^{k-m}
\langle \Omega| X_{2(w(m)-1)}(z_m)
\cdots \nonumber \\
&\times
X_{2(w(i+1)-1)}(z_i)
X_{2(w(i)-1)}(z_{i+1})
 \cdots
  X_{2(w(1)-1)}(z_1)|\Omega \rangle.
  \label{tochurecurrencetwo}
\end{align}
Finally, we apply
\begin{align}
X_i(x)X_j(y)=x/y X_i(y)X_j(x), \ \ \ i>j,
\end{align}
to further rewrite the expression as
\begin{align}
\mathcal{D}_i \, \tilde{\mathfrak{S}}_{w s_i}
=\prod_{k=1}^m z_k^{k-m}
\langle \Omega| X_{2(w(m)-1)}(z_m)
X_{2(w(m-1)-1)}(z_{m-1}) \cdots X_{2(w(1)-1)}(z_1)|\Omega \rangle.
\end{align}
Together with the recursive defintion of $\tilde{\mathfrak{S}}_{w}$,
we get
\begin{align}
\tilde{\mathfrak{S}}_{w}
=\prod_{k=1}^m z_k^{k-m}
\langle \Omega| X_{2(w(m)-1)}(z_m)
X_{2(w(m-1)-1)}(z_{m-1}) \cdots X_{2(w(1)-1)}(z_1)|\Omega \rangle,
\end{align}
which completes the induction.
\end{proof}

We list some basic properties of the divided difference operators $\mathcal{D}_i $.
\begin{lemma}
We have
\begin{align}
&\mathrm{(i)} \ \mathcal{D}_i \mathcal{D}_{i+1} \mathcal{D}_i=\mathcal{D}_{i+1} \mathcal{D}_i \mathcal{D}_{i+1}, \\
&\mathrm{(ii)} \ \mathcal{D}_{i}^2=-\mathcal{D}_i,
\end{align}
\end{lemma}
\begin{proof}
One can check by acting on functions.
For $\mathrm{(i)}$, it is enough to check $i=1$ acting on a function of three variables.
We can check
\begin{align}
 &\mathcal{D}_1 \mathcal{D}_{2} \mathcal{D}_1 f(z_1,z_2,z_3)
 =\mathcal{D}_2 \mathcal{D}_{1} \mathcal{D}_2 f(z_1,z_2,z_3)
 \nonumber \\
 =&\frac{z_3^2}{z_1^2} \,
\frac{1}{(z_1 - z_2)(z_1 - z_3)(z_2 - z_3)}
\bigl(
 z_1^2 z_2\, f(z_1, z_2, z_3)
 - z_1 z_2^2\, f(z_2, z_1, z_3)
 - z_1^2 z_3\, f(z_1, z_3, z_2)
 \nonumber \\
 &
 + z_2^2 z_3\, f(z_2, z_3, z_1)
 + z_1 z_3^2\, f(z_3, z_1, z_2)
 - z_2 z_3^2\, f(z_3, z_2, z_1)
\bigr).
\end{align}
For $\mathrm{(ii)}$, it is enough to check $i=1$ acting on a function of two variables.
We can check
\begin{align}
\mathcal{D}_1 f(z_1,z_2)=-\mathcal{D}_1^2 f(z_1,z_2)
=\frac{z_2 (z_1 f(z_1,z_2)-z_2 f(z_2,z_1) ) }{z_1 (z_1-z_2)}.
\end{align}

\end{proof}

\begin{lemma}
For two functions $f$ and $g$, we have
\begin{align}
\mathcal{D}_i(fg)=z_{i+1} (\partial_i f) \cdot g
+(s_i f) \cdot (\mathcal{D}_i g).
\end{align}
\end{lemma}
\begin{proof}
$z_{i+1} (\partial_i f) \cdot g
+(s_i f) \cdot (\mathcal{D}_i g)$ is explictly
\begin{align}
z_{i+1} \times \frac{f-s_i f}{z_i-z_{i+1}} \times g+s_i f \times \frac{z_{i+1}}{z_i} \times \frac{ z_i g-z_{i+1} s_i g }{z_i-z_{i+1}},
\end{align}
which can be simplified as
\begin{align}
\frac{z_{i+1}}{z_i(z_i-z_{i+1})} (z_i fg-z_{i+1} (s_i f)(s_i g)),
\end{align}
and can be further rewritten using $(s_i f)(s_i g)=s_i ( fg)$ as 
\begin{align}
\frac{z_{i+1}}{z_i(z_i-z_{i+1})} (z_i fg-z_{i+1} s_i( f g)),
\end{align}
which corresponds to the explicit form of $\mathcal{D}_i(fg)$.
\end{proof}

\begin{lemma}
We have
\begin{align}
&
\mathcal{D}_i \mathcal{D}_{i+1} \mathcal{D}_i \nonumber \\
=&
\frac{z_{i+2}^2}{z_i^2}
\partial_i \partial_{i+1} \partial_i
z_i^2 z_{i+1} \nonumber \\
=&\frac{z_{i+2}^2}{z_i^2}
(
1+(z_{i+1}+z_{i+2}) \partial_i+z_{i+2} \partial_{i+1}
+z_{i+1} z_{i+2} \partial_i \partial_{i+1}
+z_{i+2}^2 \partial_{i+1} \partial_i+z_{i+1} z_{i+2}^2 \partial_i \partial_{i+1} \partial_i
). \label{relopprod}
\end{align}
\end{lemma}
\begin{proof}
This can be obtained by using the relation between $\mathcal{D}_i$ and $\partial_i $
\eqref{relationdivdiff} to rewrite $\mathcal{D}_i \mathcal{D}_{i+1} \mathcal{D}_i$
in terms of $\partial_i $, $\partial_{i+1}$ to get the expression in the middle, and moving the divided difference operators 
to the right by using the exchange relations gives \eqref{relopprod}.
\end{proof}

We fix a sequence $I_k=(i_1,i_2,\dots,i_k)$
such that $s_{i_1} s_{i_2} \cdots s_{i_k}$
is a reduced expression, and
define $\mathcal{D}_{I_k}:=\mathcal{D}_{i_k} \cdots \mathcal{D}_{i_2} \mathcal{D}_{i_1}$
and $\mathcal{\partial}_{I_k}:=\partial_{i_k} \cdots \partial_{i_2} \partial_{i_1}$.
$\mathcal{D}_{I_k}$ can be expanded in the form
\begin{align}
\mathcal{D}_{I_k}=\sum_{J_\ell} c(I_k,J_\ell) \partial_{J_\ell},
\label{relprodopingen}
\end{align}
where the sum is taken over
 subsequences $J_\ell=(j_1,j_2,\dots,j_\ell)$ $(\ell \leq k)$
of $I_k$ such that $s_{j_1} s_{j_2} \cdots s_{j_\ell}$
are reduced expressions. For $J=\varnothing$,
we define
$\partial_\varnothing:=1$ and regard $c(I,\varnothing)$ as the constant term of the expansion.

For $I_{k+1}=(i_1,i_2,\dots,i_k,i_{k+1})$,
denote the expansion as
\begin{align}
\mathcal{D}_{I_{k+1}}=\sum_{J_{\ell+1}} c(I_{k+1},J_{\ell+1}) \partial_{J_{\ell+1}},
\label{expansiononelarger}
\end{align}
where $J_{\ell+1}=(j_1,j_2,\dots,j_\ell,j_{\ell+1})$.

\begin{proposition}
The expansion coefficients
satisfy the following recursion relations
\begin{align}
c(I_{k+1},J_{\ell+1})
&=\Bigg(
\frac{z_{i_{k+1}+1}}{z_{i_{k+1}}}
+
\frac{z_{i_{k+1}+1}^2}{z_{i_{k+1}}}
\partial_{i_{k+1}} \Bigg)
c(I_{k},J_{\ell+1})
+
\delta_{i_{k+1},j_{\ell+1}}
\frac{z_{i_{k+1}+1}^2}{z_{i_{k+1}}}
s_{i_{k+1}} c(I_{k},J_{\ell}),
\\
c(I_{k+1},\varnothing)
&=\Bigg(
\frac{z_{i_{k+1}+1}}{z_{i_{k+1}}}
+
\frac{z_{i_{k+1}+1}^2}{z_{i_{k+1}}}
\partial_{i_{k+1}} \Bigg)
c(I_{k},\varnothing),
\end{align}
where $\delta_{ij}$ is the Kronecker delta:
$\delta_{ij}=1$ for $i=j$ and $\delta_{ij}=0$ otherwise.
\end{proposition}
\begin{proof}
This follows from comparing \eqref{expansiononelarger} and
\begin{align}
\mathcal{D}_{I_{k+1}}&=\mathcal{D}_{i_{k+1}}
\sum_{J_\ell} c(I_k,J_\ell) \partial_{J_\ell}
=\Bigg(
\frac{z_{i_{k+1}+1}}{z_{i_{k+1}}}
+
\frac{z_{i_{k+1}+1}^2}{z_{i_{k+1}}}
\partial_{i_{k+1}} \Bigg)
\sum_{J_\ell} c(I_k,J_\ell) \partial_{J_\ell} \nonumber \\
&=\sum_{J_\ell}
\Bigg\{
\Bigg(
\frac{z_{i_{k+1}+1}}{z_{i_{k+1}}}
+
\frac{z_{i_{k+1}+1}^2}{z_{i_{k+1}}}
\partial_{i_{k+1}} \Bigg)
c(I_k,J_\ell) 
\Bigg\}
\partial_{J_\ell} \nonumber \\
&+\sum_{J_\ell}
\frac{z_{i_{k+1}+1}}{z_{i_{k+1}}}
\{ s_{i_{k+1}} c(I_k,J_\ell) \} \partial_{i_{k+1}}
\partial_{J_\ell} \nonumber \\
&=\sum_{J_{\ell+1}}
\Bigg\{
\Bigg(
\frac{z_{i_{k+1}+1}}{z_{i_{k+1}}}
+
\frac{z_{i_{k+1}+1}^2}{z_{i_{k+1}}}
\partial_{i_{k+1}} \Bigg)
c(I_{k},J_{\ell+1})
+
\delta_{i_{k+1},j_{\ell+1}}
\frac{z_{i_{k+1}+1}^2}{z_{i_{k+1}}}
s_{i_{k+1}} c(I_{k},J_{\ell})
\Bigg\}
\partial_{J_{\ell+1}}
. \label{comparisononelarger}
\end{align}
When transforming to the last equality in
\eqref{comparisononelarger}, we use 
$\partial_i^2=0$ and
that $s_{J_{\ell+1}}=s_{j_1} s_{j_2} \cdots s_{j_{\ell+1}}$ are reduced expressions
so that the operators of longer length than $\ell+1$ such as
 $\partial_{j_{\ell+2}}
\cdots \partial_{j_1} $ cannnot be reduced to shorter length operators
after adding any single divided difference operator.
We also use $s_{J_\ell} s_{j_{\ell+1}} \neq s_{J_\ell}^{\prime} s_{j_{\ell+1}}$
for two different Weyl group elements  of the same length
$s_{J_\ell}=s_{j_1} s_{j_2} \cdots s_{j_\ell}$ and
$s_{J_\ell}^{\prime}=s_{j_1}^\prime s_{j_2}^\prime \cdots s_{j_\ell}^\prime
$, so that $\partial_{J_{\ell+1}}
=\partial_{i_{k+1}}
\partial_{J_\ell}
$ if and only if $i_{k+1}=j_{\ell+1}$
(and the length of $s_{J_\ell}$ and $s_{J_{\ell+1}}$ is $\ell$ and $\ell+1$
respectively, which are assumed).
\end{proof}
\eqref{relprodopingen}, \eqref{comparisononelarger} implies that the modified Schubert polynomials
can be expressed as a linear combination of the Schubert polynomials. To determine
the explicit forms is an open problem.

The following is an easy consequence of $\eqref{relationmockSchubert}$.
\begin{proposition} \label{nonnegativity}
$\tilde{\mathfrak{S}}_w$ are Laurent polynomials in $z_1,\dots,z_m$ with nonnegative coefficients
for arbitrary $w$.
\end{proposition}
\begin{proof}
Fix an arbitrary $w$.
The left-hand side of \eqref{relationmockSchubert} is a partition function
which is a polynomial in $z_1,\dots,z_m$ with nonnegative coefficients,
since all the matrix elements of the $X_j$-operators
constructing the partition function are polynomials in $z_1,\dots,z_m$ with nonnegative coefficients.
Combining this fact with the relation~\eqref{relationmockSchubert},
we obtain the statement of the proposition.
\end{proof}
{\bf Problem}: Prove Proposition \ref{nonnegativity} without using partition functions
and give a combinatorial construction of the coefficients. \\

Finally, let us give some remark and present a conjecture.
Define $R_i(u)$ as
$R_i(u)=1+(1-\mathrm{e}^u) \mathcal{D}_i$.

The following is a special case of
\cite[Lemma 2.2]{FK}, applying the fact $\mathcal{D}_i^2=-\mathcal{D}_i$.
\begin{lemma}
We have
\begin{align}
R_i(u) R_{i+1}(u+v) R_i(v)=
R_{i+1}(v) R_i (u+v) R_{i+1}(u).
\end{align}
\end{lemma}

However,
it seems to be not easy to apply Fomin-Kirillov formalism \cite{FK} which worked
for the Schubert polynomials
and use the Yang-Baxter equation effectively to get a nice compact formula
for the modified Schubert polynomials.

At the end of this subsection, we present the following conjecture.
\begin{conjecture}
The following holds.
\begin{align}
&\mathcal{D}_i (\mathcal{D}_{i+1} \mathcal{D}_i) \cdots (\mathcal{D}_{i+n-1} \mathcal{D}_{i+n-2} \cdots \mathcal{D}_i)
\nonumber \\
=&\prod_{j=i+1}^{n+i} \frac{z_j^{j-i}}{z_{n+2-j}^{j-i}}
\partial_i (\partial_{i+1} \partial_i) \cdots (\partial_{i+n-1} \partial_{i+n-2} \cdots \partial_i)
\prod_{j=i+1}^{n+i} z_{n+2-j}^{j-i}.
\end{align}
\end{conjecture}

\section{On a $q$-deformation of partition functions}
In this section, we investigate certain subclasses of $q$-deformed partition
functions, as well as several other types of partition functions, and determine
their explicit forms. In our analysis, we use the original $L$-operator introduced
by Bazhanov--Sergeev \cite{BaSe}, which involves the quantum parameter $q$.
We also present a conjecture on the commutativity.

Before entering this subsection, we collect several basic notions from 
$q$-analysis that will be used repeatedly below.
\cite{GR} is a standard reference on $q$-series.
Let $q$ be a generic complex number satisfying $|q|<1$.
We begin with the $q$-integer and its factorial $(n \in \mathbb{Z}_{\geq 0})$:
\[
[n]_q=\frac{1-q^n}{1-q}, \qquad
[n]_q! = [1]_q [2]_q \cdots [n]_q .
\]

From these, the $q$-binomial coefficients $(n,k \in \mathbb{Z}_{\geq 0}, n \geq k)$ are defined by
\[
\begin{bmatrix} n \\ k \end{bmatrix}_q
  =\frac{[n]_q!}{[k]_q!\,[n-k]_q!}
  =\frac{(q;q)_n}{(q;q)_k (q;q)_{n-k}},
\]
where we have introduced the $q$-Pochhammer symbol
\[
(a;q)_n=(1-a)(1-aq)\cdots(1-aq^{n-1}), \qquad
(a;q)_0=1,
\]
and its infinite version
\[
(a;q)_\infty=\prod_{r=0}^{\infty}(1-aq^r).
\]

The $q$-binomial theorem is
\begin{align}
\sum_{k=0}^{n} q^{k(k-1)/2}
  \begin{bmatrix} n \\ k \end{bmatrix}_q
  (-z)^k = (z;q)_n. \label{qbinomial}
\end{align}
The following summation formula is well known:
\begin{align}
\sum_{k=0}^{\infty}\frac{(a;q)_k}{(q;q)_k} z^k
  =\frac{(az;q)_\infty}{(z;q)_\infty}. \label{qsummation}
\end{align}
Two particular limits of \eqref{qsummation} are
\begin{align}
\sum_{n=0}^{\infty}\frac{z^n}{(q;q)_n}
      &=\frac{1}{(z;q)_\infty}, \label{qsummationone} \\
\sum_{n=0}^{\infty}\frac{q^{n(n-1)/2}z^n}{(q;q)_n}
      &=(-z;q)_\infty. \label{qsummationtwo}
\end{align}

\subsection{$q$-$L$-operator and a conjecture on commutativity}

Let $V=\mathbb{C}v_0\oplus \mathbb{C}v_1$ be the two-dimensional vector space 
with standard basis $\{v_0,v_1\}$.
Its dual space is 
\[
V^\ast=\mathbb{C}v_0^\ast \oplus \mathbb{C}v_1^\ast,
\qquad
v_j^\ast(v_i)=\delta_{ij}.
\]

We also consider the $q$-bosonic Fock space
\[
\mathfrak{F}=\bigoplus_{m=0}^\infty \mathbb{C}(q)\,|m\rangle,
\]
on which the $q$-oscillator operators 
$\mathbf{k},\mathbf{a}^+,\mathbf{a}^-$ act by
\[
\mathbf{k}|m\rangle=q^m|m\rangle,\qquad
\mathbf{a}^+|m\rangle=|m+1\rangle,\qquad
\mathbf{a}^-|m\rangle=(1-q^{2m})|m-1\rangle,
\]
where $|-1\rangle:=0$.
They satisfy the relations
\begin{align}
\mathbf{k}\mathbf{a}^+ = q\,\mathbf{a}^+\mathbf{k},\qquad
\mathbf{k}\mathbf{a}^- = q^{-1}\mathbf{a}^-\mathbf{k},\qquad
\mathbf{a}^-\mathbf{a}^+ = 1-q^2\mathbf{k}^2,\\
\mathbf{a}^+\mathbf{a}^- = 1-\mathbf{k}^2,\qquad
\mathbf{k}\mathbf{k}^{-1}=\mathbf{k}^{-1}\mathbf{k}=1.
\end{align}

We now introduce the dual Fock space
\[
\mathfrak{F}^\ast=\bigoplus_{m=0}^\infty \mathbb{C}(q)\,\langle m|,
\]
so that the operators act by
\[
\langle m|\mathbf{k}=q^m\langle m|,
\qquad
\langle m|\mathbf{a}^+=(1-q^{2m})\,\langle m-1|,
\qquad
\langle m|\mathbf{a}^-=\langle m+1|.
\]
The natural pairing is
\[
\langle m|m'\rangle=(q^2)_m\,\delta_{m,m'},
\qquad
(q^2)_m=\prod_{r=0}^{m-1}(1-q^{2r}).
\]

It is also convenient to introduce the normalized dual vectors
\[
\langle\!\langle m|:=\frac{1}{(q^2)_m}\langle m|,
\]
for which 
\[
\langle\!\langle m|m'\rangle=\delta_{m,m'}.
\]

We now define the operator-valued $L$-operator
\cite{BaSe,KMO2}
acting on the tensor product
\(V\otimes V\) with values in \(\mathrm{End}(\mathfrak{F})\).
For basis vectors \(v_i\otimes v_j\in V\otimes V\), we set
\[
\mathcal{L}(z)(v_i\otimes v_j)
=
\sum_{a=0}^1\sum_{b=0}^1
v_a\otimes v_b\;
[\mathcal{L}(z)]_{ij}^{ab},
\]
where the operator-valued matrix elements are
\begin{align}
[\mathcal{L}(z)]_{00}^{00} &= 1, \qquad
[\mathcal{L}(z)]_{11}^{11} = 1, \qquad
[\mathcal{L}(z)]_{10}^{01} = z\,\mathbf{a}^+, \\
[\mathcal{L}(z)]_{01}^{10} &= z^{-1}\mathbf{a}^-, \qquad
[\mathcal{L}(z)]_{01}^{01} = \mathbf{k}, \qquad
[\mathcal{L}(z)]_{10}^{10} = -q\,\mathbf{k},
\end{align}
and all other $[\mathcal{L}(z)]_{ij}^{ab}$ vanish.

Note that, under the specialization $q=0$ and $z=1$,
the $L$-operator $\mathcal{L}(z)$ for the operator-valued six-vertex model
reduces to the operator-valued five-vertex model $\mathcal{R}$
introduced in the previous section.

We denote by $\mathfrak{F}_{k\ell}$ the bosonic Fock space associated with
$(k,\ell)\in D_n$, and by
$\mathbf{a}^+_{k\ell}$, $\mathbf{a}^-_{k\ell}$, and $\mathbf{k}_{k\ell}$
the bosonic operators acting on $\mathfrak{F}_{k\ell}$.
In this section, we define the linear operator $X_i^{(n)}(z)$ (Figure $\ref{XIoperatorfigureqgeneric}$) acting on
$\bigotimes_{(k,\ell)\in D_n}\mathfrak{F}_{k\ell}\otimes\mathbb{C}[z]$
as follows.
We replace $\mathcal{R}$ by $\mathcal{L}(z)$ and take the sum
over all configurations satisfying the same conditions as in the previous section:
(i) the bottom edges of columns $1,\dots,i$ are colored red;
(ii) those of columns $i+1,\dots,n$ are colored blue.

\begin{figure}[htbp]
\centering
    \includegraphics[width=0.8\textwidth]{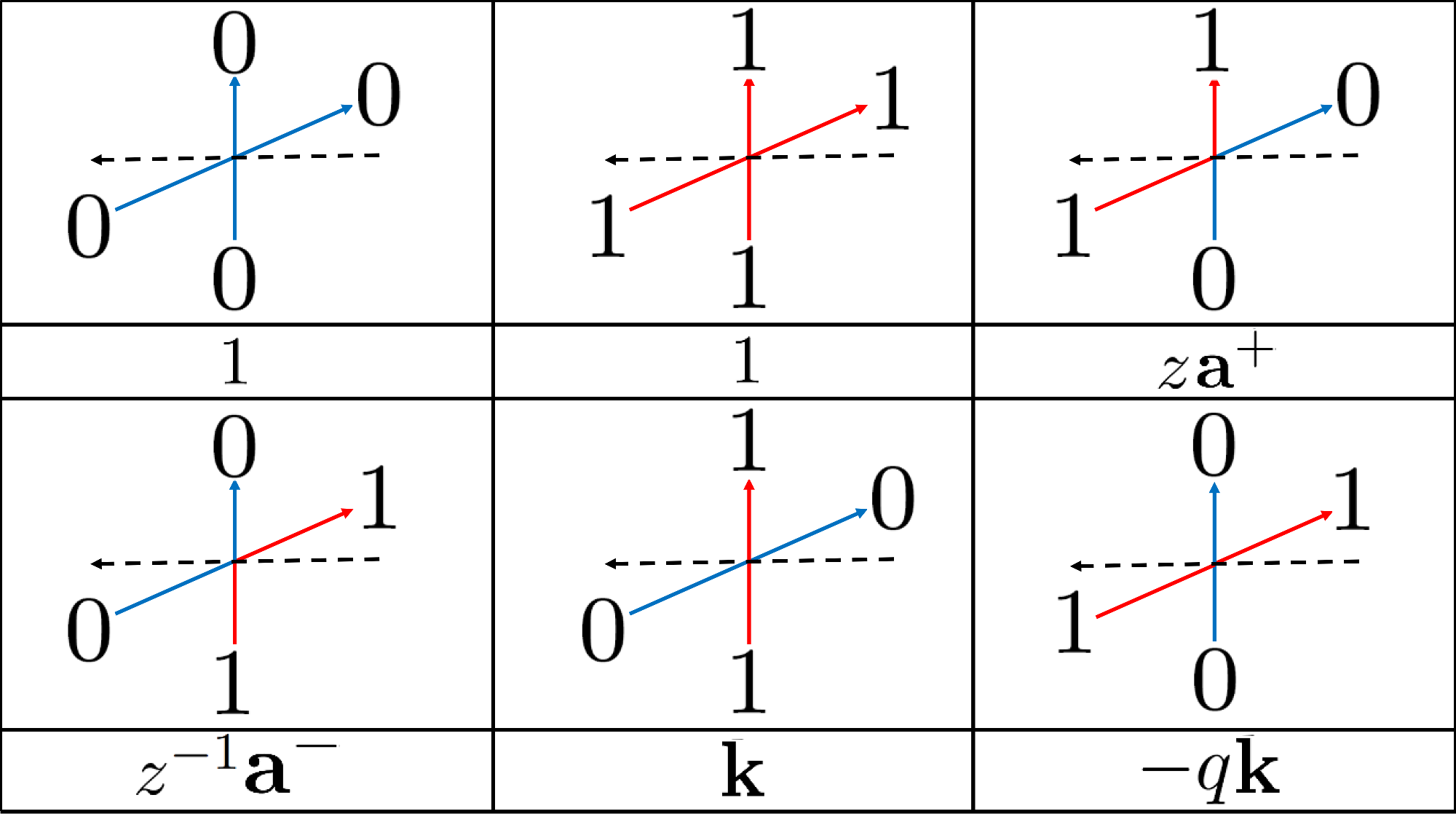}
\caption{The operator-valued $L$-operator.
For each configuration, the weight assigned is given below:
$\mathcal{L}(z)_{00}^{00}=\mathcal{L}(z)_{11}^{11}=1$, 
$\mathcal{L}(z)_{10}^{01}=z\,\mathbf{a}^+$,
$\mathcal{L}(z)_{01}^{10}=z^{-1}\mathbf{a}^-$,
$\mathcal{L}(z)_{01}^{01}=\mathbf{k}$,
$\mathcal{L}(z)_{10}^{10}=-q\,\mathbf{k}$.
By abuse of notation, we also denote the configurations as
$\mathcal{L}_{00}^{00}, \mathcal{L}_{11}^{11}, 
 \mathcal{L}_{10}^{01}, \mathcal{L}_{01}^{10},
 \mathcal{L}_{01}^{01}, \mathcal{L}_{10}^{10}$.
}
\label{figurethreedimension}
\end{figure}

\begin{figure}[htbp]
\centering
\includegraphics[width=12truecm]{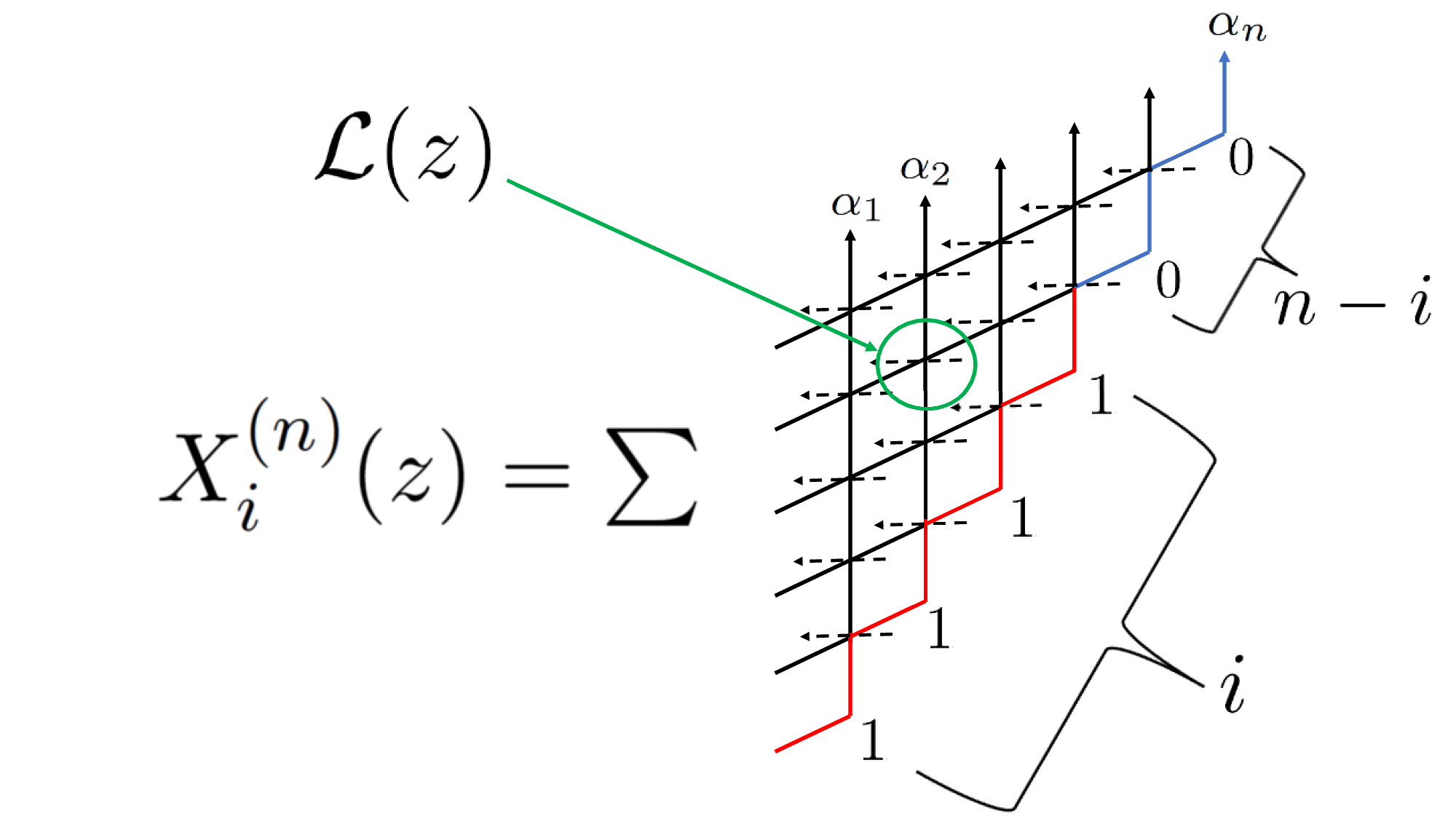}
\caption{The operator $X_i^{(n)}(z)$.  
We sum over all configurations except along one fixed boundary.  
We replace $\mathcal{R}$ by $\mathcal{L}(z)$ at every vertex and take the unweighted sum.
The dependence on $z$ comes from the $\mathcal{L}(z)$-operators.
}
\label{XIoperatorfigureqgeneric}
\end{figure}

For $n=3$, the $X$-operators are explicitly
\begin{align}
X_0^{(3)}(z)
&= 1
 + z\,\mathbf{a}_{11}^+
 + z^2\,\mathbf{a}_{12}^+ \mathbf{a}_{21}^+
 + z\,\mathbf{a}_{12}^+ \mathbf{a}_{21}^+ \mathbf{a}_{11}^-
 - qz\,\mathbf{a}_{12}^+ \mathbf{k}_{11}
 + z\,\mathbf{a}_{21}^+ \mathbf{k}_{11},
\\
X_1^{(3)}(z)
&= \mathbf{a}_{12}^+ \mathbf{a}_{11}^- \mathbf{k}_{21}
 + z\,\mathbf{a}_{12}^+ \mathbf{k}_{21}
 + \mathbf{k}_{11}\mathbf{k}_{21},
\\
X_2^{(3)}(z)
&= \mathbf{a}_{11}^+ \mathbf{a}_{21}^- \mathbf{k}_{12}
 + z^{-1}\mathbf{a}_{21}^- \mathbf{k}_{12}
 + \mathbf{k}_{11}\mathbf{k}_{12},
\\
X_3^{(3)}(z)
&= 1
 + z^{-1}\mathbf{a}_{11}^+ \mathbf{a}_{12}^- \mathbf{a}_{21}^-
 + z^{-1}\mathbf{a}_{11}^-
 + z^{-2}\mathbf{a}_{12}^- \mathbf{a}_{21}^-
 + z^{-1}\mathbf{a}_{12}^- \mathbf{k}_{11}
 - qz^{-1}\mathbf{a}_{21}^- \mathbf{k}_{11}.
\end{align}

By direct computation,
we have
\begin{align}
&[X_1^{(3)}(z),X_1^{(3)}(w)] \nonumber \\
=&(\mathbf{a}_{12}^+ \mathbf{a}_{11}^- \mathbf{k}_{21}
   + z\,\mathbf{a}_{12}^+ \mathbf{k}_{21}
   + \mathbf{k}_{11} \mathbf{k}_{21})
  (\mathbf{a}_{12}^+ \mathbf{a}_{11}^- \mathbf{k}_{21}
   + w\,\mathbf{a}_{12}^+ \mathbf{k}_{21}
   + \mathbf{k}_{11} \mathbf{k}_{21}) \nonumber \\
-&(\mathbf{a}_{12}^+ \mathbf{a}_{11}^- \mathbf{k}_{21}
    + w\,\mathbf{a}_{12}^+ \mathbf{k}_{21}
    + \mathbf{k}_{11} \mathbf{k}_{21})
  (\mathbf{a}_{12}^+ \mathbf{a}_{11}^- \mathbf{k}_{21}
    + z\,\mathbf{a}_{12}^+ \mathbf{k}_{21}
    + \mathbf{k}_{11} \mathbf{k}_{21}) \nonumber \\
=&\, w\,\mathbf{a}_{12}^+ \mathbf{a}_{11}^- \mathbf{k}_{21} \mathbf{a}_{12}^+ \mathbf{k}_{21}
 + z\,\mathbf{a}_{12}^+ \mathbf{k}_{21} \mathbf{a}_{12}^+ \mathbf{a}_{11}^- \mathbf{k}_{21}
 + z\,\mathbf{a}_{12}^+ \mathbf{k}_{21} \mathbf{k}_{11} \mathbf{k}_{21}
 + w\,\mathbf{k}_{11} \mathbf{k}_{21} \mathbf{a}_{12}^+ \mathbf{k}_{21} \nonumber \\
-& z\,\mathbf{a}_{12}^+ \mathbf{a}_{11}^- \mathbf{k}_{21} \mathbf{a}_{12}^+ \mathbf{k}_{21}
 - w\,\mathbf{a}_{12}^+ \mathbf{k}_{21} \mathbf{a}_{12}^+ \mathbf{a}_{11}^- \mathbf{k}_{21}
 - w\,\mathbf{a}_{12}^+ \mathbf{k}_{21} \mathbf{k}_{11} \mathbf{k}_{21}
 - z\,\mathbf{k}_{11} \mathbf{k}_{21} \mathbf{a}_{12}^+ \mathbf{k}_{21}
=0.
\end{align}

One can also check directly that $[X_0^{(3)}(z),X_0^{(3)}(w)] = 0$,
using only the property that operators supported on different 
Fock spaces commute, e.g.
\(
[\mathbf{a}_{ij}, \mathbf{a}_{k\ell}] = 0
\)
for $(i,j)\neq (k,\ell)$.

\begin{conjecture}
For an arbitrary positive integer $N$, the operators satisfy
\[
[X_j^{(N)}(z), X_j^{(N)}(w)] = 0,
\qquad j=0,\dots,N.
\]
\end{conjecture}

If this conjecture is true, this suggests that the $X$-operators form a family of 
good operators for this $q$-generalization.
A quantum-algebraic explanation is missing and is an open problem.

\subsection{
A multiparametric-deformation of elementary symmetric functions
}
Let $t,Q$ be generic complex parameters such that $|t|,|Q|<1$.
For an operator $X$ acting on the bosonic Fock space $\mathfrak{F}$,
we introduce two types of weighted traces
\begin{align}
\mathrm{Tr}_{t,Q}^A(X)&:= \sum_{m \geq 0}
 \langle\!\langle m|
X
| m \rangle \frac{t^m}{(Q;Q)_m},
\\
\mathrm{Tr}_{t,Q}^B(X)&:= \sum_{m \geq 0}
 \langle\!\langle m|
X
| m \rangle \frac{Q^{m(m-1)/2} t^m}{(Q;Q)_m}.
\end{align}


For an operator $X$ acting on \(\bigotimes_{(j,k) \in D_n}\mathfrak{F}_{j,k}\), we define
the type $A$ multi-weighted trace by
\begin{align}
\mathrm{Tr}_{\mathbf{t},\mathbf{Q}}^{A}(X)
:=\sum_{\{m_{j,k}\ge 0\}}
\Bigl(
 \bigotimes_{(j,k) \in D_n} \langle\!\langle m_{j,k}|
\Bigr)\,
X\,
\Bigl(
 \bigotimes_{(j,k) \in D_n} |m_{j,k}\rangle
\Bigr)
\prod_{(j,k) \in D_n}
 \frac{t_{j,k}^{\,m_{j,k}}}{(Q_{j,k};Q_{j,k})_{m_{j,k}}}.
\end{align}

Similarly, the type $B$ multi-weighted trace is defined by
\begin{align}
\mathrm{Tr}_{\mathbf{t},\mathbf{Q}}^{B}(X)
:=\sum_{\{m_{j,k}\ge 0\}}
\Bigl(
 \bigotimes_{(j,k) \in D_n} \langle\!\langle m_{j,k}|
\Bigr)\,
X\,
\Bigl(
 \bigotimes_{(j,k) \in D_n} |m_{j,k}\rangle
\Bigr)
\prod_{(j,k) \in D_n}
 \frac{Q_{j,k}^{\,m_{j,k}(m_{j,k}-1)/2}\,t_{j,k}^{\,m_{j,k}}}
      {(Q_{j,k};Q_{j,k})_{m_{j,k}}}.
\end{align}
Here $t_{j,k}$   weight the occupation numbers, and $Q_{j,k}$ are the 
deformation parameters appearing in the normalizations.

\begin{figure}[htbp]
\centering
    \includegraphics[width=0.8\textwidth]{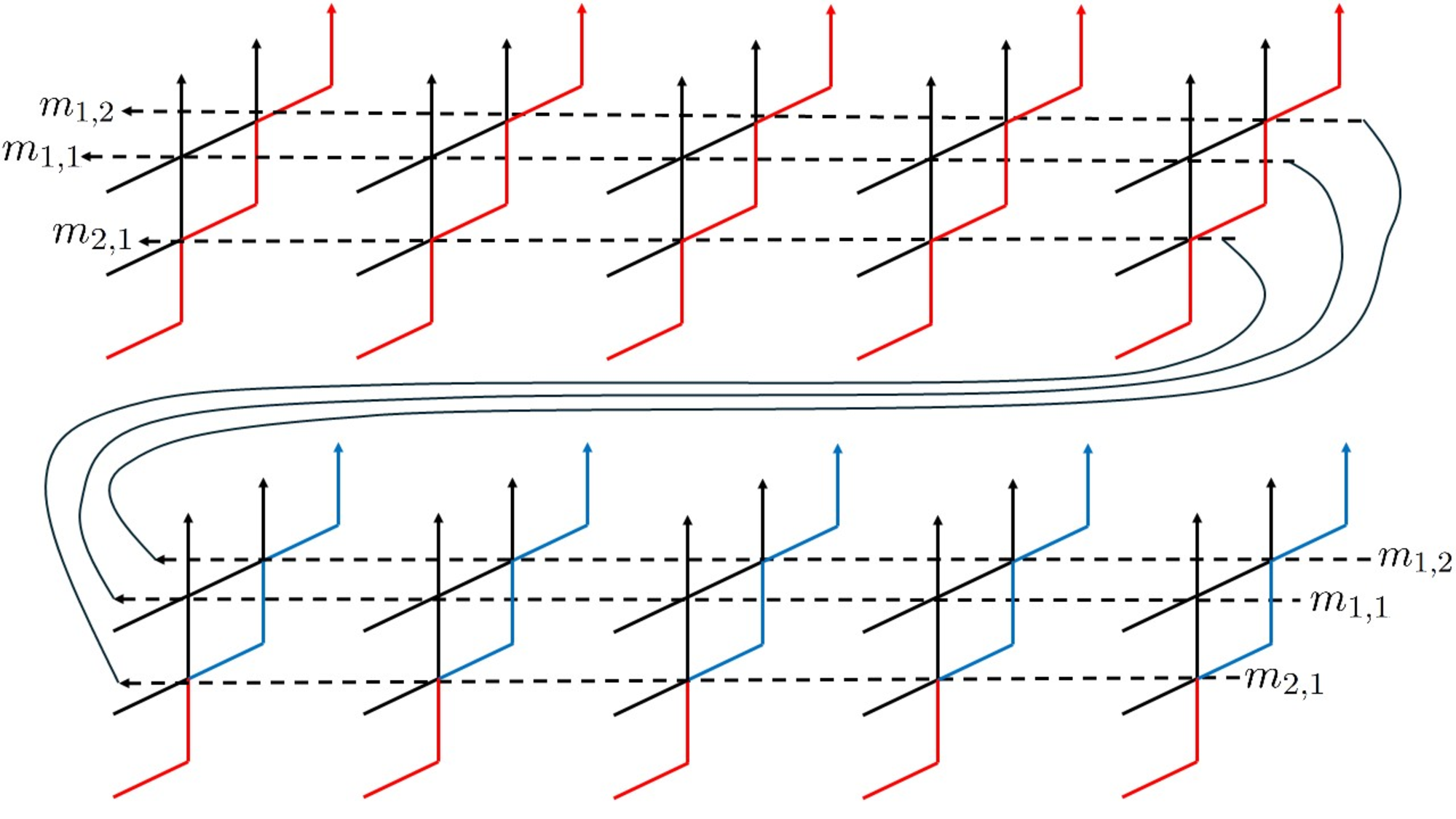}
\caption{
$\langle\!\langle m_{1,1}| \langle\!\langle m_{1,2}| \langle\!\langle m_{2,1}|
X_3^{(3)}(z_1) \cdots X_3^{(3)}(z_m) 
 X_1^{(3)}(w_1) \cdots X_1^{(3)}(w_n)
| m_{1,1} \rangle |m_{1,2} \rangle |m_{2,1} \rangle$.
$\langle\!\langle m_{1,1}| \langle\!\langle m_{1,2}| \langle\!\langle m_{2,1}| X_3^{(3)}(z_1), \dots ,X_3^{(3)}(z_m)$ 
is depicted on the top part, and the bottom part corresponds to 
$X_1^{(3)}(w_1), \dots, X_1^{(3)}(w_n) | m_{1,1} \rangle |m_{1,2} \rangle |m_{2,1} \rangle $.}
\label{figureforone}
\end{figure}


\begin{proposition}
We have
\begin{align}
&\mathrm{Tr}_{\mathbf{t},\mathbf{Q}}^{A}
(
X_N^{(N)}(z_1) \cdots X_N^{(N)}(z_m) X_{N-2}^{(N)}(w_1) \cdots X_{N-2}^{(N)}(w_n)
) \nonumber \\
=&\frac{1}{ \prod_{k=1}^{N-2} \{ ( q^n t_{1,k};Q_{1,k})_\infty (q^n t_{2,k};Q_{2,k})_\infty \} \prod_{(j,k) \in D_n, j \ge 3} (t_{j,k};Q_{j,k})_\infty } 
\nonumber \\
&\times 
 \sum_{i=0}^{\mathrm{min}(m,n)}
e_i(z_1^{-1},\dots,z_m^{-1})e_i(w_1,\dots,w_n)
\sum_{\ell=0}^i q^{\ell(\ell+1)} \begin{bmatrix} i \\ \ell \end{bmatrix}_{q^2}
\frac{1}{(-q^{2 \ell} t_{1,N-1};Q_{1,N-1})_\infty}
,
\label{qtdefelementary} \\
&\mathrm{Tr}_{\mathbf{t},\mathbf{Q}}^{B}
(
X_N^{(N)}(z_1) \cdots X_N^{(N)}(z_m) X_{N-2}^{(N)}(w_1) \cdots X_{N-2}^{(N)}(w_n)
) \nonumber \\
=& \prod_{k=1}^{N-2} \{ (- q^n t_{1,k};Q_{1,k})_\infty (-q^n t_{2,k};Q_{2,k})_\infty \} \prod_{(j,k) \in D_n, j \ge 3} (-t_{j,k};Q_{j,k})_\infty 
\nonumber \\
&\times 
 \sum_{i=0}^{\mathrm{min}(m,n)}
e_i(z_1^{-1},\dots,z_m^{-1})e_i(w_1,\dots,w_n)
\sum_{\ell=0}^i q^{\ell(\ell+1)} \begin{bmatrix} i \\ \ell \end{bmatrix}_{q^2}
(q^{2 \ell} t_{1,N-1};Q_{1,N-1})_\infty,
\label{qtdefelementarytwo}
\end{align}
\end{proposition}

\begin{proof}
Let us show the simplest nontrivial case $N=3$. The general case can be proved in the same way.
By definition,
\begin{align}
&\mathrm{Tr}^{A}_{\mathbf{t},\mathbf{Q}}
\!\left(
X_3^{(3)}(z_1)\cdots X_3^{(3)}(z_m)\,
X_1^{(3)}(w_1)\cdots X_1^{(3)}(w_n)
\right)
\nonumber \\[4pt]
=&\sum_{m_{1,1},\,m_{1,2},\,m_{2,1}\ge 0}
\left(
\frac{t_{1,1}^{m_{1,1}}}{(Q_{1,1};Q_{1,1})_{m_{1,1}}}
\right)
\left(
\frac{t_{1,2}^{m_{1,2}}}{(Q_{1,2};Q_{1,2})_{m_{1,2}}}
\right)
\left(
\frac{t_{2,1}^{m_{2,1}}}{(Q_{2,1};Q_{2,1})_{m_{2,1}}}
\right)
\nonumber \\[4pt]
&\qquad \times
\langle\!\langle m_{1,1}|\,
\langle\!\langle m_{1,2}|\,
\langle\!\langle m_{2,1}|\,
X_3^{(3)}(z_1)\cdots X_3^{(3)}(z_m)\,
X_1^{(3)}(w_1)\cdots X_1^{(3)}(w_n)\,
|m_{1,1}\rangle\,|m_{1,2}\rangle\,|m_{2,1}\rangle .
\label{weightedtraceproof}
\end{align}

We fix $m_{1,1},m_{1,2},m_{2,1}$ and compute
$\langle\!\langle m_{1,1}| \langle\!\langle m_{1,2}| \langle\!\langle m_{2,1}|
X_3^{(3)}(z_1) \cdots X_3^{(3)}(z_m)$ \\
$
 X_1^{(3)}(w_1) \cdots X_1^{(3)}(w_n)
| m_{1,1} \rangle |m_{1,2} \rangle |m_{2,1} \rangle$
which can be graphically depicted
as Figure
\ref{figureforone}.
We first make an observation on the $L$-operators on $\mathfrak{F}_{2,1}$.
From the ice-rule
$\mathcal{L}(z)_{ij}^{k\ell}=0$ unless $i+j=k+\ell$, we 
note that the fermionic edges of the $L$-operators on the second row of the $X_1$-operators
are  all colored in the same way ($\mathcal{L}_{01}^{01}$) as depicted in the bottom part of the left panel of
Figure \ref{figurefortwo}.
The weighted projection operator is assigned to each of these
configurations, and these operators do not change the bosonic number in the bosonic Fock space
$\mathfrak{F}_{2,1}$. The state $|m_{2,1} \rangle$ remains the same after acting $X_1$-operators.
We next act $X_3$-operators and sandwitch by $\langle m_{2,1}|$.
Since two of the fermionic edges of the $L$-operator of the second row for each $X_3$-operator
are already colored by red, we note the other two edges must be colored either both by blue ($\mathcal{L}_{01}^{10}$)
or both by red ($\mathcal{L}_{11}^{11}$), which produces the annihliation operator and the identity operator respectively.
However, we can never use the annihilation operator since using it will produce a state in $\mathfrak{F}_{2,1}$ with bosonic number smaller than
$m_{2,1}$, and sandwitching by $\langle m_{2,1}|$ gives zero. Hence we can only use the identity operator
and note that all the edges must be colored with red ($\mathcal{L}_{11}^{11}$) as depicted in the top part of the right panel of
Figure \ref{figurefortwo}.

Next, in the same figure, we examine
the $L$-operators on $\mathfrak{F}_{1,2}$.
Note that two of the fermionic edges of the $X_1$-operators are already colored by blue
as in the bottom part, which means that either the creation operator (if $\mathcal{L}_{10}^{01}$) or the identity operator
(if $\mathcal{L}_{00}^{00}$)
can be assigned to each part, but not the annihilation operator.
On the other hand, two of the fermionic edges of the $X_3$-operators are already colored by red
as in the top part, which implies that either the annihilation operator (if $\mathcal{L}_{01}^{10}$) or the identity operator (if $\mathcal{L}_{11}^{11}$)
can be assigned, but not the annihilation operator.
Since we finally sandwitch by $|m_{1,2} \rangle$ and $\langle m_{1,2}|$,
the particle number must be preserved which means that the number of annihilation operators
and creation operators which we use must be equal, say $i$ ($0 \leq i \leq \mathrm{min}(m,n)$).
A typical figure is given in the left panel of Figure \ref{figureforthree}.

Finally, we make the following observation on the $L$-operators on $\mathfrak{F}_{1,1}$.
We note from the previous steps that
two of the fermionic edges are already colored and we see that either
the identity operator (if $\mathcal{L}_{11}^{11}$), the projection operator (if $\mathcal{L}_{01}^{01}$)
or the annihilation operator (if $\mathcal{L}_{01}^{10}$) can be assigned,
but not the creation operator. Moreover, since we finally sandwitch by $|m_{1,1} \rangle$ and $\langle m_{1,1}|$,
we note that we cannot use the annihilation operator. This observation together with the ice-rule determines
the coloring of the remaining edges in a unique way as depicted in the right panel of Figure \ref{figureforthree}.

To read out the explicit form of $\langle\!\langle m_{1,1}| \langle\!\langle m_{1,2}| \langle\!\langle m_{2,1}|
X_3^{(3)}(z_1) \cdots X_3^{(3)}(z_m)$ \\
$
 X_1^{(3)}(w_1) \cdots X_1^{(3)}(w_n)
| m_{1,1} \rangle |m_{1,2} \rangle |m_{2,1} \rangle$,
see Figure \ref{figureforfour}.
We label the $X_3$-operators which annihilation operators appear in $\mathfrak{F}_{1,2}$ as $X_3^{(3)}(z_{j_1}),\dots,X_3^{(3)}(z_{j_i})$ $(j_1<j_2<\cdots<j_i)$,
and  the $X_1$-operators which creation operators appear in $\mathfrak{F}_{1,2}$ as $X_1^{(3)}(w_{k_1}),\dots,X_1^{(3)}(w_{k_i})$ $(k_1<k_2<\cdots<k_i)$.
The weight coming from the $L$-operators on $\mathfrak{F}_{1,1}$, $\mathfrak{F}_{1,2}$, $\mathfrak{F}_{2,1}$ of
  $X_3^{(3)}(z_{j_p})$ ($p=1,\dots,i$) is $q^{m_{1,1}}$, $(1-q^{2(m_{1,2}+p)})z_{j_p}^{-1}$, $1$ respectively,
hence the operator $X_3^{(3)}(z_{j_p})$ contributes the factor $q^{m_{1,1}}(1-q^{2(m_{1,2}+p)})z_{j_p}^{-1}$.
We also note all the other $X_3$-operators give the factor 1 since all the weights coming from the $L$-operators
on $\mathfrak{F}_{1,1}$, $\mathfrak{F}_{1,2}$, $\mathfrak{F}_{2,1}$ are 1.
The weight coming from the $L$-operators on $\mathfrak{F}_{1,1}$, $\mathfrak{F}_{1,2}$, $\mathfrak{F}_{2,1}$ of
$X_1^{(3)}(w_{k_p})$ ($p=1,\dots,i$) is $1$, $w_{k_p}$, $q^{m_{2,1}}$ respectively,
hence the operator $X_1^{(3)}(w_{k_p})$ contributes the factor $q^{m_{2,1}} w_{k_p}$.
Likewise, we note each of the other $X_1$-operators gives the factor $q^{m_{1,1}+m_{2,1}}$.
Multiplying the factors coming from all $X$-operators and simplifying, we get the total weight
\begin{align}
q^{(m_{1,1}+m_{2,1})n} \prod_{p=1}^i (1-q^{2(m_{1,2}+p)}) z_{j_1}^{-1} \cdots z_{j_i}^{-1} w_{k_1} \cdots w_{k_i},
\label{weightfortypicalconfig}
\end{align}
from this configuration. This
is the weight for one typical configutation, and we need to take all this kind of configuration into account to get
the expression for \\
$\langle\!\langle m_{1,1}| \langle\!\langle m_{1,2}| \langle\!\langle m_{2,1}|
X_3^{(3)}(z_1) \cdots X_3^{(3)}(z_m)$
$
 X_1^{(3)}(w_1) \cdots X_1^{(3)}(w_n)
| m_{1,1} \rangle |m_{1,2} \rangle |m_{2,1} \rangle$. This means that we take the summation 
$\sum_{i=0}^{\mathrm{min}(m,n)} \sum_{1 \le j_1<j_2<\cdots<j_p \le m}
\sum_{1 \le k_1<k_2<\cdots<k_p \le n}
$. We get
\begin{align}
&\langle\!\langle m_{1,1}| \langle\!\langle m_{1,2}| \langle\!\langle m_{2,1}|
X_3^{(3)}(z_1) \cdots X_3^{(3)}(z_m)
 X_1^{(3)}(w_1) \cdots X_1^{(3)}(w_n)
| m_{1,1} \rangle |m_{1,2} \rangle |m_{2,1} \rangle
 \nonumber \\
=&\sum_{i=0}^{\mathrm{min}(m,n)} \sum_{1 \le j_1<j_2<\cdots<j_p \le m}
\sum_{1 \le k_1<k_2<\cdots<k_p \le n}
q^{(m_{1,1}+m_{2,1})n} \prod_{p=1}^i (1-q^{2(m_{1,2}+p)}) z_{j_1}^{-1} \cdots z_{j_i}^{-1} w_{k_1} \cdots w_{k_i}
\nonumber \\
=&q^{(m_{1,1}+m_{2,1})n} \sum_{i=0}^{\mathrm{min}(m,n)}
\prod_{p=1}^i (1-q^{2(m_{1,2}+p)}) e_i(z_1^{-1},\dots,z_m^{-1})e_i(w_1,\dots,w_n). \label{onesummationform}
\end{align}
Inserting \eqref{onesummationform} into
\eqref{weightedtraceproof}, we get
\begin{align}
&\mathrm{Tr}_{\mathbf{t},\mathbf{Q}}^A
(
X_3^{(3)}(z_1) \cdots X_3^{(3)}(z_m) X_1^{(3)}(w_1) \cdots X_1^{(3)}(w_n)
) \nonumber \\
=&\sum_{m_{1,1},m_{1,2},m_{2,1} \geq 0}
\frac{
t_{1,1}^{m_{1,1}}
}{(Q_{1,1};Q_{1,1})_{m_{1,1}}}
\frac{
t_{1,2}^{m_{1,2}}
}{(Q_{1,2};Q_{1,2})_{m_{1,2}}}
\frac{
t_{2,1}^{m_{2,1}}
}{(Q_{2,1};Q_{2,1})_{m_{2,1}}}
q^{(m_{1,1}+m_{2,1})n} 
\nonumber \\
&\times \sum_{i=0}^{\mathrm{min}(m,n)}
\prod_{p=1}^i (1-q^{2(m_{1,2}+p)}) e_i(z_1^{-1},\dots,z_m^{-1})e_i(w_1,\dots,w_n) \nonumber \\
=&\sum_{m_{1,1} \geq 0}
\frac{
(q^n t_{1,1})^{m_{1,1}}
}{(Q_{1,1};Q_{1,1})_{m_{1,1}}}
\sum_{m_{2,1} \geq 0}
\frac{
(q^n t_{2,1})^{m_{2,1}}
}{(Q_{2,1};Q_{2,1})_{m_{2,1}}}
 \sum_{i=0}^{\mathrm{min}(m,n)}
 e_i(z_1^{-1},\dots,z_m^{-1})e_i(w_1,\dots,w_n)
 \nonumber \\
&\times \sum_{m_{1,2} \geq 0}
\frac{
(t_{1,2})^{m_{1,2}}
}{(Q_{1,2};Q_{1,2})_{m_{1,2}}} \prod_{p=1}^i (1-q^{2(m_{1,2}+p)}). \label{righthandsidetrace}
\end{align}
From \eqref{qbinomial} and \eqref{qsummationone}, we have
\begin{align}
&\sum_{m_{j,1} \geq 0}
\frac{
(q^n t_{j,1})^{m_{j,1}}
}{(Q_{j,1};Q_{j,1})_{m_{j,1}}}=\frac{1}{(q^n t_{j,1};Q_{j,1})_\infty}, \ \ \ j=1,2,
  \\
&\sum_{m_{1,2} \geq 0}
\frac{
(t_{1,2})^{m_{1,2}}
}{(Q_{1,2};Q_{1,2})_{m_{1,2}}} \prod_{p=1}^i (1-q^{2(m_{1,2}+p)}) \nonumber \\
=&\sum_{\ell=0}^i q^{\ell(\ell+1)}
\begin{bmatrix} i \\ \ell \end{bmatrix}_{q^2}
\sum_{m_{1,2} \geq 0}
\frac{
(-q^{2\ell} t_{1,2})^{m_{1,2}}
}{(Q_{1,2};Q_{1,2})_{m_{1,2}}} =\sum_{\ell=0}^i q^{\ell(\ell+1)}
\begin{bmatrix} i \\ \ell \end{bmatrix}_q
\frac{1}{(-q^{2 \ell} t_{1,2};Q_{1,2})_\infty},
\end{align}
and inserting 
into \eqref{righthandsidetrace}, we get
\begin{align}
&\mathrm{Tr}_{\mathbf{t},\mathbf{Q}}^{A}
(
X_3^{(3)}(z_1) \cdots X_3^{(3)}(z_m) X_{1}^{(3)}(w_1) \cdots X_{1}^{(3)}(w_n)
) \nonumber \\
=& \frac{1}{ ( q^n t_{1,1};Q_{1,1})_\infty (q^n t_{2,1};Q_{2,1})_\infty  } \nonumber \\
&\times
 \sum_{i=0}^{\mathrm{min}(m,n)}
e_i(z_1^{-1},\dots,z_m^{-1})e_i(w_1,\dots,w_n)
\sum_{\ell=0}^i q^{\ell(\ell+1)} \begin{bmatrix} i \\ \ell \end{bmatrix}_{q^2}
\frac{1}{(-q^{2 \ell} t_{1,2};Q_{1,2})_\infty}
.
\end{align}


The other type of trace
\begin{align}
&\mathrm{Tr}^{B}_{\mathbf{t},\mathbf{Q}}
\!\left(
X_3^{(3)}(z_1)\cdots X_3^{(3)}(z_m)\,
X_1^{(3)}(w_1)\cdots X_1^{(3)}(w_n)
\right)
\nonumber \\[4pt]
=&\sum_{m_{1,1},\,m_{1,2},\,m_{2,1}\ge 0}
\left(
\frac{ Q_{1,1}^{m_{1,1}(m_{1,1}-1)/2}   t_{1,1}^{m_{1,1}}   }{(Q_{1,1};Q_{1,1})_{m_{1,1}}}
\right)
\left(
\frac{  Q_{1,2}^{m_{1,2}(m_{1,2}-1)/2}   t_{1,2}^{m_{1,2}}}{(Q_{1,2};Q_{1,2})_{m_{1,2}}}
\right)
\left(
\frac{   Q_{2,1}^{m_{2,1}(m_{2,1}-1)/2}     t_{2,1}^{m_{2,1}}}{(Q_{2,1};Q_{2,1})_{m_{2,1}}}
\right)
\nonumber \\[4pt]
&\qquad \times
\langle\!\langle m_{1,1}|\,
\langle\!\langle m_{1,2}|\,
\langle\!\langle m_{2,1}|\,
X_3^{(3)}(z_1)\cdots X_3^{(3)}(z_m)\,
X_1^{(3)}(w_1)\cdots X_1^{(3)}(w_n)\,
|m_{1,1}\rangle\,|m_{1,2}\rangle\,|m_{2,1}\rangle .
\label{weightedtraceprooftwo}
\end{align}
can be computed in a similar way using \eqref{qbinomial} and \eqref{qsummationtwo}, resulting in 
\begin{align}
&\mathrm{Tr}_{\mathbf{t},\mathbf{Q}}^{B}
(
X_3^{(3)}(z_1) \cdots X_3^{(3)}(z_m) X_{1}^{(3)}(w_1) \cdots X_{1}^{(3)}(w_n)
) \nonumber \\
=& (- q^n t_{1,1};Q_{1,1})_\infty (-q^n t_{2,1};Q_{2,1})_\infty   \nonumber \\
&\times
 \sum_{i=0}^{\mathrm{min}(m,n)}
e_i(z_1^{-1},\dots,z_m^{-1})e_i(w_1,\dots,w_n)
\sum_{\ell=0}^i q^{\ell(\ell+1)} \begin{bmatrix} i \\ \ell \end{bmatrix}_{q^2}
(q^{2 \ell} t_{1,2};Q_{1,2})_\infty
.
\end{align}

\end{proof}

\begin{figure}[htbp]
\centering
    \includegraphics[width=0.45\textwidth]{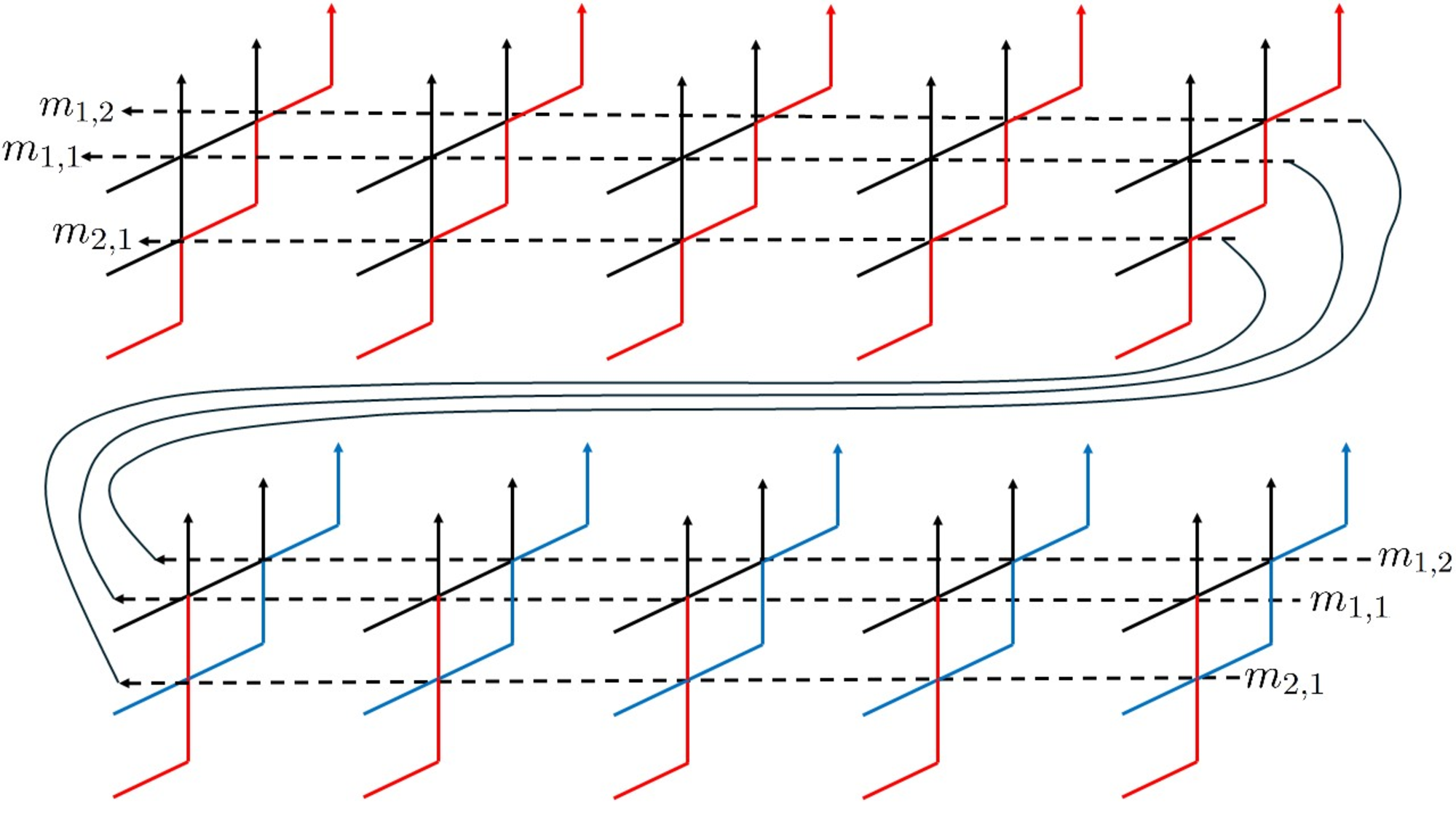}
\includegraphics[width=0.45\textwidth]{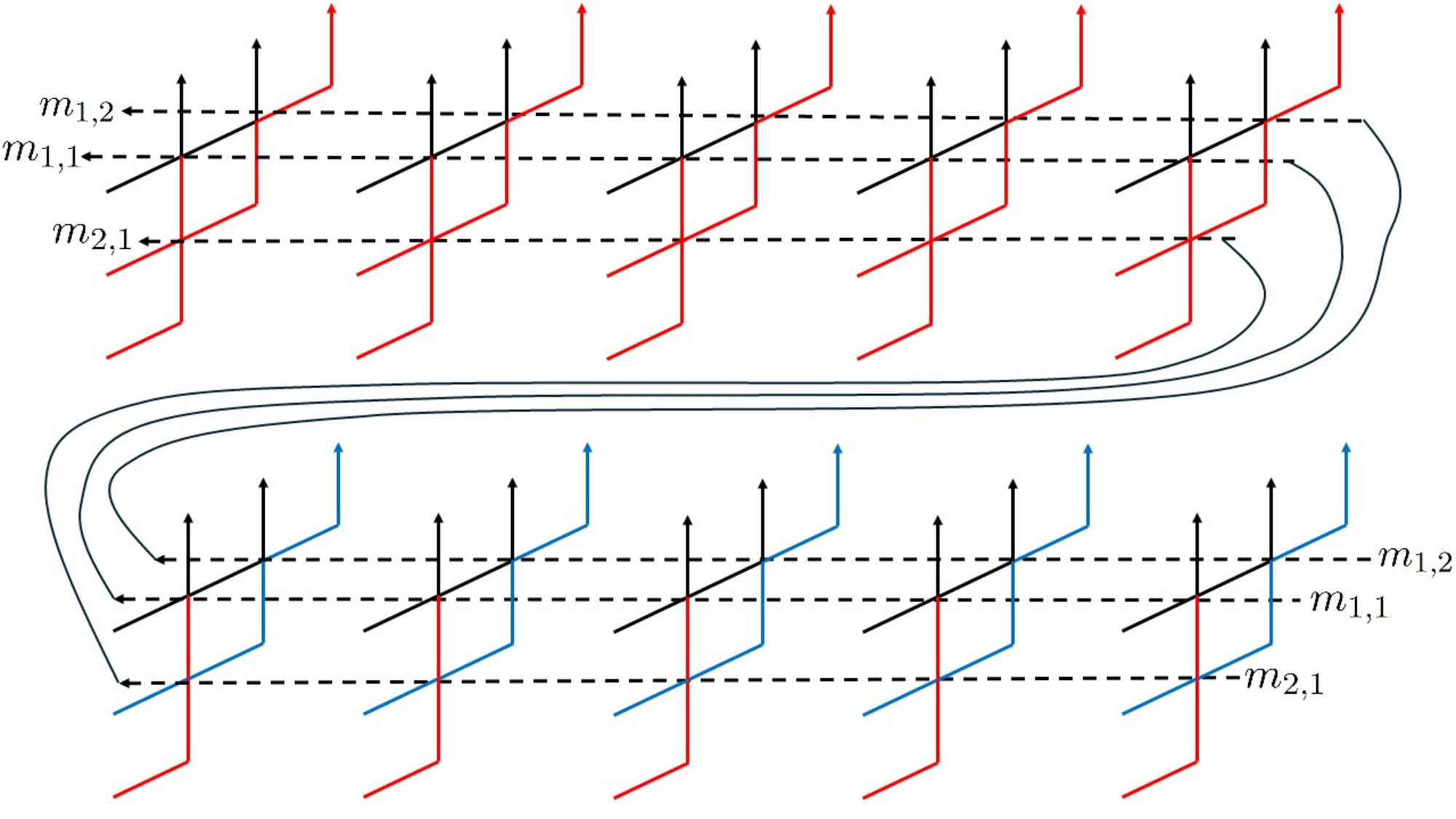}
\caption{One notes that the $L$-operators acting on the bosonic Fock space $\mathfrak{F}_{2,1}$ are fixed uniquely.}
\label{figurefortwo}
\end{figure}

\begin{figure}[htbp]
\centering
    \includegraphics[width=0.45\textwidth]{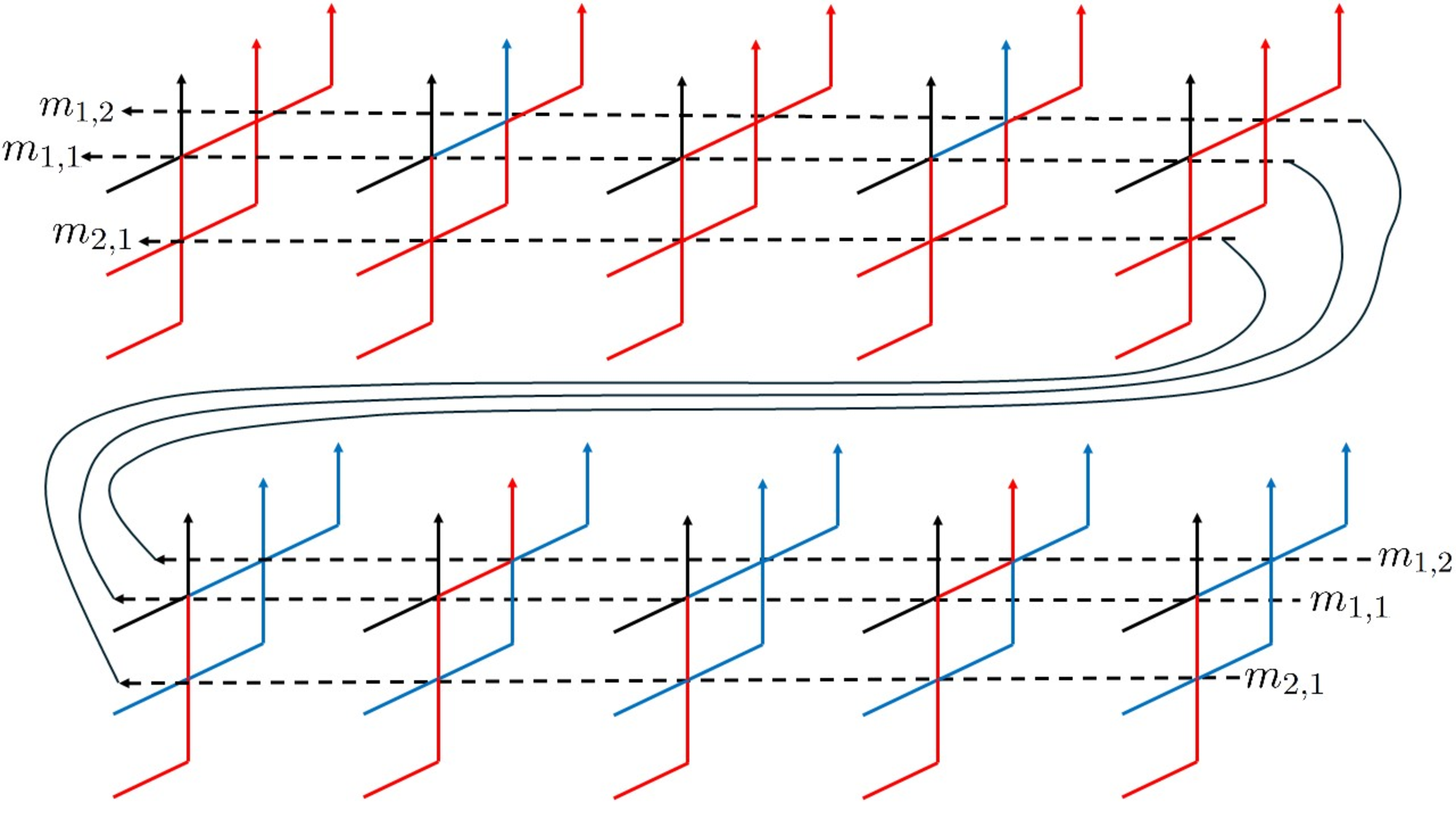}
\includegraphics[width=0.45\textwidth]{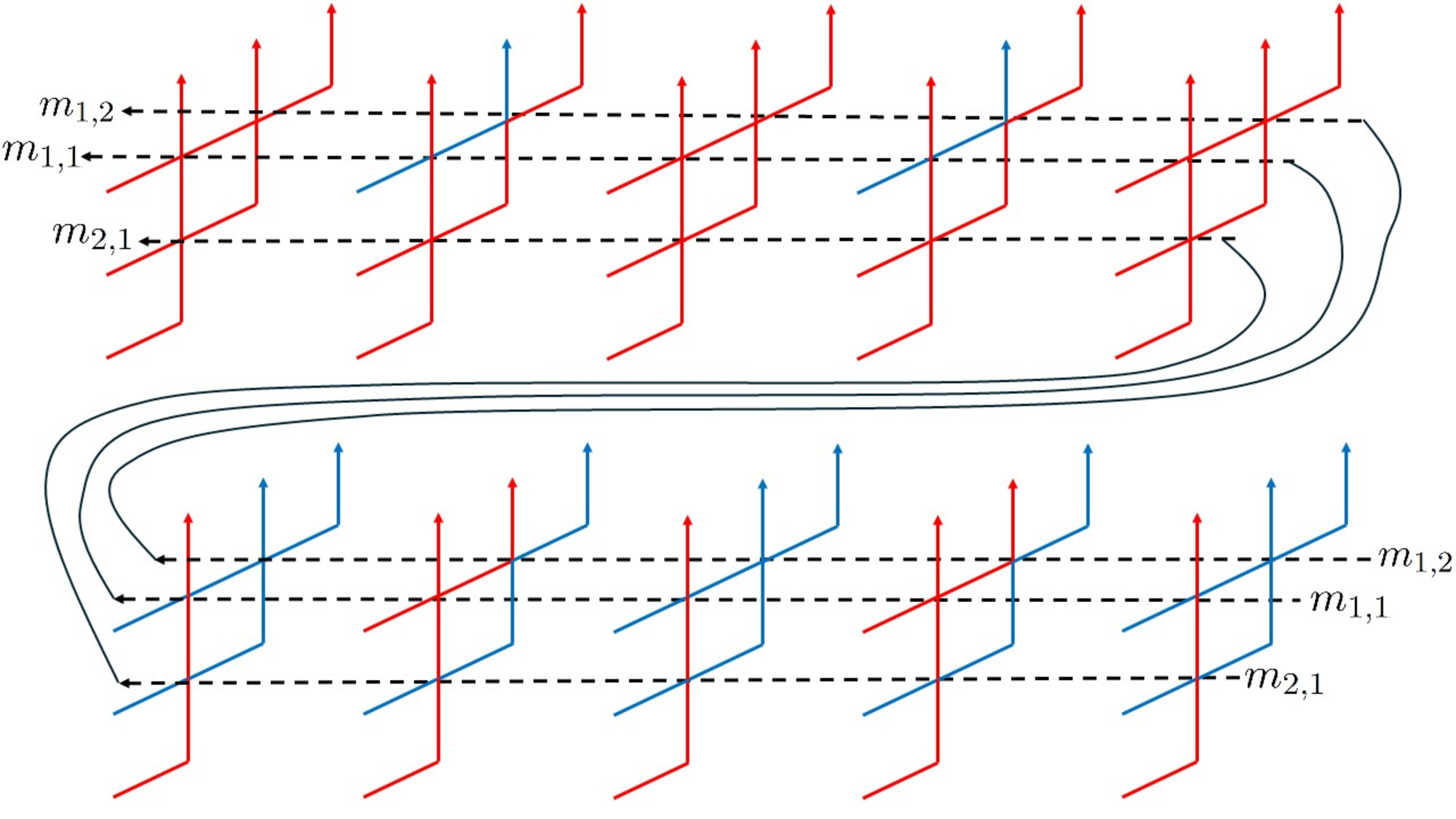}
\caption{A typical configuration which gives a nonzero contribution to the partition function.}
\label{figureforthree}
\end{figure}

\begin{figure}[htbp]
\centering
    \includegraphics[width=0.8\textwidth]{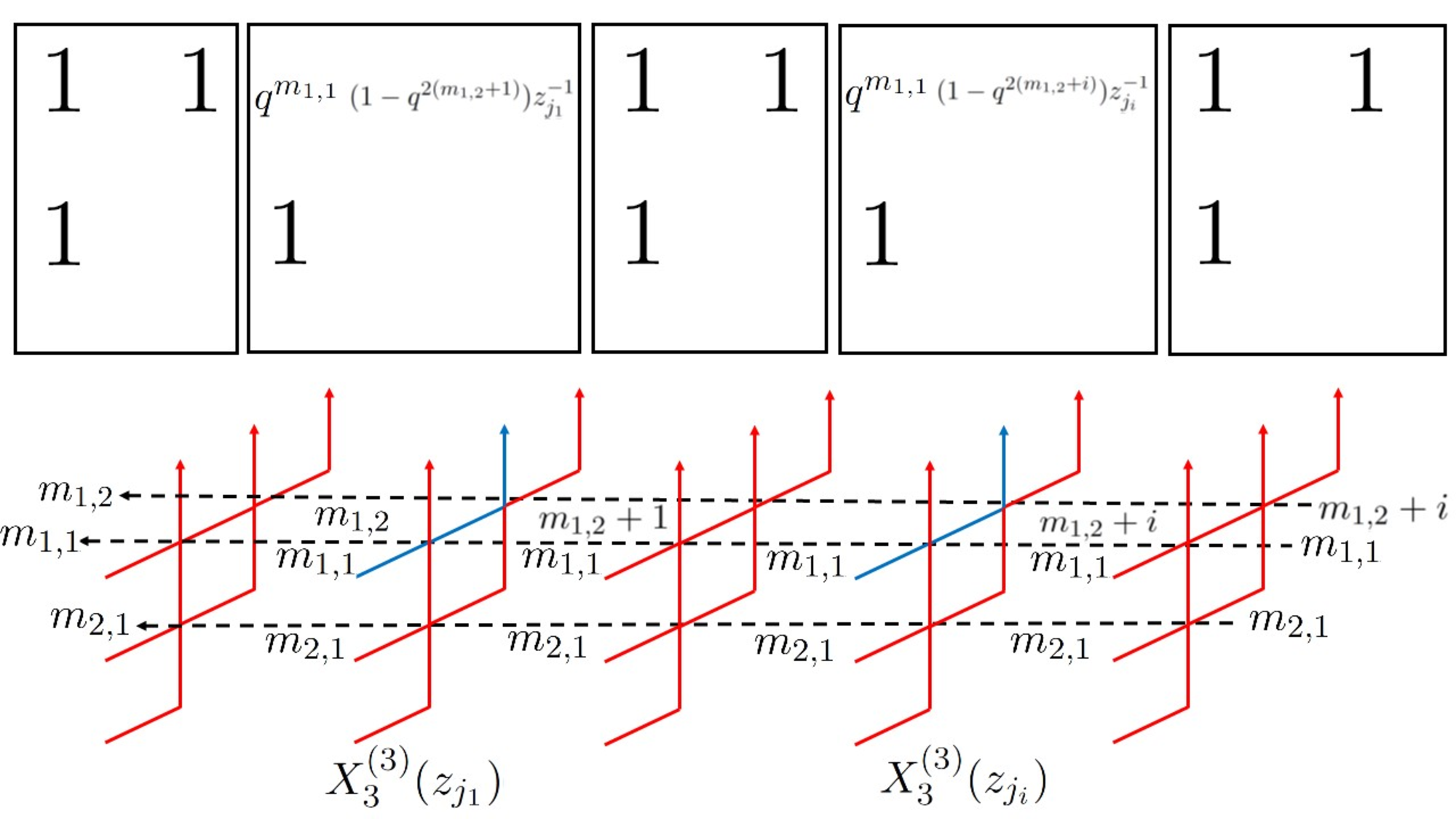}
\includegraphics[width=0.8\textwidth]{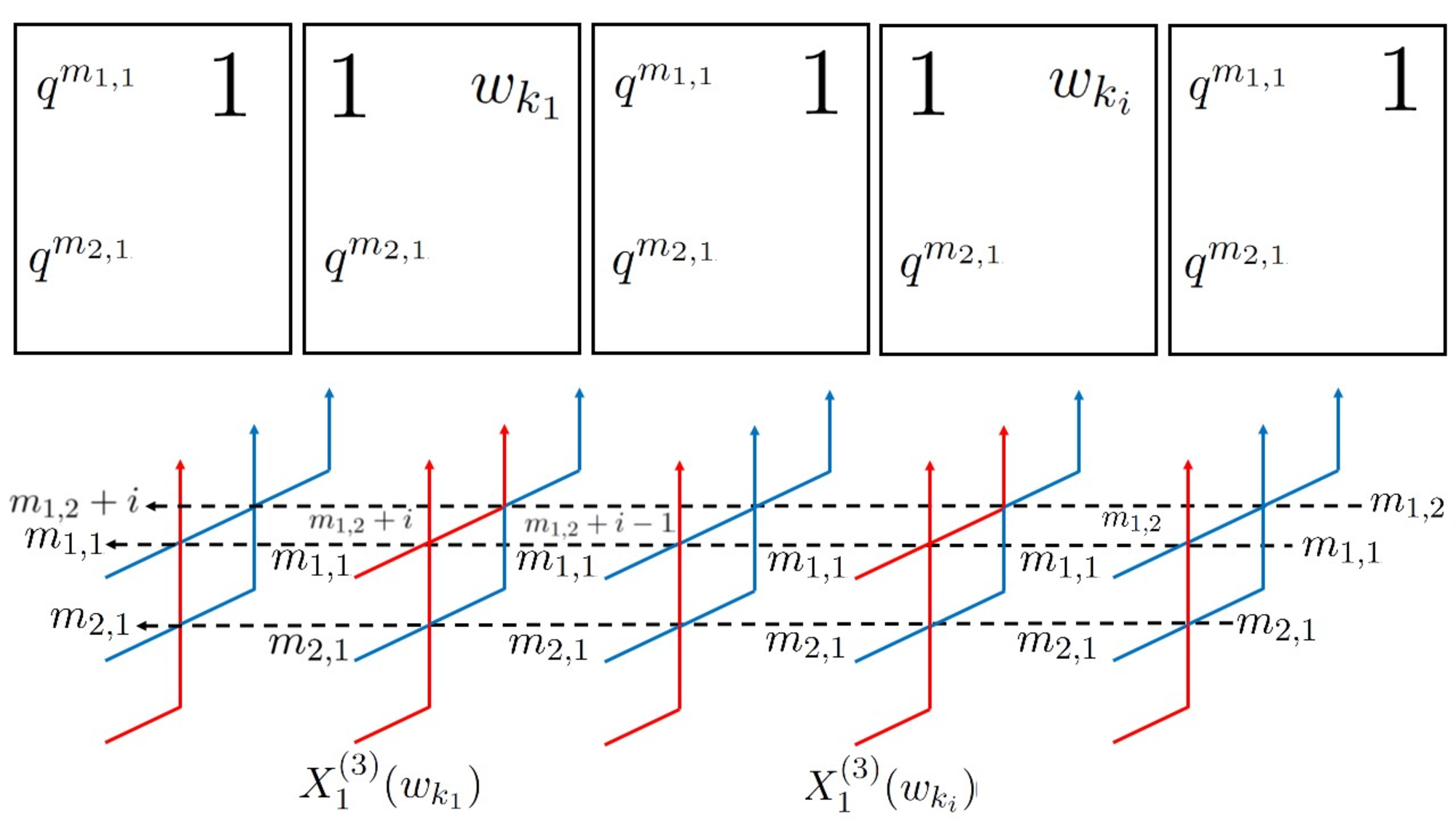}
\caption{
The weights assigned to the  $X_1$- and $X_3$-operators.}
\label{figureforfour}
\end{figure}

For $N=3$ and setting $t_{i,j}=Q_{i,j}=0$ for all $i,j$,
\eqref{qtdefelementary} can be rewitten as follows.
\begin{corollary} We have
\begin{align}
&\langle\!\langle \Omega|
X_3^{(3)}(z_1) \cdots X_3^{(3)}(z_m) X_1^{(3)}(w_1) \cdots X_1^{(3)}(w_n)| \Omega \rangle \nn \\
=&\frac{1}{z_1 \cdots z_m}
\sum_{\ell \ge 0}
\prod_{k=1}^\ell (1-q^{2k}) e_{m-\ell}(z_1,\dots,z_m) e_{\ell}(w_1,\dots,w_n).
\label{qtdefelementaryspecial}
\end{align}
\end{corollary}
{\bf Note}:
For this degenerate case,
the right-hand side of \eqref{qtdefelementaryspecial}
essentially coincides with the case $m_1=0$ and $q=0$
of \cite[Proposition~4.6]{CdGW}.
The more general version differs from
\cite[Proposition~4.6]{CdGW}.

Further specializing $q=0$
in \eqref{qtdefelementaryspecial}, we get 
\begin{align}
&\langle\!\langle \Omega|
X_3^{(3)}(z_1) \cdots X_3^{(3)}(z_m) X_1^{(3)}(w_1) \cdots X_1^{(3)}(w_n)| \Omega \rangle \nn \\
=&\frac{1}{z_1 \cdots z_m}
e_{m}(z_1,\dots,z_m,w_1,\dots,w_n).
\end{align}
From this degeneration, we note
\eqref{qtdefelementary}
can be viewed as a $q$-deformation of elementary symmetric functions.

\begin{figure}[htbp]
\centering
    \includegraphics[width=0.8\textwidth]{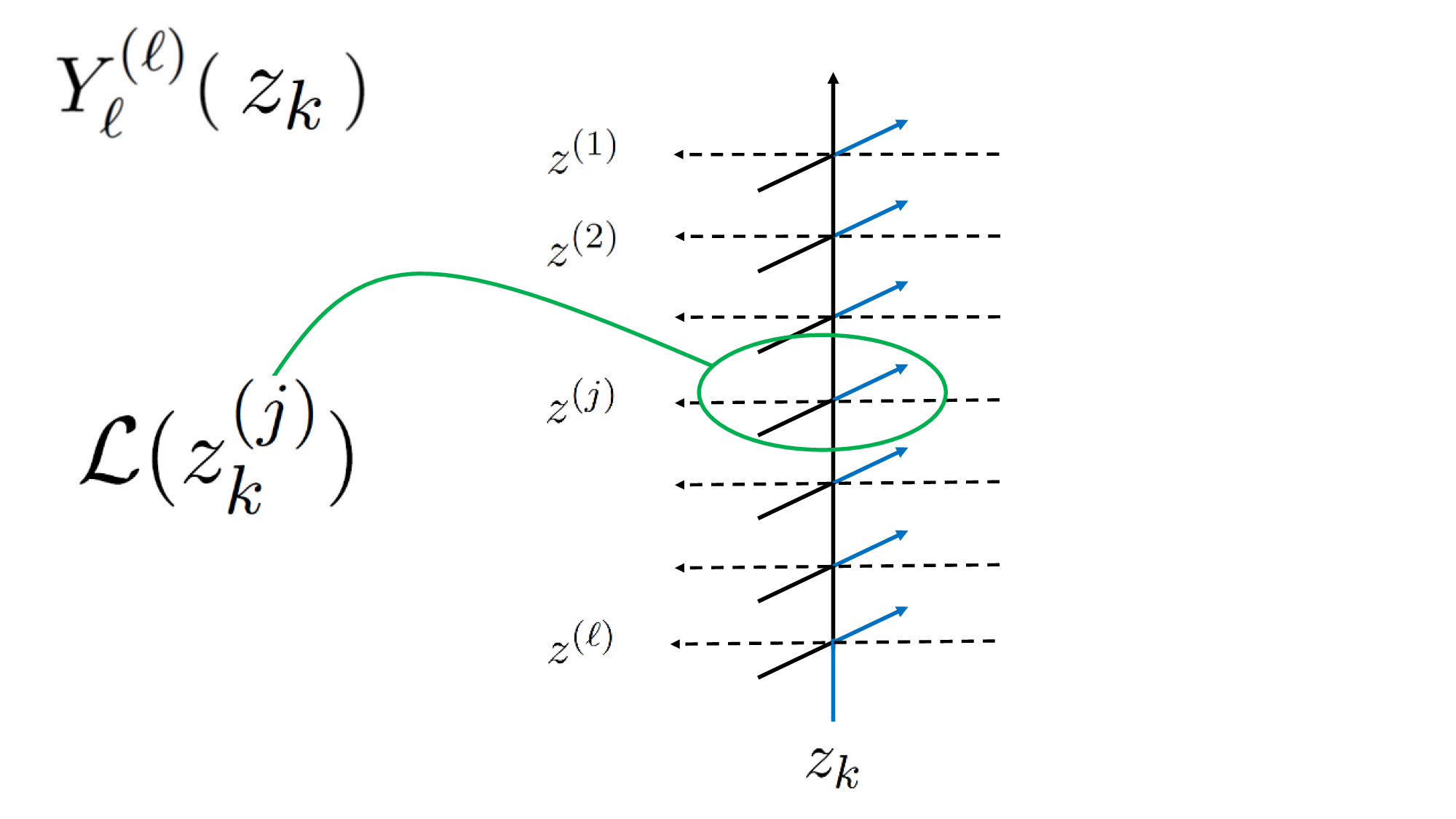}
\caption{The $Y_\ell^{(\ell)}$-operator. The $L$-operator in the $j$-th layer is assigned the spectral parameter 
$z_k^{(j)}$, with $z^{(j)}$ coming from the horizontal line and $z_k$ from 
the vertical line.
}
\label{figureforYoperator}
\end{figure}

\begin{figure}[htbp]
\centering
    \includegraphics[width=0.8\textwidth]{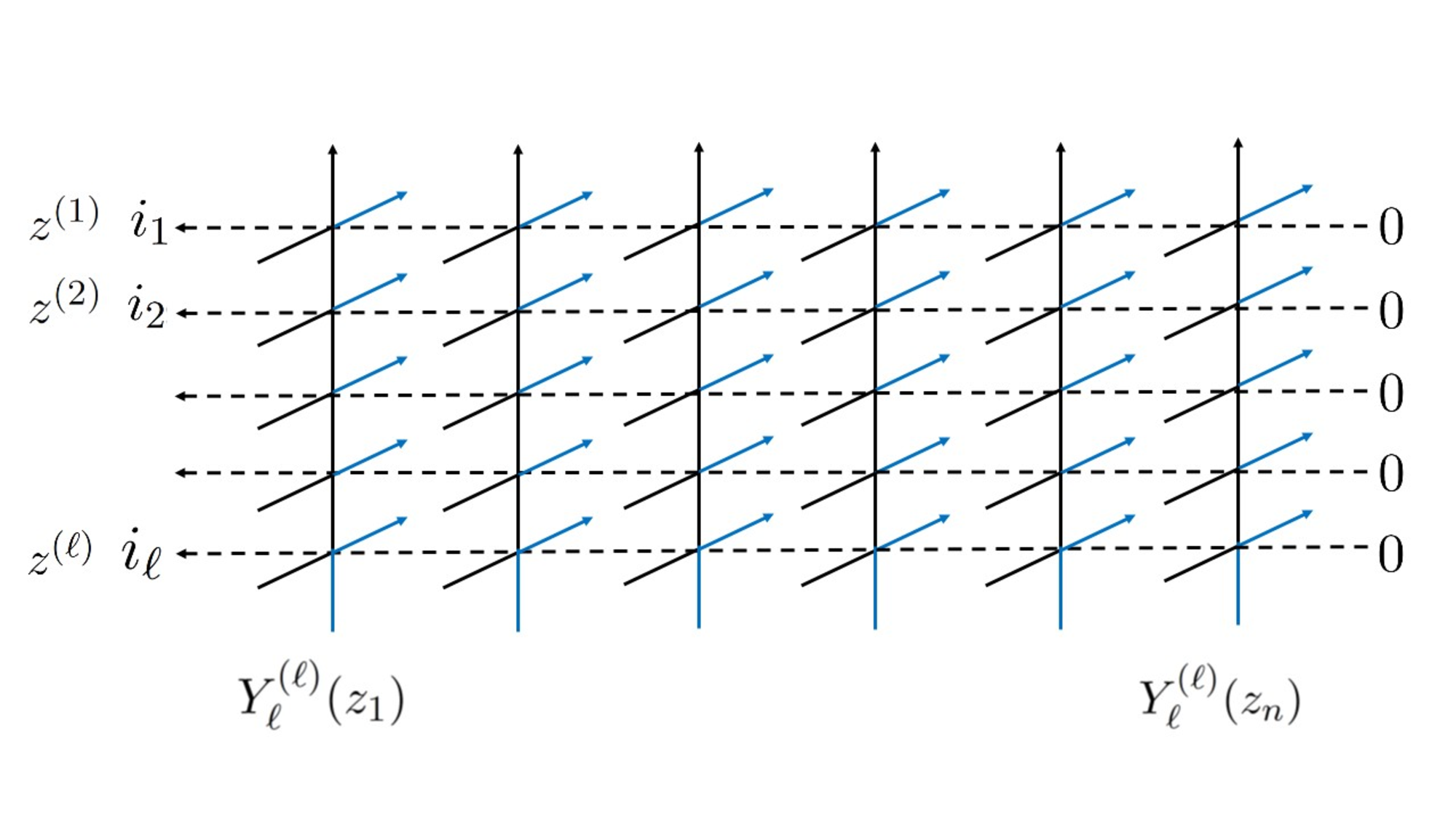}
\caption{$\langle\!\langle i_1,i_2,\dots,i_\ell 
|Y_\ell^{(\ell)}(z_1) \cdots Y_\ell^{(\ell)}(z_n)|\Omega \rangle$.}
\label{generalizedqloopelementary}
\end{figure}

\subsection{A $q$-deformation of loop elementary symmetric functions}

We introduce operators $Y_\ell^{(\ell)}(z)$, $\ell \geq 1$ acting on $\mathfrak{F}^{\otimes \ell}$,
graphically represented as Figure \ref{figureforYoperator}.
The $L$-operator in the $j$-th layer is assigned the spectral parameter 
$z_k^{(j)}$, with $z^{(j)}$ coming from the horizontal line and $z_k$ from 
the vertical line. We use the notation $Y_\ell^{(\ell)}(z)$, following the $q=0$ case used in \cite{IMO}.

In this subsection, we introduce a different type of partition functions
from the previous subsection
which are constructed from $Y_\ell^{(\ell)}$-operators and show they
give a $q$-deformation of the loop elementary symmetric functions.

Introduce the following notation
$\langle\!\langle i_1, i_2, \cdots, i_\ell |:=\langle\!\langle i_1| \otimes \langle\!\langle i_2| \otimes \cdots \otimes \langle\!\langle i_\ell|
$, $i_1,i_2,\dots,i_\ell \in \mathbb{Z}_{\geq 0}$
for tensor product of dual vectors acting on bosonic Fock spaces.

\begin{proposition}
We have
\begin{align}
&\langle\!\langle i_1,i_2,\dots,i_\ell 
|Y_\ell^{(\ell)}(z_1) \cdots Y_\ell^{(\ell)}(z_n)|\Omega \rangle \nn \\
=&
\sum_{\{ m_k^{(j)} \}} \prod_{j=1}^\ell \prod_{k=1}^{i_j} 
\Bigg(
z_{m_k^{(j)}}^{(j)} \prod_{p=1}^{j-1} q^{s(j,k,p)}
\Bigg), \label{higherrankqloop}
\end{align}
where
the sum is over $m_k^{(j)}$, $j=1,\dots,\ell$, $k=1,\dots,i_j$
satisfying
$1 \le m_{i_j}^{(j)}<m_{i_j-1}^{(j)}<\cdots<m_1^{(j)} \le n$
and $m_r^{(s)} \neq m_k^{(j)}$ if $(r,s) \neq (k,j)$.
$s(j,k,p)$ is the integer such that $m_{s(j,k,p)+1}^{(p)}<m_k^{(j)}<m_{s(j,k,p)}$
with the convention $m_0^{(j)}:=n+1$, $m_{i_j+1}^{(j)}:=0$, $j=1,\dots,\ell$.
\end{proposition}
See Figure \ref{generalizedqloopelementary} for a graphical description of $\langle\!\langle i_1,i_2,\dots,i_\ell 
|Y_\ell^{(\ell)}(z_1) \cdots Y_\ell^{(\ell)}(z_n)|\Omega \rangle$.

\begin{figure}[htbp]
\centering
    \includegraphics[width=0.8\textwidth]{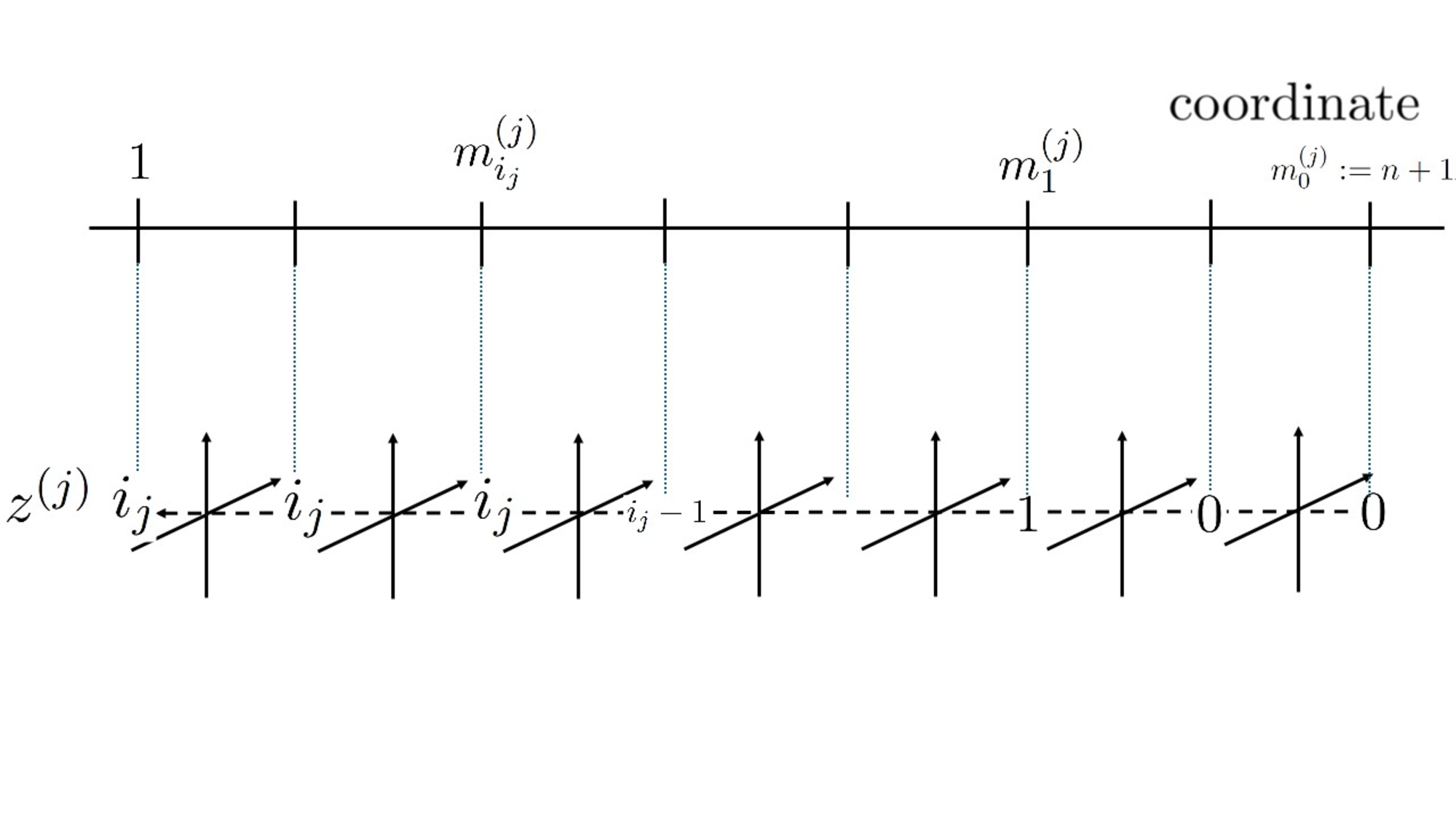}
\caption{
$m_k^{(j)}:=\mathrm{Max} \{ m \ | \ 1 \le m \le n, i_m=k  \}$ for $k=1,\dots,i_j$.
These are coordinates which the number of bosons changes.
We also define $m_0^{(j)}:=n+1$ and $m_{i_j+1}^{(j)}:=0$.}
\label{labelforposition
}
\end{figure}

\begin{proof}

First, observe that the only type of $Y$-operators that is used is  $Y_\ell^{(\ell)}$, and we note that locally three types of configurations
$\mathcal{L}_{00}^{00}, \mathcal{L}_{10}^{01}, \mathcal{L}_{01}^{01}$ can appear for $Y_\ell^{(\ell)}$.
Also note that none of the three types of configurations produces the annihilation operator.
Keeping this in mind, we first label configurations of the dashed line (bosonic Fock space) for each row.
Since annihilation operators do not appear, one notes
that the configurations in the $j$-th dashed line can always be written as a sequence of nonincreasing integers
$(b_1,\dots,b_n)=(i_j,\dots i_j,i_{j}-1,\dots,i_{j}-1,\dots, 1\dots,1,0,\dots,0)$.
We define $m_k^{(j)}:=\mathrm{Max} \{ m \ | \ 1 \le m \le n, i_m=k  \}$ for $k=1,\dots,i_j$.
We also define $m_0^{(j)}:=n+1$ and $m_{i_j+1}^{(j)}:=0$.
A bosonic configuration in the $j$-th dashed line can be encoded as $m_k^{(j)}$, $k=0,\dots,i_j$.
For example, if the bosonic configuration in the first dashed row is $33332222211100$, then
$m_3^{(1)}=4$, $m_2^{(1)}=9$, $m_1^{(1)}=12$, $m_0^{(1)}=14$. See also Figure \ref{labelforposition
}.
Let us expand
$
\langle\!\langle i_1,i_2,\dots,i_\ell
|Y_\ell^{(\ell)}(z_1) \cdots Y_\ell^{(\ell)}(z_n)|\Omega \rangle$
as
\begin{align}
\langle\!\langle i_1,i_2,\dots,i_\ell 
|Y_\ell^{(\ell)}(z_1) \cdots Y_\ell^{(\ell)}(z_n)|\Omega \rangle
=\sum_{\{m_k^{(j)} \}} Z_{\{m_k^{(j)} \}},
\end{align}
where
the sum is over $m_k^{(j)}$, $j=1,\dots,\ell$, $k=1,\dots,i_j$
satisfying
$1 \le m_{i_j}^{(j)}<m_{i_j-1}^{(j)}<\cdots<m_1^{(j)} \le n$.
Each summand $Z_{\{m_k^{(j)} \}}$ is the sum
of weights corresponding to configurations which have the same bosonic configuration.
There can be several in principle, but we observe that
there is only one configuration if the bosonic configuration is fixed
and $m_{r}^{(s)} \neq m_{k}^{(j)}$ for all $(r,s) \neq (k,j)$,
and there is no configuration if $m_r^{(s)}=m_k^{(j)}$ for some $(r,s)=(k,j)$.

To see this, consider the case
$m_{r}^{(s)} \neq m_{k}^{(j)}$ for all $(r,s) \neq (k,j)$ first.
In this case, there is only one creation or no creation operator
in each column.
The $m_{k}^{(j)}$-th column has exactly one creation operator in the 
$j$-th bosonic Fock space which is graphically
depicted as  the left panel of Figure \ref{columnfreezingandweights}.
One notes that due to the ice rule,
the configurations from the 1st to the $(j-1)$-th rows are all fixed to be
$\mathcal{L}_{01}^{01}$,
and those from the $(j+1)$-th to the last rows are all fixed to be
$\mathcal{L}_{00}^{00}$.
This means that there is a red line passing through the vertical line from the $j$-th row 
as depicted in the right panel.
The $L$-operator on the $j$-th row gives the weight $z_{m_k^{(j)}}^{(j)}$.
The $L$-operators from the $1$-st to the $(j-1)$-th rows produce weighted projection operators.
For the $p$-th row $p=1,\dots,j-1$, the number of particles in the bosonic Fock space sandwiching the $L$-operator is $s(j,k,p)$,
hence the $L$-operator gives the factor $q^{s(j,k,p)}$.
From the $m_k^{(j)}$-th column we have the weight $z_{m_k^{(j)}}^{(j)} \prod_{p=1}^{j-1} q^{s(j,k,p)}$ in total.
For a column which a creation operator does not appear, 
we can also make a similar observation and note that all fermionic edges must be colored by blue,
and the weight for a column with no creation operator is always 1.
After all, we observe that the fermionic edges are all colored in a unique way
once the bosonic configuration is fixed and $m_{r}^{(s)} \neq m_{k}^{(j)}$ for all $(r,s) \neq (k,j)$
(Figure \ref{beforeandaftercolordrawing}),
and also conclude that $\displaystyle
Z_{\{m_k^{(j)} \}}= \prod_{j=1}^\ell \prod_{k=1}^{i_j} 
\Bigg(
z_{m_k^{(j)}}^{(j)} \prod_{p=1}^{j-1} q^{s(j,k,p)}
\Bigg)
$.

When the bosonic configuration is fixed and $m_{r}^{(s)} = m_{k}^{(j)}$ for some $(r,s) \neq (k,j)$,
we have
$Z_{\{m_k^{(j)} \}}=0$ since there is no nontrivial configuration in this case.
The condition $m_{r}^{(s)} = m_{k}^{(j)}$ for some $(r,s) \neq (k,j)$
means that there are at least two creation operators in the $m_{k}^{(j)}$-th column.
As observed in the former case, we note that there is a red line passing upward through the vertical line,
starting from the row where the creation operator is produced.
This implies that if there are two creation operators in the same column, then a configuration with three edges colored with red
and one edge colored with blue appears, but this is a forbidden configuration since this violates the so-called ice-rule.

\end{proof}

\begin{figure}[htbp]
\centering
    \includegraphics[width=0.7\textwidth]{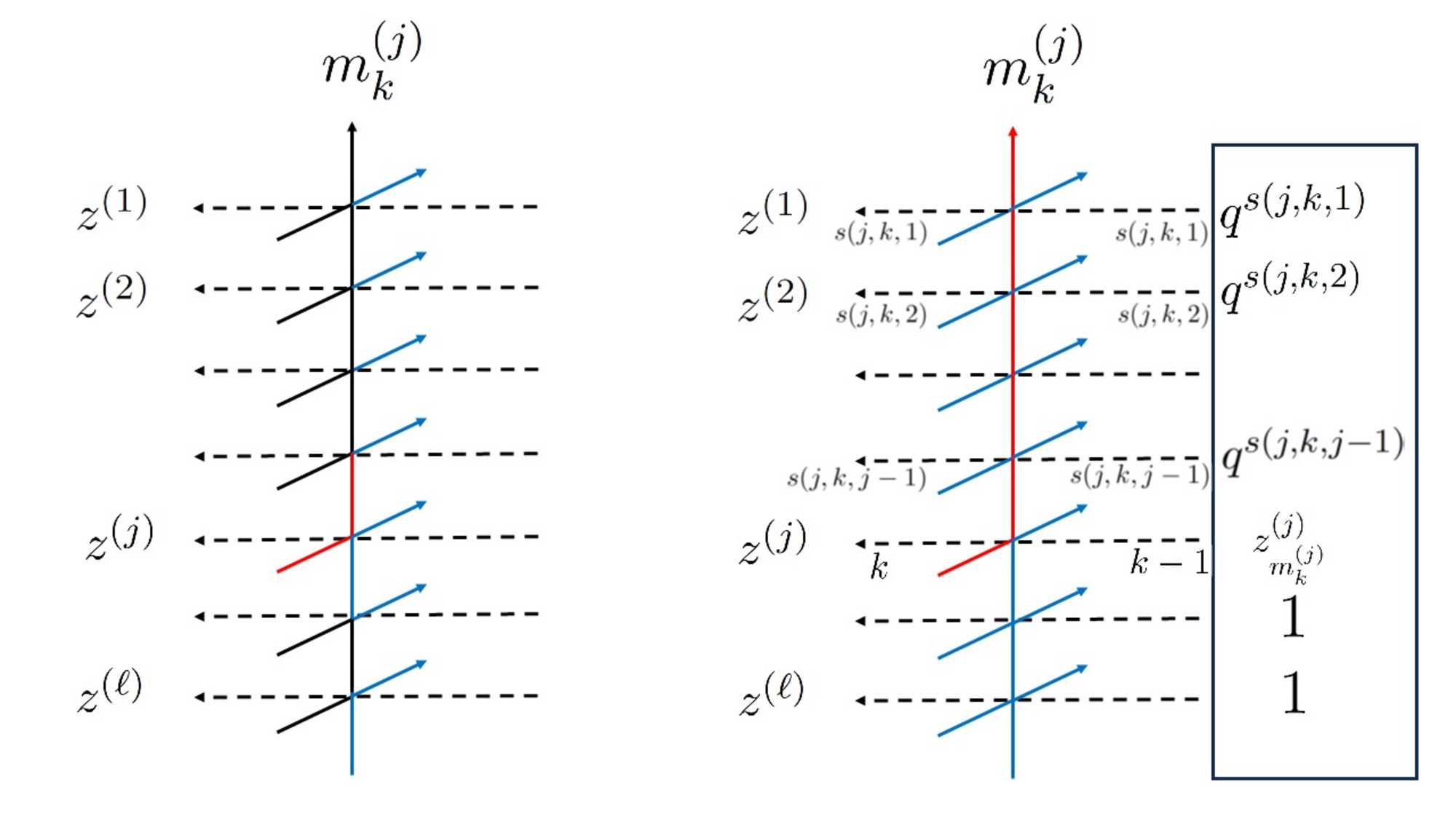}
\caption{The freezing of colors in the $m_k^{(j)}$th column with one creation operator.
Starting from the left panel, due to the ice-rule, we note that
from the $j$-th row, a line passes upward through the column,
and all the remaining uncolored edges are colored with blue as depicted
in the right panel. The weights coming from the elements of the $L$-operators
are presented in the framed box.
}
\label{columnfreezingandweights}
\end{figure}

\begin{figure}[htbp]
\centering
    \includegraphics[width=0.8\textwidth]{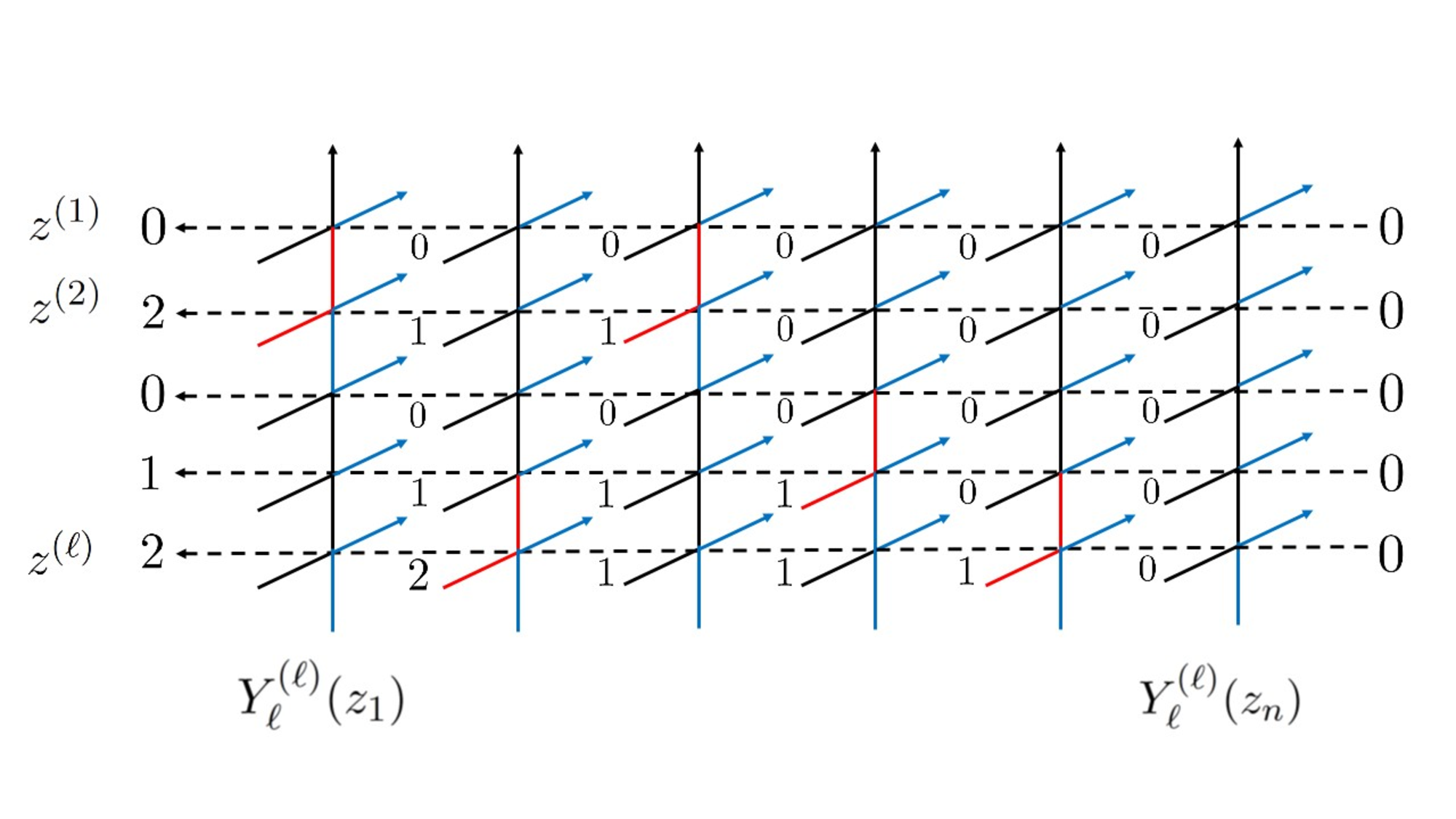}
 \includegraphics[width=0.8\textwidth]{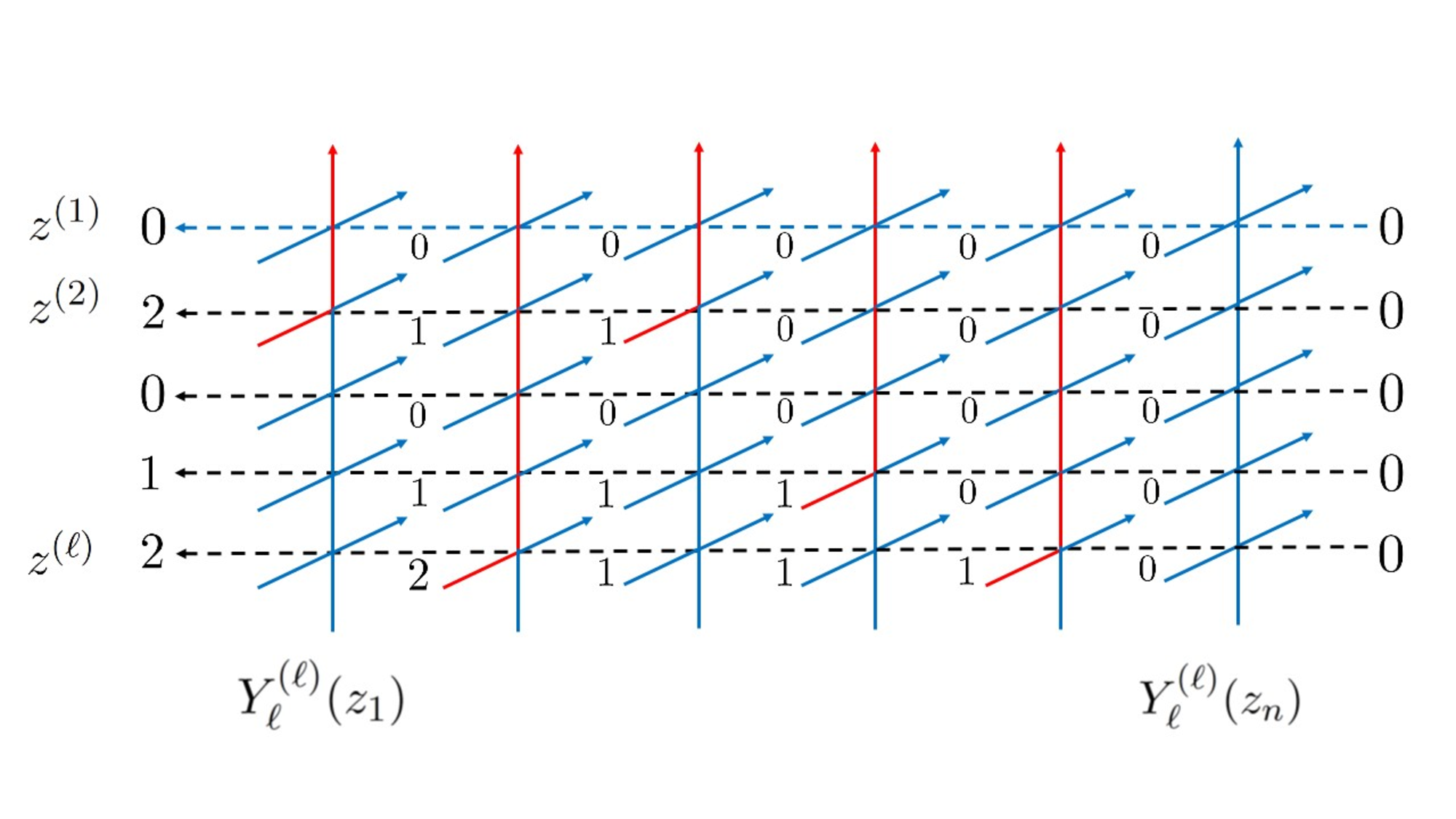}
\caption{A configuration which the states in the bosonic Fock spaces are fixed,
and the fermionic edges are not fully colored yet (top panel).
The edges in each column are colored in a unique way and the bottom panel is the figure after
all edges are colored.}
\label{beforeandaftercolordrawing}
\end{figure}

Special cases of
\eqref{higherrankqloop} restricting to $i_1,i_2,\dots,i_\ell \in \{0,1 \}$ gives a $q$-deformation of the loop elementary symmetric functions.

\begin{corollary}
For $f=i_1, i_2, \cdots, i_\ell$ a sequence of 0s and 1s such that the number of
1s is $m$ and the number of 0s is $\ell-m$,
let $\langle f|=\langle\!\langle i_1, i_2, \cdots, i_\ell |$.
We define $k_j$ $(j=1,\dots,m)$ to be integers such that 
$i_{k_1}=i_{k_2}=\cdots=i_{k_m}=1$ and $1 \le k_1<k_2<\cdots<k_m \le \ell$.
We have
\begin{align}
\langle\!\langle
&f |Y_\ell^{(\ell)}(z_1) \cdots Y_\ell^{(\ell)}(z_n)|\Omega \rangle \nn \\
=&\sum_{\sigma \in S_m}
\sum_{1 \le j_1 < j_2 < \cdots < j_m \le n}
q^{\mathrm{Inv}(\sigma(k_1),\sigma(k_2),\dots,\sigma(k_m))}
z_{j_1}^{(\sigma(k_1))} z_{j_2}^{(\sigma(k_2))} \cdots z_{j_m}^{(\sigma(k_m))}.
\label{qdefloop}
\end{align}
Here $S_m$ is the symmetric group of order $m$
with elements given by permutations of \\
 $(k_1, k_2, \cdots, k_m)$,
and $\mathrm{Inv}(\sigma(k_1),\sigma(k_2),\dots,\sigma(k_m))
=|\{ (k_i,k_j) \ | \ \sigma(k_i) > \sigma(k_j) \}|$.
For example,
\begin{align}
&\langle\!\langle 10010|
Y_5^{(5)}(z_1)Y_5^{(5)}(z_2)Y_5^{(5)}(z_3)|\Omega \rangle \nn \\
=&z_1^{(1)}z_2^{(4)}+qz_1^{(4)}z_2^{(1)}
+z_1^{(1)}z_3^{(4)}+qz_1^{(4)}z_3^{(1)}
+z_2^{(1)}z_3^{(4)}+qz_2^{(4)}z_3^{(1)}.
\end{align}
\end{corollary}

Further specializing to $z_k^{(j)}=z_k$ for all $j,k$ give the elementary symmetric functions.

\begin{corollary}
For $f=i_1, i_2, \cdots, i_\ell$ a sequence of 0s and 1s such that the number of
1s is $m$ and the number of 0s is $\ell-m$,
let $\langle\!\langle f|=\langle\!\langle i_1, i_2, \cdots, i_\ell |$.
We have
\begin{align}
\langle\!\langle
f |Y_\ell^{(\ell)}(z_1) \cdots Y_\ell^{(\ell)}(z_n)|\Omega \rangle
=[m]_q! e_m(z_1,\dots,z_n).
\end{align}
Here, recall $[m]_q=(1-q^m)/(1-q)$ and $[m]_q!=\prod_{j=1}^m [j]_q$.
For example,
\begin{align}
\langle\!\langle
1^m 0^{\ell-m} |Y_\ell^{(\ell)}(z_1) \cdots Y_\ell^{(\ell)}(z_n)|\Omega \rangle
=[m]_q! e_m(z_1,\dots,z_n).
\end{align}
\end{corollary}

The right hand side of
\eqref{qdefloop} can be regarded as a $q$-weighted sum of the loop elementary symmetric functions.
The $q=0$ and $i_1=i_2=\cdots=i_\ell=1$ case of \eqref{qdefloop}
is
\begin{align}
\langle\!\langle 1^\ell |Y_\ell^{(\ell)}(z_1) \cdots Y_\ell^{(\ell)}(z_n)|\Omega \rangle
=e_\ell^{(1)}(z_1,\dots,z_n), \label{specialcaseforclassicalintepretation}
\end{align}
where $e_k^{(a)}(z_1,\dots,z_n)$ is the loop elementary symmetric functions
\begin{align}
e_k^{(a)}(z_1,\dots,z_n)= \sum_{1 \le i_1 < i_2 < \dots < i_k \le n}
z_{i_1}^{(a)} z_{i_2}^{(a+1)} \cdots z_{i_k}^{(a+k-1)}.
\end{align}

\eqref{specialcaseforclassicalintepretation}
can be transformed to the classical lattice path interpretation of the loop elementary symmetric functions
as follows.
First, we draw all the lattice points in a rectangular grid with $(\ell+1)$ columns and 
$(n+1)$ rows as depicted in the top panel of
Figure \ref{figureforclassicalinterpretation}.
Consider lattice paths on this grid that start from the southeast corner and move to the northwest corner using only west and north-west steps.
At each lattice point in the $(j+1)$-th row counted from top, 
if the particle number in the bosonic Fock space of the $j$-th row
increases, we choose a northwest step which passes through the central point
of the $L$-operator; otherwise, we choose a west step.
When the $L$-operator being traversed is located at the $j$-th row and $k$-th column, we assign the weight 
$z_k^{(j)}$ to the northwest step, and assign the weight 1 to every west step.
The weight of a single path is defined as the product of the weights assigned to its steps.
An example can be seen in the bottom panel of Figure \ref{figureforclassicalinterpretation}.
Note that when $q=0$, the particle number increases faster in the bosonic Fock space of the lower row
and we can always transform every configuration to a lattice path which starts from the
southeast corner and ends at the northwest corner.
Partition functions are the sum of the weights from all possible configurations,
and all possible configurations correspond to all admissible lattice paths.
The sum of the weights over all such paths gives a classical interpretation of the loop elementary symmetric functions
$e_\ell^{(1)}(z_1,\dots,z_n)$,
which corresponds to the matrix elements of multiplication of the Lax operators of the discrete Toda lattice
\cite{Yamada,Lam,LP}.

\begin{figure}[htbp]
\centering
    \includegraphics[width=0.8\textwidth]{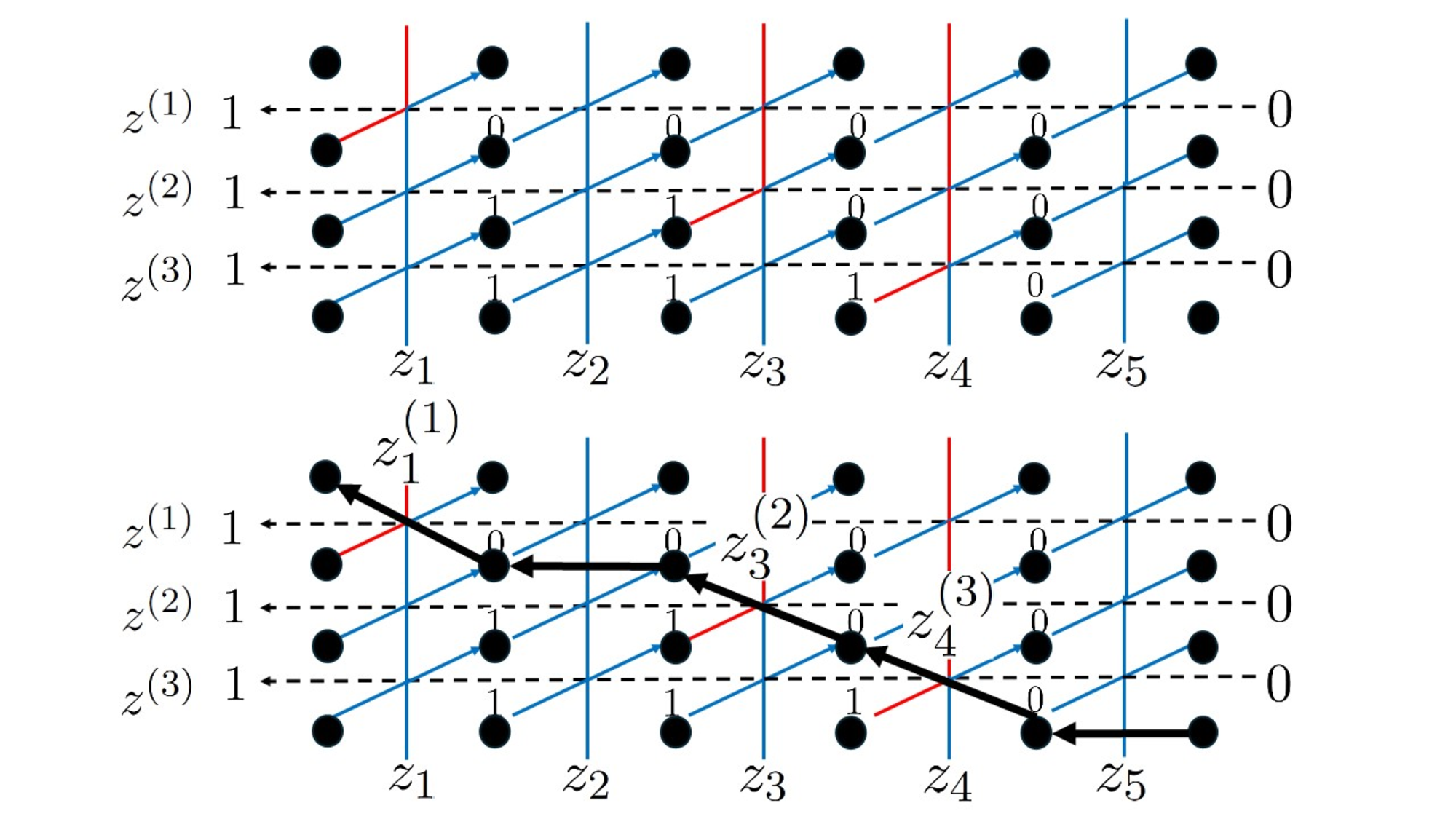}
\caption{
Transforming the partition function description to the classical lattice path description of the loop elementary symmetric functions.
The figure corresponds to the case $\ell=3,n=4$.
The lattice path drawn in the bottom part gives the weight $z_1^{(1)}z_3^{(2)}z_4^{(3)}$.
Taking all possible configurations of partition functions is transformed to taking all possible lattice paths into account,
and the sum of the weights of the lattice paths give $e_3^{(1)}(z_1,\dots,z_5)$.
}
\label{figureforclassicalinterpretation}
\end{figure}

\section*{
Acknowledgements
}

This work was partially supported by Grant-in-Aid for Scientific
Research (C) 21K03180, 23K03056, 26K06839.


\end{document}